\documentclass[acmsmall]{acmart}

\AtBeginDocument{%
  \providecommand\BibTeX{{%
    \normalfont B\kern-0.5em{\scshape i\kern-0.25em b}\kern-0.8em\TeX}}}

\setcopyright{rightsretained}
\copyrightyear{20xx}
\acmYear{20yy}

\usepackage{hyperref}
\usepackage{prettyref}

\newcommand{\rref}[2][]{\prettyref{#2}}
\newrefformat{sec}{Section\,\ref{#1}}
\newrefformat{app}{Appendix\,\ref{#1}}
\newrefformat{def}{Def.\,\ref{#1}}
\newrefformat{thm}{Theorem\,\ref{#1}}
\newrefformat{prop}{Proposition\,\ref{#1}}
\newrefformat{lem}{Lemma\,\ref{#1}}
\newrefformat{cor}{Corollary\,\ref{#1}}
\newrefformat{rem}{Remark\,\ref{#1}}
\newrefformat{ex}{Example\,\ref{#1}}
\newrefformat{tab}{Table\,\ref{#1}}
\newrefformat{fig}{Fig.\,\ref{#1}}
\newrefformat{case}{case\,\ref{#1}}
\newrefformat{foot}{Footnote\,\ref{#1}}

\usepackage{amsthm}

\usepackage{amsmath}
\usepackage{amssymb}
\usepackage{natbib}
\usepackage{mathtools}
\usepackage{tensor}
\usepackage{bm}
\usepackage{tikz}
\usepackage{tikz-cd}
\usepackage{graphicx}
\usepackage{stmaryrd}

\definecolor{semblue}{rgb}{0,0,0.7}
\definecolor{vgreen}{rgb}{.1,.5,0}%
\definecolor{vdarkgreen}{rgb}{.06,.3,0}%
\definecolor{vred}{rgb}{.7,0,0}%
\definecolor{vblue}{rgb}{.1,.15,.62}%
\definecolor{vgray}{rgb}{.35,.35,.35}
\definecolor{darkishgray}{rgb}{.35,.35,.35}
\definecolor{vvblue}{rgb}{.14,.21,.868}%

\usepackage{xcolor}
\usepackage{float}
\usepackage[prefixflatinterpret,prefixflatbracketmodalinterpret,prefixflatsetinterpret,sidenotecalculus,longseqcontext]{logic}
\usepackage[prefixflatinterpret,prefixflatbracketmodalinterpret,differentialdL,simplenames]{dL}

\graphicspath{ {./Images/} }

\usepackage{xspace}

\newtheorem{theorem}{Theorem}[section]

\newtheorem{proposition}[theorem]{Proposition}
\newtheorem{corollary}[theorem]{Corollary}
\newtheorem{lemma}[theorem]{Lemma}
\newtheorem{definition}[theorem]{Definition}

\theoremstyle{definition}
\newtheorem{example}[theorem]{Example}
\newtheorem{remark}[theorem]{Remark}

\newcommand{\R}{\mathbb{R}}

\newcommand{\N}{\mathbb{N}}
\newcommand{\Q}{\mathbb{Q}}
\newcommand{\V}{\mathbb{V}}

\newcommand{\IQ}{\mathbb{I}\mathbb{Q}}
\newcommand{\A}{\mathcal{A}}
\newcommand{\B}{\mathcal{B}}

\renewcommand{\S}{\mathbb{S}}

\renewcommand{\phi}{\varphi}

\newcommand{\set}{\pmb{\text{Set}}}

\newcommand{\eps}{\varepsilon}
\newcommand{\eval}[1]{\left\llbracket #1 \right\rrbracket}
\newcommand{\norm}[1]{\left\lVert#1\right\rVert}
\newcommand{\bebecomes}{\mathrel{::=}}
\newcommand{\alternative}{~|~}

\newcommand{\folr}{\text{FOL}_{\R}}

\newcommand{\leftrule}{L}
\newcommand{\rightrule}{R}

\newsavebox{\Rval}%
\sbox{\Rval}{$\scriptstyle\mathbb{R}$}
\newsavebox{\backiterateb}%
\sbox{\backiterateb}{$\scriptstyle\overleftarrow{\dibox{{}^*}}$}

\newcommand{\fvarA}{\phi}
\newcommand{\fvarB}{\psi}

\newcommand{\rfvar}{P}

\newcommand{\D}[1]{#1'}
\begin{document}

\renewcommand{\linferPremissSeparation}{\hspace{0.8cm}}

\title[Axiomatization of Compact Initial Value Problems: Open Properties]{Axiomatization of Compact Initial Value Problems:
\\ Open Properties}
\author{Andr\'e Platzer}
\author{Long Qian}
\email{platzer@kit.edu}
\email{longq@andrew.cmu.edu}
\orcid{0000-0001-7238-5710}
\orcid{0000-0003-1567-3948}
\affiliation{%
  \institution{Karlsruhe Institute of Technology}
  \city{Karlsruhe}
  \country{Germany}
  }
\affiliation{%
\institution{Carnegie Mellon University}
\city{Pittsburgh}
\country{USA}
}

\begin{abstract}
    This article proves the completeness of an axiomatization for \emph{initial value problems (IVPs)} with compact initial conditions and compact time horizons for bounded open safety, open liveness and existence properties.
    Completeness systematically reduces the proofs of these properties to a complete axiomatization for differential equation invariants.
    This result unifies symbolic logic and numerical analysis by a computable procedure that generates symbolic proofs with differential invariants for rigorous error bounds of numerical solutions to polynomial initial value problems.
    The procedure is modular and works for all polynomial IVPs with rational coefficients and initial conditions and symbolic parameters constrained to compact sets. Furthermore, this paper discusses generalizations to IVPs with initial conditions/symbolic parameters that are not necessarily constrained to compact sets, achieved through the derivation of fully symbolic axioms/proof-rules based on the axiomatization. 
    
\end{abstract}



\maketitle

\section{Introduction}

Differential equations and their analysis play a fundamental role in cyber-physical systems (CPS) correctness \cite{Alur_2015, DBLP:books/sp/Platzer18}.
Classically, the descriptive power of differential equations exceeds the analytic power of differential equations \cite{DBLP:journals/jacm/PlatzerT20}, since solutions of differential equations are usually significantly more complicated, not computable in closed form or less analyzable than the differential equations themselves.
That is why Henri Poincar\'e in 1881 called for the qualitative theory of differential equations \cite{Poincare81}, i.e., the study of differential equations directly via their differential equations rather than indirectly via their solutions.
The logical foundations of the qualitative theory of differential equation invariants have been discovered in a complete axiomatization of differential equation invariants \cite{DBLP:journals/jacm/PlatzerT20}.
In that axiomatization, every true (semialgebraic) invariant of a (polynomial) differential equation system can be proved effectively in differential dynamic logic \dL \cite{DBLP:conf/lics/Platzer12b,DBLP:journals/jar/Platzer17}, and every false invariant can be disproved, thereby leading to a purely logic-based proof-producing decision procedure.
But in CPS applications, even just finding invariants is challenging.
A CPS starts at an initial state within an initial region and follows a differential equation, where the question is whether it then always stays safe, which may still be far from an invariance question if the initial and safe region are very different.

This article, thus, studies the logical foundations of \emph{(compact) Initial Value Problems} (\emph{IVPs}).
In a (compact) IVP a (polynomial) differential equation on a compact time interval (with rational endpoints) starts from some initial value in a compact semialgebraic set.
The (semi)algebraic shape of those syntactic expressions ensures that the required concepts are definable in first-order logic of real arithmetic ($\folr$).
IVPs are one of the most fundamental problems studied in numerical analysis \cite{Kress_1998}.
Unlike in numerical algorithms for classical IVPs \cite{numerical_ivp}, however, the initial state is not given numerically as a single concrete vector of numbers such as $(0,4.2,-6)$, because those are typically not known when analyzing \emph{all} possible behavior of a CPS.
Instead, \emph{compact IVPs} generalize classical IVPs by supporting a compact initial region from which the symbolic initial state is selected nondeterministically.

This article proves completeness of \dL's axiomatization \cite{DBLP:journals/jar/Platzer17,DBLP:journals/jacm/PlatzerT20, DBLP:journals/fac/TanP21} for bounded open safety, open liveness and existence properties of compact IVPs such that every true such property can be proved.
Moreover these completeness theorems are effective, i.e., a direct computable procedure produces the \dL proofs based on \dL's effective axiomatization of differential equation invariants \cite{DBLP:journals/jacm/PlatzerT20}.
In order to achieve completeness and thereby complete Henri Poincar\'e's qualitative theory of differential equations for these properties of compact IVPs, this article will do something superficially frivolous: the completeness proofs will use solutions of IVPs, but ultimately of symbolic IVPs and only to guide the proofs of the required invariance properties of the IVPs. Besides, these solutions used for the guidance of the proofs will be approximate solutions only, not true solutions.
And, indeed, Henri Poincar\'e was still correct that both the true solution and their approximations are more complicated than the IVP, and that the indirect symbolic invariance proofs that this article's procedure constructs are both simpler and the key to the complete theory of IVPs.
In fact, one of the hard parts will be the need to prove that sufficient control can be exerted over the accumulating approximation errors to provide rigorous symbolic proofs with sufficiently small errors to justify every true bounded open safety, open liveness and existence property of a compact IVP.

While this article and its results are proof-theoretical in nature, they can also be viewed through a practically motivated angle. The problem of reachability analysis for ODEs and hybrid dynamical systems over a compact time horizon is an important area of study in the safety verification of CPS \cite{Alur_2015, DBLP:books/sp/Platzer18}, particularly for bounded model checking \cite{DBLP:conf/cdc/GiorgettiPB05}. Consequently, practical tools \cite{DBLP:conf/cav/ChenAS13, DBLP:journals/jar/Immler18, DBLP:conf/hybrid/BresolinCGSVG20} have been developed to tackle this problem, essentially computing interval enclosures of compact IVPs. Such procedures are all inherently based on numerical approximation techniques in contrast to the deductive, symbolic proof approach offered by $\dL$. 

In safety-critical applications however, the trustworthiness of such numerical approaches is challenging to justify rigorously. Even when the numerical approximations computed by such numerical procedures are mathematically rigorous (which is itself difficult to fully justify in a trustworthy fashion), subtle errors can still arise in the implementation of such algorithms. Even the verification of floating-point arithmetic has proven to be intricate and non-trivial \cite{DBLP:conf/arith/BoldoF07}. 

Deductive approaches based on symbolic proofs in contrast are much more trustworthy. Properties of dynamical systems are proved by applying a sequence of sound proof rules based on a small set of sound axioms \cite{Fulton_Mitsch_Quesel_Volp_Platzer_2015}. Certifying the correctness of such proofs only relies upon a small trusted core of the proof checker \cite{DBLP:conf/cpp/BohrerRVVP17}. Such deductive approaches are more symbolic in nature, seemingly orthogonal to numerical approximations and less capable in verifying inherently numerical properties of compact IVPs. On the contrary, this article crucially shows that this is not the case, numerical approximations and symbolic logic can be harmoniously integrated to obtain the best of both worlds - symbolically proving properties of dynamical systems using numerical approximations. Thus, the desired properties can be proven deductively in a trustworthy manner accompanied by a certifying proof, while not losing the computational capabilities of using numerical approximations. This article thereby unifies computation and deduction for compact IVPs. 

All in all, this article explores the \emph{proof theory} of compact IVPs, providing \emph{complete} reasoning principles for bounded open safety, open liveness and existence properties for compact IVPs by drawing upon both numerical algorithms and deductive verification techniques. 

The following presents an overview of the main results established in this article, first defining the basic notions needed. Let
\begin{align*}
    &x' = f(x)\\
    &x(0) \in \eval{C} \subset \R^n
\end{align*}
be an arbitrary IVP on a compact time horizon $[t_0, T]$ with rational endpoints, each component of $f(x) = (f_1(x), \cdots, f_n(x))$ is a rational polynomial in the ($n$-dimensional vectorial) variable $x$ and $\eval{C}$ is a non-empty compact subset of $\R^n$ defined via the $\folr$ formula $C(x)$ (i.e. $\eval{C} = \{x \in \R^n \mid \R \models C(x)\}$). The main contributions of the article concern the completeness of fragments of $\dL$, a brief explanation of the necessary fragment is given here and a more complete account of $\dL$ is provided in \rref{sec: prelim_dL}.

The fragments this article is concerned with comprises of $\dL$ formulas \cite{DBLP:journals/jar/Platzer08} of the following form, where $P, Q \in \folr$ and $t_0, T \in \Q$. 
\begin{align*}
    &P \land t = t_0 \rightarrow \dbox{x' = f(x), t' = 1 \& t \leq T}{Q}\\
    &P \land t  = t_0 \rightarrow \ddiamond{x' = f(x), t' = 1 \& t \leq T}{Q}
\end{align*}
Such formulas express safety/liveness properties of the flow induced by the differential equation $x' = f(x)$. If $\phi(x, t)$ denotes the corresponding flow function\footnote{The flow is assumed to be well-defined here for brevity, the complication of finite-time blow up is treated in \rref{sec: proving constraint}.} starting at $t = t_0$ (i.e. $\phi(x, t_0) = x$), then the formulas above correspond exactly to the following formulas that quantify over the flow:
\begin{align*}
    & \mforall x \big(P(x) \rightarrow \boldsymbol{\forall} t \in [t_0, T]\text{ } Q(\phi(x, t))\big)\\
    &\mforall x \big(P(x) \rightarrow \boldsymbol{\exists} t \in [t_0, T]\text{ } Q(\phi(x, t))\big)
\end{align*}
I.e. the first formula expresses the safety property that every trajectory starting in the set characterized by $P$ evolving on the time horizon $[t_0, T]$ remains in the safety region characterized by $Q$. Dually, the second formula expresses the liveness property that every trajectory starting in $P$ can reach the target set $Q$ by evolving on $[t_0, T]$. The following main results are established in this article:

\begin{enumerate}
    \item \textbf{Completeness for convergence}: Suppose the (compact) IVP admits a solution/flow $\phi(x, t)$ on the domain $\eval{C} \times [t_0, T]$ (i.e. $\phi(x, t_0) = x, \phi'(x, t) = f(\phi(x, t))$ for all $(x, t) \in \eval{C} \times [t_0, T]$), let $(p_n)_n \in C^0(\eval{C} \times [t_0, T], \R^n)$ be any sequence of definable approximants\footnote{Definable functions in $\folr$, which is exactly when each $p_n$ is a semialgebraic function over $\Q$. In particular, this includes polynomials in $\Q[x, t]$, see \rref{def: definable function} for details.} that converges uniformly to $\phi(x, t)$ in the space $C^{0}(\eval{C} \times [t_0, T], \R^n)$. For all $\eps \in \Q^+$, one can computably find some $k \in \N$ such that
    \[C(x) \land x = x_0 \land t = t_0 \rightarrow \dbox{x' = f(x), t' = 1 \& t \leq T}{\norm{x - p_k(x_0, t)}^2 \leq \eps^2}\]
    is a valid formula of $\dL$ where $x, x_0, t$ are \emph{symbolic variables}. In fact, we will show that this can be syntactically derived in $\dL$'s axiomatization. This formula is equivalent to the following sentence involving the true flow $\phi$ of the IVP as a function symbol 
    \[\boldsymbol{\forall} x_0 \in \eval{C}~\mforall t \in [t_0, T] \left(\norm{\phi(x_0, t) - p_k(x_0, t)}^2 \leq \eps^2\right)\]
    I.e. $p_k$ is an approximant of \emph{uniform error} at most $\eps$ for the true flow $\phi(x, t)$ on $\eval{C} \times [t_0, T]$. In other words, this formula along with its syntactic derivation provides a \emph{proof} of the accuracy of the approximant $p_k$. This establishes that $\dL$ is \emph{complete for convergence}. I.e. if a sequence $(p_n)_n \xrightarrow{n \to \infty} \phi$ converges in $C^0(\eval{C} \times [t_0, T], \R^n)$, then this convergence is provable in $\dL$, succinctly denoted as the following:

    \[\vDash (p_n)_n \xrightarrow{n \to \infty} \phi \qquad \implies \qquad \vdash (p_n)_n \xrightarrow{n \to \infty} \phi\]

    In particular, the definable approximants can be taken to be outputs of numerical solvers applied on the IVP, obtained via standard interpolation procedures (e.g. polynomials, splines). The above result shows that $\dL$ is capable of \emph{symbolically proving} the accuracy of numerical solvers. In contrast to ODE solvers that deal with rigorous numerics using one specific formally verified algorithm \cite{DBLP:journals/cnsns/KapelaMWZ21, DBLP:journals/jar/Immler18}, this result rather gives a procedure that decides if \emph{any} such numerical algorithm is correct from its \emph{outputs}, along with supporting formal proofs. Crucially this procedure does not rely on any particular ODE solver to be correct, it rather takes outputs of ODE solvers as inputs (represented by the sequence of approximants) and returns a certificate of correctness for the accuracy of the approximants in the form of a proof in $\dL$. 
    
    \item \textbf{Completeness of (compact) IVPs}: This article proves completeness of $\dL$'s axiomatization for bounded open safety, open liveness and existence properties of compact IVPs:
        \begin{itemize}
            \item \textbf{Completeness for bounded open safety}: Let $O(x)$ be a $\folr$ formula that characterizes a bounded open subset of $\R^n$. $\dL$ is complete for formulas of the form
            \[C(x) \land t = t_0 \rightarrow \dbox{x' = f(x), t' = 1 \& t \leq T}{O(x)}\]
            I.e. if all flows of the IVP starting anywhere in $\eval{C}$ always remains within the set of safe states characterized by $O(x)$ on the time horizon $[t_0, T]$, then this is provable in $\dL$. 

            \item \textbf{Completeness for open liveness}:
            Let $O(x)$ be a $\folr$ formula that characterizes an open subset of $\R^n$ (not necessarily bounded as stronger assumptions are placed on the flow instead), and suppose that the true flow $\phi : \eval{C} \times [t_0, T] \to \R^n$ is well-defined (i.e. does not exhibit finite-time blow up on $[t_0, T]$)\footnote{Such an assumption also suffices for arbitrary open safety properties.}. Then $\dL$ is complete for formulas of the form
            \[C(x) \land t = t_0 \rightarrow \ddiamond{x' = f(x), t' = 1 \& t \leq T}{O(x)}\]
            I.e. if a target state characterized by $O(x)$ is reachable from starting anywhere in $\eval{C}$ in the time horizon $[t_0, T]$ by following the IVP, then this is provable in $\dL$. 

            \item \textbf{Completeness for existence}: $\dL$ is complete for formulas of the form
            \[C(x) \land t = t_0 \rightarrow \ddiamond{x' = f(x), t' = 1}{t \geq T}\]
            I.e. if the solution exists for time at least $t \geq T$ for all initial conditions from $\eval{C}$, then this is provable in $\dL$.
        \end{itemize}
        By considering the case where $C(x) \equiv x = x_0$ defines a singleton, corresponding completeness results for IVPs with fixed initial conditions are obtained as a special case.  
    \item \textbf{Axioms/proof-rules for symbolic IVPs}: In proving completeness of existence for IVPs, we derive fundamental \emph{symbolic} axioms/proof-rules (\rref{thm:axioms for step liveness}) for deductive verification of symbolic IVPs on compact time horizons without placing constraints on the initial conditions. We prove symbolic derivations of the classical Picard-Lindel\"of theorem, the intermediate value theorem and the property that the solution to an IVP exists on some time horizon if and only if the solution has no finite time blow-up on that time horizon. Due to the fundamental nature of such axioms/proof-rules \cite{Walter_1998}, their derivations are of independent interest. 
\end{enumerate}

\section{Related Work}

The results presented in this article build upon the proof theory of dynamical systems using the framework of \emph{differential dynamic logic} ($\dL$) \cite{DBLP:conf/lics/Platzer12b,DBLP:conf/fm/SogokonJ15, DBLP:journals/jacm/PlatzerT20, DBLP:journals/fac/TanP21}.  This article establishes the first complete axiomatization for compact IVPs, showing that all true (bounded) open properties can be deduced completely from symbolic axioms/proof rules, in the spirit of Poincar\'e's qualitative theory of differential equations \cite{Poincare81}. Consequently, it is possible to deductively prove properties of compact IVPs with trustworthy symbolic logic whilst retaining the computational capabilities of numerical techniques. 

\noindent
{\bf{\emph{Computability of compact IVPs}}}: The computability of IVPs have been studied extensively \cite{DBLP:conf/cie/Graca018}. More recent works have shown interesting connections between computable ordinals and the solutions of discontinuous IVPs \cite{DBLP:conf/stacs/BournezG24}. In the context of IVP verification, such computability results can be viewed as the theoretical foundation of numerical techniques. Indeed, the statement that arbitrarily accurate numerical approximations can be computed for compact IVPs is a restatement of the result that solutions to compact IVPs are type-two computable (see \rref{sec: Computable analysis prelims} for details). However, the trustworthiness of such approaches are much more delicate and it is difficult to formalize requirements on the trustworthiness of numerical algorithms purely on the computability level. Such questions are more naturally expressed as \emph{provability} questions, which is exactly what this article addresses. 

In the same way that computability is the theoretical foundation for numerical techniques, provability is the theoretical foundation for symbolic deductive techniques. Such provability properties are generally more fine-grained and delicate compared to computability properties. As $\dL$ is computably axiomatized, any property it is complete for is trivially computably enumerable by searching through all possible proofs, while the converse implication is generally false. Numerical techniques are often viewed to be more scalable than deductive techniques for IVPs, but symbolic proofs enjoy a higher level of rigor and reliability. 

Nonetheless, the completeness results presented in this article precisely bridge this gap, showing that in the context of (open) properties of compact IVPs, provability and computability notions ``coincide'' - numerical approximations can be carried out entirely deductively in $\dL$ with symbolic proofs. There is no fundamental distinction between numerical and symbolic computations for compact IVPs. Properties that can be verified by numerical techniques with direct computations can also be verified deductively with logic, resulting in trustworthy proofs of such properties while enjoying the generality of numerical techniques. Furthermore, building upon works on the computability of IVPs \cite{DBLP:conf/cie/Graca018}, this article establishes a direct computable correspondence between valid (open) properties of compact IVPs and their proofs in $\dL$. 

\noindent
{\bf{\emph{Proof theory of compact IVPs}}}: The completeness results presented in this article applies to all (open) properties for compact IVPs, and does so in a computable fashion. In contrast to the results established in this article, earlier works either only prove relative completeness with some non-computable oracle \cite{DBLP:conf/lics/Platzer12b, DBLP:conf/lics/Platzer12a}, global completeness that cannot handle compact IVPs which are sensitive to their initial conditions \cite{DBLP:journals/jacm/PlatzerT20}, or does not achieve general completeness results \cite{DBLP:conf/fm/SogokonJ15, DBLP:journals/fac/TanP21}. To the best of our knowledge, this is the first result that establishes the provability of such properties of compact IVPs. 

In addition, this article also provides novel syntactic derivations of classical theorems in $\dL$ that are fundamental to the study of IVPs, allowing for the deductive verification of general symbolic IVPs beyond compact IVPs. In contrast with earlier works \cite{DBLP:journals/fac/TanP21, DBLP:conf/fm/SogokonJ15}, these axioms/proof rules focus on the case where the IVP is considered on a compact time horizon and crucially does not assume global existence of solutions. This situation is much more delicate as solutions of IVPs might exhibit finite time blow-ups. The derivations themselves are of independent interest, not only are the axioms/proof rules themselves fundamental in the study of IVPs, such derivations also improve upon earlier works (e.g. for \irref{ivt} \cite{Platzer_Tan_report}) where soundness was proven but no derivation was known. 

\noindent
{\bf{\emph{The Continuous Skolem Problem and limitations}}}: The Continuous Skolem Problem is a central problem in the theory of continuous dynamical systems \cite{DBLP:journals/tcs/BellDJB10}. Given an IVP $x' = f(x)$ with $x(0) = x_0 \in \Q^n$ and a vector $u \in \Q^n$, the Continuous Skolem Problem asks if the solution $x(t)$ reaches the hyperplane defined by $u$. I.e. if there exists some $t \geq 0$ such that $u^Tx(t) = 0$. The Bounded Continuous Skolem problem \cite[Open Problem~17]{DBLP:journals/tcs/BellDJB10} asks if such a $t$ exists in some pre-determined interval $[0, T]$ with $T \in \Q^+$. The decidability of both problems have been long-standing open problems, with partial progress being made in the case where the ODE is linear, i.e. $x' = Ax$ for $A \in \Q^{n \times n}$ \cite{DBLP:conf/icalp/ChonevOW16}. In the linear case, the Bounded Continuous Skolem problem was shown to be decidable assuming Schanuel's conjecture \cite[Theorem~7]{DBLP:conf/icalp/ChonevOW16}, a unifying conjecture in number theory which implies the decidability of the real exponential field \cite{real_exp_field_decidability}. Such problems have also been studied when $f(x)$ is allowed to be a polynomial \cite{DBLP:conf/cie/Hainry08}. In this setting, the decidability of the Bounded Continuous Skolem problem remains open. In the context of this article, such problems place inherent restrictions on possible generalizations of the new results presented here. The results established can be viewed as the form of ``(compact, (bounded) open)'', where the initial condition is required to come from a compact set and the post-condition is required to be (bounded) open. A natural generalization is to consider ``(compact, compact)'' where post-conditions are compact semialgebraic sets. However, such completeness results (if possible) are at least as hard as the Bounded Continuous Skolem problem for polynomial dynamical systems. This is because the Bounded Continuous Skolem problem is co-computably enumerable (co-c.e.): 
\[\boldsymbol{\exists} t \in [0, T]~u^Tx(t) = 0 \iff \min_{t \in [0, T]} \abs{u^T x(t)} = 0\]
and minimums of computable functions over compact sets are computable (\rref{thm: computable extreme value theorem}), therefore the second relation is co-c.e. At the same time, the reachability requirement can be naturally formulated in $\dL$ via:
\[x = x_0 \land t = 0 \rightarrow \ddiamond{x' = f(x), t' = 1 \& t \leq T}{u^Tx = 0}\]
Thus, if completeness results of the form "(compact, compact)" hold, then the Bounded Continuous Skolem problem would also be c.e. by searching through all possible proofs in $\dL$ (as $\dL$ is computably axiomatized \cite{DBLP:conf/lics/Platzer12a}), implying the decidability of the Bounded Continuous Skolem problem (independent of Schanuel's conjecture). Hence, generalizations of the results in this article by relaxing the topological constraints on the post-conditions are likely to be challenging and difficult. 

\noindent
{\bf{\emph{Reachability computation of dynamical systems}}}: The problem of computing interval enclosures for hybrid/continuous dynamical systems shows up frequently in the safety verification of CPS. Practical implementations \cite{DBLP:conf/cav/ChenAS13, DBLP:conf/cav/FrehseGDCRLRGDM11, DBLP:journals/cnsns/KapelaMWZ21, DBLP:journals/jar/Immler18, DBLP:conf/hybrid/BresolinCGSVG20} exist to carry out such computations, based on numerical approximations. The correctness of these will depend on both the correctness of the underlying mathematical theory of such approximations and the correctness of the implementation details, both of which are prone to errors. Attempts in improving the reliability of such procedures focus on the formal verification of such algorithms (e.g. \cite{DBLP:journals/cnsns/KapelaMWZ21, DBLP:journals/jar/Immler18} in the continuous case), where the numerical algorithm implemented is formally verified to be mathematically sound. These formal verifications are inherently dependent on the specific algorithm used, and modifications to the algorithm require corresponding complex modifications to the proof of correctness, in addition to the possibility of implementation errors. The lack of compositionality implies that it is non-trivial to combine different algorithms harmoniously. This is in particular highlighted by the fact that, to the best of our knowledge, there currently does not exist a formally verified algorithm computing the interval enclosure of hybrid dynamical systems. More fundamentally, such approaches are providing formal verifications of the \emph{algorithm} used to compute the approximations, inducing potential error when transforming from the abstract algorithm verified to the actual implementation executed. This is in contrast to logic-based deductive approaches where every verified property has a certifying \emph{syntactic proof} which can be independently checked. 

The results presented in this article provide a complementary possibility through $\dL$: Such verifications can \emph{all} be carried out deductively with sound axioms/proof rules, therefore the correctness of the approximations can be trusted as certified by their corresponding symbolic proofs. Numerical approximations can be computed deductively, ending up with compositional proofs in $\dL$ that can be used in symbolic proofs of safety of the overall hybrid dynamical system. In particular, the completeness results (e.g. \rref{thm: effective proofs without upper bounds}) are agnostic to how the numerical approximants were computed. Therefore potentially unreliable approximation algorithms can be used in computations as the computed approximants can always be symbolically proven to be accurate if they truly are accurate. In contrast with the formal verification of numerical algorithms, the deductive approach certifies the correctness of outputs with corresponding \emph{proofs} that can be checked with proof checkers such as \KeYmaeraX \cite{Fulton_Mitsch_Quesel_Volp_Platzer_2015, DBLP:conf/cpp/BohrerRVVP17}. The aim of the article is not to argue for the superiority of symbolic techniques over numerical ones, but rather that such approaches are in fact intimately related and it is possible to simultaneously achieve the strengths of both approaches at once as shown by the completeness results. 

\section{Preliminaries}
\label{sec: preliminaries}
We give a self-contained overview of the computable analysis and differential dynamic logic needed for the paper. More details on computable analysis, computability theory \cite{Weihrauch_2000, Soare_2016} and differential dynamic logic \cite{DBLP:journals/jar/Platzer08} can be found in the corresponding references. 

\subsection{Differential Dynamic Logic}
\label{sec: prelim_dL}
This section provides a brief review of differential dynamic logic ($\dL$) and its axiomatization, fixing some notation conventions along the way. This article focuses on the continuous fragment of $\dL$.

\subsubsection{Syntax}
Terms in $\dL$ are formed by the following grammar, where $\V$ denotes the set of all variables, $x \in \V$ is a variable and $c \in \Q$ is a rational constant. Equivalently, terms are polynomials over $\V$ with rational coefficients.
\[p, q\bebecomes x\alternative c\alternative p + q\alternative p\cdot q\]

Prior works \cite{DBLP:journals/jacm/PlatzerT20} make it possible to consider $\dL$ with an expanded language that includes familiar mathematical functions such as $\exp, \sin, \cos$. Such expansions will not be considered in this article due to the subtle concerns regarding computability of such expanded functions.

$\dL$ formulas have the following grammar, where ${\sim} \in \{=, \neq, \geq, >, \leq, <\}$ is a comparison relation and $\alpha$ is a system of differential equations ($\dL$ allows for $\alpha$ to be from the more general class of \emph{hybrid programs} \cite{DBLP:journals/jar/Platzer08}, which is not needed here)
\begin{align*}
    \phi, \psi &\bebecomes p {\sim} q \alternative \phi \land \psi \alternative \phi \lor \psi \alternative \neg \phi \alternative \forall x \phi \alternative \exists x \phi \alternative 
\ddiamond{\alpha}{\phi} \alternative \dbox{\alpha}{\phi}\\
\alpha &\bebecomes \cdots \alternative x' = f(x) \& Q
\end{align*}
In this paper, we will only be dealing with the case $\alpha \equiv x' = f(x) \& Q$, where $x' = f(x)$ represents an autonomous system of ODEs $x_1' = f_1(x), \cdots, x_n' = f_n(x)$ and $x = (x_1, \cdots, x_n)$ is understood to be vectorial. $Q$ here refers to some $\dL$ formula known as the \emph{domain constraint}. Intuitively, this restricts the region for which the ODE $x' = f(x)$ is allowed to evolve. In contrast with some of the earlier works \cite{DBLP:journals/fac/TanP21, DBLP:journals/jacm/PlatzerT20}, the domain constraint $Q$ is in general allowed to be any $\dL$ formula in this paper, resulting in ``rich-test'' $\dL$ \cite{DBLP:journals/jar/Platzer17, DBLP:journals/jar/Platzer08, DBLP:conf/lics/Platzer12b}. 

Lastly, we state some conventions that are used throughout this paper. For terms and formulas that appear in contexts involving ODEs $x' = f(x)$, it is sometimes needed to restrict the variables that they can refer to. When such cases arise, we will indicate such free variables by explicitly writing them as arguments. For example, $p()$ means that the term $p$ cannot refer to any bound variable of the ODE $x' = f(x)$. In contrast, $P(x)$ (or just $P$) indicates that all the variables may be referred to as free variables. Such variable dependencies can be made formal and rigorous through $\dL$'s uniform substitution calculus \cite{DBLP:journals/jar/Platzer17}.

\subsubsection{Semantics}
A state $\omega$ is a mapping $\omega : \V \to \R$ that assigns a value to every variable. We denote $\S$ as the set of all such states. For a term $p$, its semantics in state $\omega \in \S$ written as $\eval{p}$ is the real value obtained by evaluating the polynomial $p$ at the state $\omega$. For a $\dL$ formula $\phi$, its semantics $\eval{\phi}$ is defined to be the set of all states $\omega \in \S$ such that $\omega \models \phi$, i.e the formula $\phi$ is true in $\omega$. The semantics of first-logical connectives are defined as expected, e.g. $\eval{\phi \lor \psi} = \eval{\phi} \cup \eval{\psi}$. For ODEs $\alpha \equiv x' = f(x) \& Q$, the semantics for $\dbox{\alpha}{\phi}$ and $\ddiamond{\alpha}{\phi}$ are defined as follows. For the given ODE $\alpha$ and any state $\omega \in \S$, let $\Phi_\omega : [0, T) \to \S$ be the solution to $x' = f(x)$ extended maximally to the right with $0 < T \leq \infty$ and $\Phi_\omega(0) = \omega$. We then have:
\begin{align*}
    &\omega \in \eval{\dbox{\alpha}{\phi}} \text{ iff for all $0 \leq \tau < T$ such that $\Phi_\omega(\xi) \models Q$ for all $0 \leq \xi \leq \tau$, we have $\Phi_\omega(\tau) \models \phi$}\\
    &\omega \in \eval{\ddiamond{\alpha}{\phi}} \text{ iff there exists some $0 \leq \tau < T$ such that $\Phi_\omega(\xi) \models Q$ for all $0 \leq \xi \leq \tau$ and $\Phi_\omega(\tau) \models \phi$}
\end{align*}
Intuitively, the formula $\dbox{\alpha}{\phi}$ expresses a \emph{safety} property, that $\phi$ holds along all flows of the ODE $x' = f(x)$ that remain inside the domain constraint $Q$. Similarly, the formula $\ddiamond{\alpha}{\phi}$ expresses a \emph{liveness} property, that there is some flow along $x' = f(x)$ staying within $Q$ eventually reaching a state where $\phi$ is true. 

Finally, a formula $\phi$ is said to be valid if $\eval{\phi} = \S$, i.e. it is true in all states. For a formula $I$, we say it is a \emph{differential invariant} of the ODE $x' = f(x) \& Q$ if the formula $I \rightarrow [x' = f(x) \& Q]I$ is valid. One important fact that is used throughout this article is that $\dL$ is (effectively) complete for differential invariants in $\folr$ \cite{DBLP:journals/jacm/PlatzerT20}. In other words, if $I$ is a differential invariant, then one can effectively find a syntactic proof of $I \rightarrow [x' = f(x) \& Q]I$.

\subsubsection{Proof calculus}
The derivations in this paper are presented in a standard, classical sequent calculus with the usual rules for manipulating logical connectives and sequents. For a \emph{sequent} $\Gamma \vdash \phi$, its semantics is equivalent to the formula $\left(\bigwedge_{\psi \in \Gamma} \psi\right) \rightarrow \phi$, and the sequent is called valid if its corresponding formula is valid. For a sequent $\Gamma \vdash \phi$, formulas $\Gamma$ are called antecedents, and $\phi$ the succedent. We mark completed proof branches with $*$ in a sequent proof, and since $\R$ has a decidable theory via quantifier elimination \cite{Tarski_1948}, statements that follow from real arithmetic are proven with the rule \irref{qear}. An axiom (schema) is called \emph{sound} iff all of its instances are valid, and a proof rule is sound if the validity of all its premises entail the validity of its conclusion. Axioms and proof rules are \emph{derivable} if they can be proven from $\dL$ axioms and proof rules via the aforementioned sequent calculus. Derivable axioms are automatically sound due to the soundness of $\dL$'s axiomatization \cite{DBLP:journals/jar/Platzer08, DBLP:journals/jacm/PlatzerT20}.

This article uses a fragment of the base axiomatization of $\dL$ \cite{DBLP:conf/lics/Platzer12b} (focusing on the continuous case) along with an extended axiomatization developed in prior works used to handle ODE invariants and liveness properties \cite{DBLP:journals/jacm/PlatzerT20, DBLP:journals/fac/TanP21}. A complete list of the axioms used is provided in \rref{app: dL axiomatization}.

An important feature of the axiomatization used is that it is complete for all differential invariants \cite{DBLP:journals/jacm/PlatzerT20}. Since this will be used extensively throughout the paper, this fact is explicitly stated below.

\begin{theorem}[Completeness of Differential Invariants \cite{DBLP:journals/jacm/PlatzerT20}]
    \label{thm:diff inv completeness}
    $\dL$ is complete for differential invariants. For all $\folr$ formulas $P, Q$ and ODE $x' = f(x)$, construct the corresponding $\dL$ formula as follows
    \[P \rightarrow \dbox{x' = f(x) \& Q}{P}\]
    If the formula is valid, then one can effectively find a proof of it in $\dL$. We will make use of this result with the following derived proof rule:
    
    \begin{calculus}
    \cinferenceRule[dinv|dInv]{quantifier elimination real arithmetic}
        {\linferenceRule[sequent]
          {\lclose}
          {\lsequent[g]{}{P \rightarrow \dbox{x' = f(x) \& Q}{P}}}
            \qquad\qquad\qquad\qquad\qquad
        }
        {$\text{If}~ P \rightarrow \dbox{x' = f(x) \& Q}P ~\text{is valid}$}
    \end{calculus}
    
\end{theorem}

\rref{thm:diff inv completeness} will be utilized frequently to obtain syntactic proofs by first reducing the goals down to some differential invariant, and then proving the validity of this invariant semantically. This completeness is effective, so computability properties are preserved by appealing to \rref{thm:diff inv completeness}.

\subsection{Computability and Computable analysis}
\label{sec: Computable analysis prelims}
The completeness properties established in this article are \emph{effective}. Not only are valid formulas provable, there is a direct (computable) correspondence between the valid formulas and their proofs. I.e., there is a computable algorithm taking valid formulas as inputs and outputting corresponding proofs in $\dL$.\footnote{As $\dL$'s axiomatization is effective, completeness automatically implies such an algorithm by searching through all proofs. However, this paper establishes a direct correspondence rather than resorting to the brute-force search.} To achieve the desired completeness results effectively, it is necessary to utilize the computability-theoretic properties of IVPs, which are framed in the language of \emph{computable analysis}. The following provides the required background on computable analysis, under the standard framework of \emph{Type Two Theory of Effectivity} (TTE) \cite{Weihrauch_2000}.  

\begin{definition}[Name]
    Let $x \in \R$ be any real number, a $\it{name}$ for $x$ is a sequence of rationals $(q_i)_i \subseteq \Q$ such that
    \[\boldsymbol{\forall} i \in \N~(\abs{q_i - x} < 2^{-i})\]
    This definition naturally extends to $\R^n$ by requiring names to reside in $\Q^n$ and using the standard Euclidean norm. For $x \in \R^n$, we denote the set of all names of $x$ as $\Gamma(x)$.
\end{definition}

For a fixed real number $x \in \R^n$, one should think of its names as the ``descriptions'' of $x$. We then define $x$ to be computable if it exhibits a computable description. 

\begin{definition}[Type-two computable number]
    \label{def: computable real}
    Let $x \in \R^n$ be any real number, $x$ is \emph{Type-Two computable} if it has a computable name. i.e. there is some computable sequence $(q_i)_i \subseteq \Q^n$ that is a name for $x$. 
\end{definition}

Intuitively, this means that a real $x \in \R^n$ is (Type-Two)computable if and only if it can be computably approximated by a sequence of vectors of rational numbers. The definition above relies on some ordering of the rationals, but any fixed effective enumeration of rationals gives equivalent notions. From now on, whenever we refer to the computability of numbers in $\R^n$, we mean Type-Two computability.

\begin{definition}
An \emph{oracle machine $M$} is a Turing machine that allows for an additional one-way read-only input tape that represents some input oracle used. The machine is allowed to read this input tape up to arbitrary, but finite, lengths.
\end{definition}

One can think of oracle machines as regular Turing machines but with some access to outside information, namely the ``oracle'' input tape. The machine may use any finite amount of information on this tape. For an oracle machine $M$, and an infinite binary sequence $p \in 2^{\omega}$, $M^p$ represents the oracle machine $M$ with oracle $p$. By standard encoding, we do not differentiate between $\Q^\omega$ and $2^\omega$. 

Having defined a notion of computability on individual elements of $\R^n$, the following definition provides a notion of computability on the closed subsets of $\R^n$.

\begin{definition}[{\cite[Corollary~5.1.8]{Weihrauch_2000}}]
    A closed subset $E \subseteq \R^n$ is \emph{computable} if its corresponding distance function $x \mapsto \inf_{y \in E} \norm{x - y}$ is computable.
\end{definition}

It can be easily seen that every $\folr$ definable closed set is computable.

\begin{theorem}
    \label{thm: definable sets are computable}
    If $E \subseteq \R^n$ is a closed subset defined by the $\folr$ formula $\phi(x)$, then it is a computable closed set and its distance function is computable uniformly in $\phi(x)$.
\end{theorem}

\begin{proof}
    Let $d: \R^n \to \R$ denote the distance function for the closed set $E = \eval{\phi}$ defined via
    \[d(x) = \inf_{y \in E} \norm{x - y}\]
    It suffices to show that the relation $d(q) < r$ is uniformly decidable for $q \in \Q^n, r \in \Q^+$, which is true as this relation can be defined by the following $\folr$ formula :
    \[ \psi(q, r) \equiv \exists y (\phi(y) \land \norm{y - q}^2 < r^2)\]
    hence decidability follows as $\R$ has a decidable theory, proving $d$ to be computable.
\end{proof}

The following definition relates the use of oracle machines to computable functions in TTE.

\begin{definition}[Computable function]
    A function $f : E \subseteq \R^n \to \R^m$ with $E$ a computable closed set is \emph{computable} if there is some oracle machine $M$ such that 
    \[\boldsymbol{\forall}x \in E~\boldsymbol{\forall} p \in \Gamma(x)~((M^p(i))_i \in \Gamma(f(x)))\]
    I.e. $M$ maps names of $x$ to names of $f(x)$ for all $x \in E$. 
\end{definition}
Intuitively, this means that a function $f : \R^n \to \R^m$ is computable if and only if there is some computable algorithm such that for every $x \in \R^n$, the algorithm can output more and more accurate approximations of output $f(x)$ given more and more accurate approximations of input $x$. By this definition, any Type-Two computable function is necessarily continuous, since oracle machines can only read a finite amount of its oracle before producing an output. In other words, for all $x \in \R^n, i \in \N$, there is some corresponding $j \in \N$ such that if $f$ is provided with an approximation of $x$ accurate up to $2^{-j}$, then the output is an approximation of $f(x)$ accurate up to $2^{-i}$, therefore $f$ is continuous. The standard functions $\sin(x), \cos(x), x^2, e^x, \cdots$ are all computable through their Taylor expansions.

A useful result of computable analysis is that the classical extreme value theorem holds computably \cite[Corollary~6.2.5]{Weihrauch_2000}. The following theorem states this for functions $f : K \subset \R^n \to \R^m$ with $K$ a definable compact subset of $\R^n$, the proof is included for completeness. 

\begin{theorem}[Computable extreme value theorem {\cite[Corollary~6.2.5]{Weihrauch_2000}}]
    \label{thm: computable extreme value theorem}
    Let $f : K \to \R$ be a computable function on the compact set $K \subset \R^n$ defined by some $\folr$ formula $\phi(x)$. Then $\max_{x \in K} (f(x))$ and $\min_{x \in K} (f(x))$ are uniformly computable in $f, \phi(x)$. 
\end{theorem}

\begin{proof}
     As $K$ is definable and closed, it is a computable closed set. 
     In addition, an upper bound on the radius of $K$ can be computed from $\phi(x)$: search for $R\in \Q^+$ such that the $\folr$ formula $\phi(x) \rightarrow \norm{x}^2 < R^2$ is valid, hence a representation of the compact set $K$ \cite[~Remark 5.2.3]{Weihrauch_2000} is computable from $\phi(x)$. Consequently, a representation of the image of $K$ under the computable function $f$, $f(K)$, is computable from $\phi(x)$ as well. The computability of $\max_{x \in K} f(x)$ then follows from the computability of maximums on compact sets \cite[Lemma~5.2.6]{Weihrauch_2000} applied to $f(K)$.
\end{proof}

\section{Completeness Under Domain Constraints}
\label{section: proof generation for ode solvers}
This section establishes the completeness of $\dL$'s axiomatization for convergence with additional assumptions on domain constraints. To accomplish this, we will reduce the problem of proving error bounds for approximants of compact IVPs to differential invariance questions, which $\dL$ is effectively complete for \cite{DBLP:journals/jacm/PlatzerT20}. Intuitively, this reduction is achieved by proving a syntactically provable version of ``continuous dependence on initial data'' for ODEs in $\dL$, which, as a special case, shows that the flow function induced by the ODE is continuous. Consequently, if an approximant starts off close to the initial condition, then (after some bounded time) it will forever remain close to the true flow in the supremum norm. Thus, proofs of future error bounds of approximants provably reduce to arithmetic questions \emph{at the initial time $t_0$}.

However, since polynomial vector fields are generally non-linear and therefore do not exhibit global Lipschitz constants, it is tricky to obtain explicit and computable bounds in this reduction process. As such, this section will first assume the presence of some bounded domain constraint, which essentially reduces to the case of Lipschitz vector fields since polynomials are locally Lipschitz. \rref{sec: symbolic axioms} improves upon this, establishing that such assumptions are not necessary and can be removed, proving completeness for convergence without any additional assumptions. 

\subsection{Compact IVPs and Approximants}

The following definitions fix standard notations that will be used throughout this article. 

\begin{definition}[Notation]
The following notation will be used throughout the article.
\begin{itemize}
    \item $\R^+, \Q^+$ denotes the set of positive real/rational numbers respectively.
    \item $x$ always denotes some vectorial variable $x = (x_1, \cdots, x_n)$. 
    \item For a ring $R$, denote its ring of polynomials in the variables $x_1, \ldots, x_n$ as $R[x_1, \ldots, x_n]$. This paper only considers $R \in \{\Q, \R\}$. By a slight abuse of notation, elements $p(x) \in R[x]$ are also identified with the corresponding polynomial $p : R^n \to R$. 
    \item By a \emph{rational polynomial}, we mean some element of $\Q[x]$ where $x$ is understood to be vectorial. 
    \item $\norm{x}$ for $x \in \R^n$ always refers to the Euclidean norm, and $\norm{f}$ always refers to the sup norm for functions $f$. We sometimes write $\norm{f}_A = \sup_{x \in A} \norm{f(x)}$ to emphasize that the supremum norm of $f$ is taken on the set $A$, which is $\folr$ definable when $A$ is $\folr$ definable. 
    \item $C^k([a, b], \R^n)$ for $k \in \N$ denotes the set of functions from the closed interval $[a, b]$ to $\R^n$ with $k$ continuous derivatives on $[a, b]$. For $K$ a compact Hausdorff space, $C^0(K, \R^n)$ denotes the space of continuous functions with the usual supremum norm $\norm{f}_{K} = \sup_{x \in K}\norm{f(x)}$. When the co-domain is clear, these are also abbreviated as $C^k([a, b]), C^0(K)$. 
    \item $\IQ$ denotes the set of all compact intervals with rational endpoints, i.e
    \[\IQ = \{[a, b] : a \leq b, a, b \in \Q\}\]
    \item For $x \in \R^n, R \in \R^+$, write $B[x, R]$ for the closed ball of radius $R$ around $x$, and $B(x, R)$ for the open ball. When $x, R$ are expressible in $\dL$, $y \in B[x, R]$ and $y \in B(x, R)$ are definable via 
    \begin{align*}
        &y \in B[x, R] \iff \norm{y - x}^2 \leq R^2\\
        &y \in B(x, R) \iff \norm{y - x}^2 < R^2
    \end{align*}
    For a set $A \subseteq \R^n$, write $B[A, R]$ (and similarly $B(A, R)$) for $\bigcup_{x \in A}B[x, R]$. 
    
    \item $\folr$ denotes the set of all first-order formulas in the language of real closed fields. In this article, definable always refers to $\folr$ definable unless explicitly stated otherwise. As real closed fields admit quantifier elimination \cite{Tarski_1948}, we may assume without loss of generality that every element of $\folr$ is quantifier-free. Finally, for formulas $\phi(x) \in \folr$, $\eval{\phi}$ denotes the set defined by the formula in $\R$. I.e:
    \[\eval{\phi} = \{y \in \R^n \mid \R \models \phi(y)\}\]
    which coincides with the semantics of $\phi$ in $\dL$. 
\end{itemize}
\end{definition}

\begin{definition}[Compact IVP]
    \label{def: rational IVP}
    A compact initial value problem (IVP) is a triple 
    \[(f(x), C(x), [t_0, T]) \in \Q^n[x] \times \folr \times \IQ\] 
    where $\eval{C(x)}$ is a non-empty compact set. The variable $x$ is often suppressed for brevity, and $\eval{C}$ refers to $\eval{C(x)}$. Such a triple represents the following IVPs on $[t_0, T]$:
    \begin{align*}
        &x' = f(x)\\
        &x(t_0) = x_0 \in \eval{C}
    \end{align*}
    That is, the triple defines a collection of IVPs on some compact time horizon $[t_0, T]$ where the initial conditions are constrained to the compact set $\eval{C}$. The flow $\phi : \eval{C} \times [t_0, T] \to \R^n$ of the compact IVP (if it exists) is the flow of the vector field $x' = f(x)$ starting at $t = t_0$. I.e. $\phi(x, t_0) = x, \phi'(x, t) = f(\phi(x, t))$ for all $(x, t) \in \eval{C} \times [t_0, T]$. 
\end{definition}

Since singletons are compact, the standard notion of IVPs with a fixed initial condition $x(0) = x_0 \in \Q^n$ is a special case of \rref{def: rational IVP} where $C(x) \equiv x = x_0$. 

\begin{remark}
In practice, many IVPs contain \emph{parameters}. I.e. $x' = f(x, a)$, where the vectorial variable $a$ denotes the parameters used. It is always possible to rewrite such IVPs into:
\begin{align*}
    &x' = f(x, a)\\
    &a' = 0\\
    &x(t_0) = x_0, a(t_0) = a
\end{align*}
which forms a compact IVP when the parameters $a$ are constrained to a compact set. 
\end{remark}

\begin{example}[Moore-Greitzer jet engine model]
    \label{ex: jet engine}
    The Moore-Greitzer model of a jet engine \cite{Aylward_Parrilo_Slotine_2008, Rwth_Sankaranarayanan_Abraham_2014} for scalars $u, v$ is given by:
    \begin{align*}
        &u' = -v - 1.5u^2 - 0.5u^3 - 0.5\\
        &v' = 3u - v
    \end{align*}
    with initial conditions
    \begin{align*}
        & u(0) \geq 0.9\\
        & v(0) \geq 0.9\\
        & u(0) + v(0) - 2 \leq 0
    \end{align*}
    where $u, v$ measures the mass flow and the pressure rise respectively. Since the initial conditions define a (semialgebraic) compact subset of $\R^2$, for any $T \in \Q^+$, we may express this model on the time horizon $[0, T]$ as a compact IVP $(f(u, v), \Delta(u, v), [0, T])$ where:
    \begin{itemize}
        \item $f(u, v) = (-v - 1.5u^2 - 0.5u^3 - 0.5, 3u - v)$.
        \item $\Delta(u, v) \equiv u \geq 0.9 \land v \geq 0.9 \land u + v - 2 \leq 0$
    \end{itemize}
    This model will serve as a running example through this article, culminating in a proof of the error bound of a numerically computed approximation to the true flow in \rref{ex: jet engine bound}. All proofs concerning the Moore-Greitzer model have been verified using \KeYmaeraX\footnote{\url{https://github.com/LongQianQL/Compact_IVP_Example}}. 
\end{example}

The first step is to establish a suitable representation for approximants to solutions of compact IVPs. In this article, such approximants are taken to be functions definable in $\folr$, which are also the semialgebraic functions over $\Q$ \cite{Bochnak_Coste_Roy_2013}. The following definition restricts to the particular case of definable functions with domain being a subset of $\R^{n + 1}$ and co-domain $\R^n$ for $n \geq 1$. This is because the approximants represent approximations to the flow induced by compact IVPs, as such, they will always be functions from $\R^{n + 1}$ ($n$ space variables, $1$ time variable) to $\R^n$. 

\begin{definition}[$\folr$ Definable functions]
    \label{def: definable function}
    A function $f : A \subseteq \R^{n + 1} \to \R^n$ with definable domain $A$ is \emph{definable} if there exists a $\folr$ formula $\phi(x, t, y)$ such that for all $x, y \in \R^n, t \in \R$
    \[f(x, t) = y \iff \R \models \phi(x, t, y)\]
    In this case, we say that $\phi(x, t, y)$ is a \emph{representation} of $f$. 
\end{definition}

\begin{remark}
    As $\dL$ strictly extends $\folr$ \cite{DBLP:conf/lics/Platzer12b}, $\folr$ definable functions are also $\dL$ definable. 
\end{remark}

The class of definable functions is very versatile. In particular, polynomials and splines with rational coefficients are definable in a natural way. As a consequence, one can always carry out spline/polynomial interpolation on a mesh-grid of points to arrive at a definable approximant. 

\begin{remark}
    \label{rem: errors for definable functions expressible}
    While vanilla $\dL$ only allows for polynomials as terms, definable functions in the sense of \rref{def: definable function} can be expressed as well using their representations. E.g. suppose $f : \R^{n + 1} \to \R^n$ has representation $\phi(x, t, y)$ and $u \in \V^n$ is some vectorial variable, $\norm{f(x, t) - u}^2 \leq M^2$ can then be expressed by
    \[\exists y (\phi(x, t, y) \land \norm{y - u}^2 \leq M^2)\]
    such abbreviations will be used throughout the article for formulas containing definable functions.  
\end{remark}

The following definition makes precise the notion of approximations used in this article. 

\begin{definition}[Local definable approximant]
\label{def: LDA}
Let $(f(x), C(x), [t_0, T])$ be a compact IVP and $\phi(x, t)$ be its corresponding flow function. A {\emph{local definable approximant}} (LDA) for this compact IVP is a computable function $\Phi : \N \to \folr$ such that the following holds: 
\begin{enumerate}
    \item $\phi(x, t) : \eval{C} \times [t_0, T] \to \R^n$ is well-defined (i.e. does not exhibit finite time blow-up). 
    \item For all $k \in \N$, $\Phi(k)$ defines a function $\Phi_k : \eval{C} \times [t_0, T] \to \R^n$ (thus each $\Phi_k$ is a definable function with representation $\Phi(k)$). 
    \item The sequence of functions $(\Phi_k)_k$ converges to $\phi$ in $C^0(\eval{C} \times [t_0, T], \R^n)$. 
    \item For all $k \in \N$, the function $\Phi_k$ is differentiable in its second (time) variable, and the sequence of time derivatives $(\Phi_k')_k$ converges to $\phi'$ in $C^0(\eval{C} \times [t_0, T], \R^n)$.
\end{enumerate} 
\end{definition}

\begin{example}
For IVPs the sequence of Picard iterates \cite{Walter_1998} always form a LDA over a sufficiently small interval. I.e. For every IVP $x' = f(x), x(t_0) = x_0 \in \Q^n$, there always exists a sufficiently small $S > t_0$ such that the Picard iterates form a LDA for the compact IVP $(C(x) \equiv x = x_0, f(x), [t_0, S])$. To show this, recall that the Picard iterates $(\phi_k)_k$ of the IVP $x' = f(x), x(t_0) = x_0$ are defined inductively by

\begin{itemize}
    \item $\phi_0(t) = x_0$.
    \item $\phi_{k + 1}(t) = x_0 + \int_{t_0}^tf(\phi_{k}(s))ds$.
\end{itemize}
By the Picard-Lindel\"of theorem the iterates converge uniformly to the unique solution on some interval $[t_0, S]$ for some $S > t_0$. Furthermore, this sequence of iterates are simply polynomials in $t$ with rational coefficients since integrals of rational polynomials are rational polynomials. As integrals of polynomials are computable, the sequence of iterates $(\phi_k)_{k}$ and their representations are computable. It remains to show that the sequence $(\phi_k')_k$ converges to $x'$ on $[t_0, S]$. Indeed, let $x(t) : [t_0, S] \to \R^n$ denote the unique solution that this sequence converges to. We have
\[\abs{x'(t) - \phi_{k + 1}'(t)} = \abs{f(x(t)) - f(\phi_k(t))}\]
Note that $B[x([t_0, S]), 1]$ (the set of points of Euclidean distance at most 1 away from $x([t_0, S])$) is compact as $x$ is continuous and $[t_0, S]$ is compact. Hence, as $f$ is a polynomial vector field and therefore locally Lipschitz, there exists some $L > 0$ which is the Lipschitz constant of $f$ on $B[x([t_0, T]), 1]$ (we can computably find such a value by computing the maximum of $f$'s partial derivatives on $B[x([t_0, T]), 1]$, but this is \emph{not required} to prove the iterates form a LDA). Since $(\phi_k)_k$ converges to $x$ on $[t_0, S]$, we have $\phi_k([t_0, S]) \subseteq B[x([t_0, S]), 1]$ for all sufficiently large $k$. In other words, for all sufficiently large $k$, for all $t \in [t_0, S]$, we have:
\[\abs{x'(t) - \phi_{k + 1}'(t)} = \abs{f(x(t)) - f(\phi_k(t))} \leq L\norm{x - \phi_k}_{[t_0, S]} \xrightarrow{k \to \infty} 0\]
The Picard-Lindel\"of theorem says that $\phi_k \to x$ uniformly i.e. in the supremum norm, and the above computation shows $(\phi_k')_k \to x'$ on $[t_0, S]$ under the sup-norm as well, therefore the sequence of Picard iterates $(\phi_k)_k$ forms a LDA.
\end{example}

The example above shows that the Picard iterates will always be  LDAs over sufficiently small intervals for IVPs with fixed initial values. The following theorem shows that for \emph{any compact IVP}, a corresponding LDA can always be constructed effectively on the entire interval $[t_0, T]$ provided that the compact IVP does not exhibit finite-time blow up on $[t_0, T]$. 

\begin{theorem}[Computable LDA]
    \label{thm: abstract solvers always exist}
    Let $(f(x), C(x), [t_0, T])$ be a compact IVP where the corresponding flow $\phi(x, t)$ is well-defined on $\eval{C} \times [t_0, T]$. Then there exists a corresponding LDA $\Phi$ that is uniformly computable in the compact IVP such that for all $k \in \N$, every component of the function defined by $\Phi(k)$ is a rational polynomial in $x, t$. 
\end{theorem}

\begin{proof}
    Since rational polynomials are $\folr$ definable, it suffices to computably construct a sequence of rational polynomials $(p^i_{k})_{1 \leq i \leq n, k \in \N} \subseteq \Q[x, t]$ such that the corresponding sequence $(p_k)_k \subseteq \Q^n[x, t]$ defined via $p_k = (p^1_k, \cdots, p^n_k)$ satisfies:
    \begin{enumerate}
        \item The sequence $(p_k)_k$ converges to $\phi$ in $C^0(\eval{C} \times [t_0, T], \R^n)$. 
        \item The sequence $(p_k')_k$ converges to $\phi'$ in $C^0(\eval{C} \times [t_0, T], \R^n)$. 
    \end{enumerate}
    This is sufficient as one can then define the formulas: 
    \[\psi_k(x, y, t) \equiv  \bigwedge_{1 \leq i \leq n} y_i = p^i_k(x, t)\]
    Properties $(1), (2)$ then imply that the function $\Phi : \N \to \folr$ defined by $\Phi(k) = \psi_k$ forms a LDA for the compact IVP. 
    
    To construct the desired sequence $(p^i_{k})_{1 \leq i \leq n, k \in \N}$, first fix some $1 \leq i \leq n$, let $k \in \N$ be arbitrary and notice that it suffices to construct some $p^i_k \in \Q[x, t]$ satisfying the following:
    \begin{enumerate}
        \item $\norm{\phi_i - p^i_k}_{C^0(\eval{C} \times [t_0, T])} < 2^{-k}$ where $\phi_i$ denotes the $i$-th component of $\phi$. 
        \item $\norm{\left(p^i_k\right)' - \phi_i'}_{C^0(\eval{C} \times [t_0, T])} < 2^{-k}$ where the derivative is taken with respect to time variable. 
    \end{enumerate}
    As we may then carry out the same construction for arbitrary $1 \leq i \leq n$ and $k\in \N$ to obtain
    \[\norm{(\phi_1, \cdots, \phi_n) - (p^1_k, \cdots, p^n_k)}_{C^0(\eval{C} \times [t_0, T])} \leq \sum_{i = 1}^n \norm{\phi_i - p^i_k}_{C^0(\eval{C} \times [t_0, T])} < n2^{-k}\]
    which converges to $0$ as $k \to \infty$, likewise for $\norm{\phi' - p_k'}_{C^0(\eval{C} \times [t_0, T])}$. To carry out the construction for a fixed index $1 \leq i \leq n$ and $k \in \N$, first note that the flow function $\phi(x, t) : \eval{C} \times [t_0, T] \to \R^n$ is computable for compact IVPs \cite{Graca_Zhong_Buescu_2009} as $\eval{C}$ is a computably closed set by \rref{thm: definable sets are computable}. Because $f \in \Q^n[x, t]$ is also computable, consequently the time-derivative of $\phi$, $\phi'(x, t) = f(\phi(x, t))$ is also computable on $\eval{C} \times [t_0, T]$. The effective Stone-Weierstrass theorem \cite[Theorem~6.1.10]{Weihrauch_2000} then allows us to compute some $q^i_k \in \Q[x, t]$ such that 
    \[\norm{q^i_k - \phi_i'}_{C^0(\eval{C} \times [t_0, T])} < \frac{2^{-k - 1}}{\max(T - t_0, 1)}\]
    Define $p^i_k \in \Q[x, t]$ by
    \[p^i_k(x, t) = x_i + \int_{t_0}^t q^i_k(x, s) ds\]
    which is computable since $q^i_k$ is a polynomial with rational coefficients, hence its integral in the time variable $t$ can be directly computed symbolically using the elementary power rule. It remains to verify that conditions (1) and (2) are met:
    \begin{enumerate}
        \item Direct computations for $(x_0, t) \in \eval{C} \times [t_0, T]$ yields:
        \[\abs{p^i_k(x_0, t) - \phi_i(x_0, t)} \leq \int_{t_0}^t \abs{q^k_i(x_0, s) - \phi_i'(x_0, s)} \leq (T - t_0)\frac{2^{-k - 1}}{\max(T - t_0, 1)} \leq 2^{-k - 1} < 2^{-k}\]
        \item Noticing that the time derivative of $p^i_k$ is $q^i_k$ which is continuous in both variables, a similar computation to the above for any $(x_0, t) \in \eval{C} \times [t_0, T]$ yields:
        \[\abs{(p^i_{k})'(x_0, t) - \phi_i'(x_0, t)} = \abs{q^i_k(x_0, t) - \phi_{i}'(x_0, t)} \leq \norm{q^i_k - \phi_i'}_{C^0(\eval{C} \times [t_0, T])} < 2^{-k - 1}\]
        thereby condition $(2)$ is also satisfied. 
    \end{enumerate}
    Since the construction is uniformly computable for all $1 \leq i \leq n$ and $k \in \N$, this completes the proof. 
\end{proof}

\subsection{Provable IVP Approximants}

The following technical lemma proves the validity of a class of differential invariants capturing the ``continuous dependence on initial conditions'' characteristic of flow functions. Such invariants are then used in proving the desired error bounds under the presence of a bounded domain constraint $B(x)$ containing the true flow of the compact IVP $(f(x), C(X), [t_0, T])$. Note that this is an assumption on the $\folr$ formula $B(x)$ itself, rather than a constraint on the values of the variables. 
\begin{lemma}[continuous dependence on initial conditions]
\label{lem: continuous dependence}
Let $(f(x), C(x), [t_0, T])$ be a compact IVP and $B(x) \in \folr$. Further assume that the following holds:

\begin{enumerate}
    \item The flow $\phi(x, t)$ of the compact IVP is well-defined on $\eval{C} \times [t_0, T]$.
    \item $\eval{B} \subset \R^n$ is a bounded set containing $\phi(\eval{C}, [t_0, T])$. 
\end{enumerate}

Then for all $K \in \Q^+$ greater than or equal to the Lipschitz constant of $f(x)$ on $\eval{B}$, for all LDA $\Phi$, for all positive rational $h \in \Q^+$, for all sufficiently large $k \in \N$, the following is a valid differential invariant in $\dL$:
\[\psi_{k}(x_0, x, g, t) \rightarrow [x' = f(x), g' = Kg, t' = 1 \& t \leq T \land B(x)]\psi_{k}(x_0, x, g, t)\]
With $\psi_{k}(x_0, x, g, t)$ defined as:
\begin{align*}
    &\psi_{k}(x_0, x, g, t) \mequiv t \geq t_0 \land g \geq 1 \land C(x_0) \land \norm{x - \Phi_k(x_0, t)}^2 \leq \eps(g, t)^2\\
    & \eps(g, t) \mequiv h(1 + t - t_0)g - h
\end{align*}
A corresponding witness $k$ can also be computed uniformly from the compact IVP, $B(x)$, $\Phi$ and $h$.
\end{lemma}

\rref{lem: continuous dependence} computes some $k$ witnessing the validity of the differential invariant, but it proves the stronger assertion that there exists some $k_0 \in \N$ such that for all $k \geq k_0$, the differential invariant at index $k$ is valid. Such a threshold $k_0$ is in general not computable, because LDAs are not required to have a computable rate of convergence to the true flow to allow for more general approximants. This is similar to the difference between computably enumerable real numbers, which have computable sequences of rationals converging to them, and computable real numbers, which have computable sequences of rationals converging to them with \emph{computable rates of convergence}. 

\begin{remark}
    Intuitively, the idea of the proof of \rref{lem: continuous dependence} is to find some $\folr$ formula $\text{Small}(x, t)$ that captures the difference between the flow $\phi(x, t)$ and the approximation $\Phi_k(x, t)$ being small. By the continuous dependence of $\phi$ on its initial conditions, the differential invariant
    \[\text{Small}(x, t) \rightarrow \dbox{x' = f(x), t' = 1 \& t \leq T}{\text{Small}(x, t)}\]
    is valid and by \rref{thm:diff inv completeness} \irref{dinv} will give a syntactic proof. However, while the dependence of the flow $\phi(x, t)$ on its initial conditions is continuous, the error rate may grow like $e^{Lt}$ with $L$ being the Lipschitz constant of the vector field $f$ on $\phi(\eval{C}, [t_0, T])$. Since the theory of real exponential fields is not known to be decidable \cite{real_exp_field_decidability} and $e^x$ is not supported in vanilla $\dL$, we will have to encode it via an ODE. In the definition of $\psi_k(x_0, x, g, t)$, the variable $g$ represents the exponential function, as indicated by its ODE $g' = Kg$. The error function $\eps(g, t)$ represents this error rate being scaled by the exponential function $g$. Lastly, as $f(x)$ is in general only \emph{locally} Lipschitz, the domain constraint $B(x)$ is needed in order to obtain a \emph{fixed} upper bound on the Lipschitz constant. 
\end{remark}

The following integral form of Gr\"onwall's inequality is needed to prove \rref{lem: continuous dependence}.

\begin{lemma}[Gr\"onwall's inequality \cite{Gronwall_1919, Walter_1998}]
\label{lem: gronwall}
    Let $[a, b] \subset \R$ be an interval of the real line, $u \in C([a, b], \R)$ and $\alpha, \beta \in \R$. Further suppose that for all $t \in [a, b]$, we have:
    \[u(t) \leq \alpha + \int^t_a\beta u(s) ds\]
    Then the following inequality holds for all $t \in [a, b]$:
    \[u(t) \leq \alpha e^{\beta(t - a)}\]
\end{lemma}

With the lemma above, we are now ready to prove \rref{lem: continuous dependence}.

\begin{proof}[Proof of \rref{lem: continuous dependence}]
Suppose that $\psi_{k}(x_0, x, g, t)$ is satisfied at some initial state, that is, there is some $t_1, g_0 \in \R, y_0, y \in \R^n$ such that $\psi_{k}(y_0, y, g_0, t_1)$ holds. Let $\phi(x, t)$ denote the flow of the compact IVP and $\psi(g_0, t)$ denote the flow of $g$ along $g' = Kg$ with initial condition $\psi(g_0, t_0) = g_0$. By definition at time $t \in [t_1, T]$ the variable $x$ equals $\phi(y, t - t_1 + t_0)$. Define the following function for $(x_0, t) \in \eval{C} \times [t_0, T]$ recording the difference between the true solution and the approximant at time $t$ with initial condition $x_0$:

\[R_{k}(x_0, t) = \Phi_k(x_0,t) - \phi(x_0, t)\]
To establish the validity of the invariant, it suffices to show 
\[\norm{\Phi_k(y_0, t) - \phi(y, t - t_1 + t_0)} \leq \eps(\psi(g_0, t - t_1 + t_0), t)\]
for all $t \in [t_1, T]$ such that the domain constraint is maintained. This is because $g$ satisfies the ODE $g' = Kg$, thus $\psi(g_0, t - t_1 + t_0) = g_0e^{K(t - t_1)} \geq g_0 \geq 1$, hence $g \geq 1$ is always satisfied by the assumption of $K \in \Q^+$. The condition $t \geq t_0$ is also satisfied as the ODE $t' = 1$ is strictly increasing, therefore $t \geq t_1 \geq t_0$. Finally $C(y_0)$ remains true since $y_0$ does not change along the ODE. To handle the non-trivial inequality, notice that $t \geq t_1$, therefore:
\begin{align*}
    \Phi_k(y_0, t) - \phi(y, t - t_1 + t_0) &= R_{k}(y_0, t) + \phi(y_0, t) - \phi(y, t - t_1 + t_0)\\
    &= R_{k}(y_0, t) + \phi(\phi(y_0, t_1), t - t_1 + t_0) - \phi(y, t - t_1 + t_0)\\
    &= R_k(y_0, t) + \phi(y_0, t_1) + \int_{t_0}^{t - t_1 + t_0} f(\phi(\phi(y_0, t_1), s)) ds - y - \int_{t_0}^{t - t_1 + t_0} f(\phi(y, s)) ds
\end{align*}
Applying the triangle inequality gives:
\begin{equation}
\label{eqn: main ineq}
\norm{\Phi_k(y_0, t) - \phi(y, t - t_1 + t_0)} \leq \norm{R_k(y_0, t) + \phi(y_0, t_1) - y} + \int_{t_0}^{t - t_1 + t_0} \norm{f(\phi(\phi(y_0, t_1), s)) - f(\phi(y, s))} ds
\end{equation}
Now we crucially use the fact that $B(x)$ is both a domain constraint and assumed to contain the flow $\phi(x_0, t)$ for $(x_0, t) \in \eval{C} \times [t_0, T]$ to see that for $s \in [t_0, t - t_1 + t_0]$, we will always have $B(\phi(\phi(y_0, t_1), s))$ and $B(\phi(y, s))$ (i.e. $\phi(y_0, t_1 + s)$, $\phi(y, s)$ both belong to the bounded set $\eval{B}$). Letting $L$ denote the Lipschitz constant of $f(x)$ on $\eval{B}$ (recall that such a constant always exists since $f(x)$ is locally Lipschitz), we have:

\begin{equation}
    \label{eqn: Lipschitz bound}
    \int_{t_0}^{t - t_1 + t_0} \norm{f(\phi(\phi(y_0, t_1), s)) - f(\phi(y, s))} ds\\
    \leq L\int_{t_0}^{t - t_1 + t_0} \norm{\phi(\phi(y_0, t_1), s) - \phi(y, s)} ds
\end{equation}

We will now establish the following bound: 
\begin{equation}
    \label{eqn: gronwall bound}
    \norm{\phi(\phi(y_0, t_1), s) - \phi(y, s)} \leq \norm{\phi(y_0, t_1) - y}e^{L(s - t_0)}
\end{equation}
To do this, define $d \in C^1([t_0, t - t_1 + t_0], \R^n)$ by $d(s) = \phi(\phi(y_0, t_1), s) - \phi(y, s)$. Direct manipulations yield:
\begin{align*}
    \norm{d(s)} &= \norm{d(t_0) + \int_{t_0}^sd'(r) dr} \leq \norm{d(t_0)} + \int_{t_0}^s\norm{d'(r)}dr\\
    &= \norm{\phi(y_0, t_1) - y} + \int_{t_0}^s\norm{f(\phi(\phi(y_0, t_1), r)) - f(\phi(y, r))}dr\\
    &\leq \norm{\phi(y_0, t_1) - y} + L\int_{t_0}^s\norm{\phi(\phi(y_0, t_1), r) - \phi(y, r)}dr\\
    &= \norm{\phi(y_0, t_1) - y} + L\int_{t_0}^s\norm{d(r)}dr
\end{align*}
Note that in this derivation, we again utilized the assumption that $f(x)$ is Lipschitz on $\eval{B}$ with Lipschitz constant $L$ in the second to last inequality. Applying Gr\"onwall's inequality (\rref{lem: gronwall}) with $u(s) = \norm{d(s)}, \alpha = \norm{\phi(y_0, t_1) - y}, \beta = L$ then gives the desired bound (equation \ref{eqn: gronwall bound}).  Applying this to inequality \eqref{eqn: Lipschitz bound} results in:
\begin{align*}
    \int_{t_0}^{t - t_1 + t_0} \norm{f(\phi(\phi(y_0, t_1), s)) - f(\phi(y, s))} ds &\leq L\norm{\phi(y_0, t_1) - y}\int_{t_0}^{t - t_1 + t_0} e^{L(s - t_0)} ds\\
    &= \norm{\phi(y_0, t_1) - y}\left(e^{L(t - t_1)} - 1\right)
\end{align*}
Substituting this back into inequality \eqref{eqn: main ineq} gives:
\[\norm{\Phi_k(y_0, t) - \phi(y, t - t_1 + t_0)} \leq \norm{R_k(y_0, t) + \phi(y_0, t_1) - y} + \norm{\phi(y_0, t_1) - y}\left(e^{L(t - t_1)} - 1\right)\]
Recalling $R_k(y_0, t) = \Phi_k(y_0, t) - \phi(y_0, t)$ yields:
\begin{align*}
    \norm{\Phi_k(y_0, t) - \phi(y, t-t_1 + t_0)} &\leq \norm{R_k(y_0, t) + \Phi_k(y_0, t_1) - R_k(y_0, t_1) - y}\\
    &+ \norm{R_k(y_0, t_1) - \Phi_k(y_0, t_1) + y}\left(e^{L(t - t_1)} - 1 \right)
\end{align*}
Utilizing the triangle inequality and rearranging, we arrive at:
\[\norm{\Phi_k(y_0, t) - \phi(y, t - t_1 + t_0)} \leq \norm{\Phi_k(y_0, t_1) - y}e^{L(t - t_1)} + \norm{R_k(y_0, t) - R_k(y_0, t_1)} + \norm{R_k(y_0, t_1)}(e^{L(t - t_1)} - 1)\]
Recall that we may choose $k$ arbitrarily large and $\norm{R_k}_{\eval{C} \times [t_0, T]} \xrightarrow{k \to \infty} 0$, hence assume that $k$ is large enough to witness $\norm{R_k}_{\eval{C} \times [t_0, T]} \leq h$. Also by assumption on $g_0, y_0, y, t_1$, the following holds:
\[\norm{\Phi_k(y_0, t_1) - y} \leq \eps(g_0, t_1)\]
Rearranging yields:
\[\norm{\Phi_k(y_0, t) - \phi(y, t - t_1 + t_0)} \leq (\eps(g_0, t_1) + h)e^{L(t - t_1)} + \norm{R_k(y_0, t) - R_k(y_0, t_1)} - h\]
Expanding $\eps(g_0, t_1) = h(1 + t_1 - t_0)g_0 - h$ by construction and requiring $K \geq L$ yields:
\begin{align*}
    \norm{\Phi_k(y_0, t) - \phi(y, t - t_1 + t_0)} &\leq h(1 + t_1 - t_0)g_0e^{L(t - t_1)} + \norm{R_k(y_0, t) - R_k(y_0, t_1)} - h\\
    &\leq h(1 + t_1 - t + t - t_0)g_0e^{K(t - t_1)} + \norm{R_k(y_0, t) - R_k(y_0, t_1)} - h\\
    &= h(1 + t - t_0)g_0e^{K(t - t_1)} - h + \norm{R_k(y_0, t) - R_k(y_0, t_1)} - hg_0e^{K(t - t_1)}(t - t_1)\\
    &= \eps(\psi(g_0, t - t_1 + t_0), t) + \norm{R_k(y_0, t) - R_k(y_0, t_1)} - hg_0e^{K(t - t_1)}(t - t_1)\\
    &\leq \eps(\psi(g_0, t - t_1 + t_0), t) + \norm{R_k(y_0, t) - R_k(y_0, t_1)} - h(t - t_1)
\end{align*}
where the second equality uses the fact that $\psi(g_0, t)$ is the flow of $g' = Kg$ starting at $g(t_0) = g_0$, thus $\psi(g_0, t - t_1 + t_0) = g_0e^{K(t - t_1)}$. The final inequality follows from $t \geq t_1$ and $g_0 \geq 1$. Now define
\[M_k = \max_{y_0 \in \eval{C}, t \in [t_0, T]}\norm{R_k'(y_0, t)}\] 
which is well-defined as $R'_k \in C^0(\eval{C} \times [t_0, T], \R^n)$. Since $(\Phi_k')_k$ converges uniformly to $\phi'$ on $\eval{C} \times [t_0, T]$, $(M_k)_k$ will converge to $0$. Thus, choose $k$ large enough so that $M_k \leq h$, which allows us to deduce:
\begin{align*}
    \eps(\psi(g_0, t - t_1 + t_0), t) + \norm{R_k(y_0, t) - R_k(y_0, t_1)} - h(t - t_1) &\leq \eps(\psi(g_0, t - t_1 + t_0), t) + M_k(t - t_1) - h(t - t_1)\\
    &\leq \eps(\psi(g_0, t - t_1 + t_0), t)
\end{align*}
exactly as desired. Thus, for any $h > 0$, choosing $k$ large enough such that the following conditions are met will witness the validity of the differential invariant.

\begin{itemize}
    \item $\max_{y_0 \in \eval{C}, t \in [t_0, T]} \norm{R_k(y_0, t)} \leq h$
    \item $\max_{y_0 \in \eval{C}, t \in [t_0, T]} \norm{R'_k(y_0, t)}  \leq h$
\end{itemize}

Furthermore, since maximums of computable functions are computable by \rref{thm: computable extreme value theorem}, a satisfying index $k$ can be found computably. To see that $K$ can be effectively computed and chosen to be a rational, note that we only require $K \geq L$ to hold, so one can search through all positive rationals $K \in \Q^+$ and halt when the following $\folr$ formula is decided to be true
\[\forall x \forall y \left(B(x) \land B(y) \rightarrow \norm{f(x) - f(y)}^2 \leq K^2\norm{x - y}^2\right)\]
and since $\R$ has a computable theory by quantifier elimination \cite{Tarski_1948}, this search is computable. 
\end{proof}

The ``continuous dependence on initial conditions'' property proven by \rref{lem: continuous dependence} provides control on the errors induced by LDAs, and is crucial in establishing completeness for LDAs in \rref{thm: effective proofs without upper bounds and constraints}. The following example gives a sense of how this can be achieved. 
\begin{example}

Consider the simple compact IVP $x' = x$, $x(0) = 1$ over the interval $[0, 5]$ (i.e. $C(x) \equiv x = 1$), which has a solution of $x(t) = e^{t}$ (and therefore we know that $\max_{t \in [0, 5]}x(t) = e^5 < 300$). In this case, the Picard iterates will form a LDA on the compact time horizon $[0, 5]$. The Picard iterates of this ODE are: 
\begin{align*}
    &\phi_0(t) = x_0\\
    &\phi_{n + 1}(t) = x_0 + \int_0^t \phi_n(s) ds
\end{align*}
Listing out the first few terms
\[\phi_0(t) = 1, \phi_1(t) = 1 + t, \phi_2(t) = 1 + t + \frac{t^2}{2}, \phi_3(t) = 1 + t + \frac{t^2}{2} + \frac{t^3}{6}\]
Where $\phi_n(t)$ is just the $n$-th Taylor approximate, and $\phi'_n(t) = \phi_{n - 1}(t)$. By Taylor's theorem, the $n$-th remainder term $R_n$ will be bounded by
\[\abs{R_n} \leq \frac{e^55^{n + 1}}{(n + 1)!}\]
And similarly, $M_n$, the $n$-th error in the derivative, will be bounded by
\[\abs{M_n} = \abs{R_{n - 1}} \leq \frac{e^55^{n}}{n!}\]
Suppose one wants to generate a proof witnessing that some Picard iterate is within $10^{-3}$ of the true solution. Picking $h = 10^{-6}$, one has (note that the Lipschitz constant is 1 here and $t \in [0, 5]$)
\[\abs{\eps(\psi(1, t), t)} \leq 10^{-6}(1 + 5)e^{5} + 10^{-6} \approx 8 \times 10^{-4} < 10^{-3}\]
Thus, if $\eps(g, t)$ gives a valid differential invariant in the sense of \rref{lem: continuous dependence}, then the error of the approximant is necessarily bounded by $10^{-3}$. Per the proof of the \rref{lem: continuous dependence}, $n$ just needs to be picked large enough so that
\[\abs{R_{n - 1}}, \abs{R_{n}} \leq 10^{-6}\]
and the differential invariant corresponding to $\eps(g, t)$ is valid. By the bound given above, we see that for $n = 28$, $\abs{R_{27}}, \abs{R_{28}} \leq 2 \times 10^{-7}$. Now consider the invariant:
\[\psi_{28} \rightarrow [x' = f(x), t' = 1 \& t \leq 5 \land \norm{x}^2 \leq 300] \psi_{28}\]
Since Picard iterates always satisfy $\psi_{k}(1, x(0), 1, 0)$ as they have the correct values at $t = 0$, the invariant generated above is valid and witnesses an error bound of at most $10^{-3}$. Furthermore, this differential invariant can be independently verified by a proof checker for $\dL$ \cite{DBLP:conf/cade/FultonMQVP15, DBLP:conf/cpp/BohrerRVVP17}, taking advantage of the effective axiomatisation of differential invariants \cite{DBLP:journals/jacm/PlatzerT20} which reduces the verification of differential invariants down to questions of real arithmetic. When combined with formally-verified decision procedures for real arithmetic \cite{DBLP:conf/cpp/KosaianTP23}, this gives a complete verification of the validity of the invariant,
illustrating how \rref{lem: continuous dependence} can be used to produce proofs of error bounds of approximants. 
\end{example}

\begin{example}[Invariant for Moore-Greitzer]
    \label{ex: jet engine invariant}
    Recall that the dynamics of the Moore-Greitzer model is given by:
    \begin{align*}
        &u' = f_1(u, v) =  -v - 1.5u^2 - 0.5u^3 - 0.5\\
        &v' = f_2(u, v) = 3u - v
    \end{align*}
    with compact initial conditions 
    \[\Delta(u, v) \equiv 0.9 \leq u \land 0.9 \leq v \land u + v \leq 2\]
    Motivated by prior works which numerically computes reachability enclosures of this system via successive iterations \cite[~Table 1]{Rwth_Sankaranarayanan_Abraham_2014} over many time steps without corresponding syntactic proofs, we compute a \emph{provable approximant} to the flow over one such time step, corresponding to $T = 0.02$. The approximant $\Phi(u_0, v_0, t)$ for which we will prove its accuracy is given by (recall that $\Phi(u_0, v_0, t) = (\Phi_1(u_0, v_0, t), \Phi_2(u_0, v_0, t))$):
    \begin{align*}
        \Phi_1(u_0, v_0, t) =\hspace{0.5em}& u_0 + t c^1_u(u_0, v_0) + t^2c^2_u(u_0, v_0) + t^3 c^3_u(u_0, v_0)\\
        c^1_u(u_0, v_0) =\hspace{0.5em}& -\frac{u_0^3}{2}-\frac{3 u_0^2}{2}-v_0-0.5\\
        c^2_u(u_0, v_0) =\hspace{0.5em}& \frac{3 u_0^5}{8}+\frac{15 u_0^4}{8}+\frac{9 u_0^3}{4}+\frac{3 u_0^2 v_0}{4}+0.375 u_0^2+\frac{3 u_0 v_0}{2}-0.75
   u_0+\frac{v_0}{2}\\
        c^3_u(u_0, v_0) =\hspace{0.5em}& -\frac{5 u_0^7}{16}-\frac{35 u_0^6}{16}-\frac{39 u_0^5}{8}-\frac{7 u_0^4 v_0}{8}-3.8125 u_0^4-\frac{7 u_0^3
   v_0}{2}-0.75 u_0^3\\
   \hspace{0.5em}&-\frac{13 u_0^2 v_0}{4}+\frac{3 u_0^2}{4}-\frac{u_0 v_0^2}{2}-u_0 v_0+0.375
   u_0-\frac{v_0^2}{2}-\frac{v_0}{6}+0.125\\
        \Phi_2(u_0, v_0, t) =\hspace{0.5em}& v_0 + t c^1_v(u_0, v_0) + t^2c^2_v(u_0, v_0) + t^3 c^3_v(u_0, v_0)\\
        c^1_v(u_0, v_0) =\hspace{0.5em}& -\frac{u_0^3}{2}-\frac{3 u_0^2}{2}-v_0-0.5\\
        c^2_v(u_0, v_0) =\hspace{0.5em}& 3 u_0-v_0\\
        c^3_v(u_0, v_0) =\hspace{0.5em}& \frac{3 u_0^5}{8}+\frac{15 u_0^4}{8}+\frac{5 u_0^3}{2}+\frac{3 u_0^2 v_0}{4}+1.125 u_0^2+\frac{3 u_0 v_0}{2}-0.25
   u_0+\frac{5 v_0}{6}+0.25
    \end{align*}
    Such an approximant was computed by Picard iteration with appropriate rounding on the coefficients. It is important to note that LDA approximants are not limited to be Picard iterates, and the proofs of accuracy only depends on the true errors. To apply \rref{lem: continuous dependence}, the following constructs are needed:
    \begin{itemize}
        \item $h \in \Q^+$, bounding the error of the approximant.
        \item $B(u, v) \in \folr$ characterizing a bounded set that contains the flow.
        \item $K \in \Q^+$ larger than or equal to the Lipschitz constant of $f(u, v)$ on $\eval{B}$.  
    \end{itemize}
    These values can be computed numerically by any method of choice. For example by numerically sampling, we see that the choices $h = 4 \times 10^{-3}, K = 8$ and 
    \[B(u, v) \equiv 0.781 < u < 1.109 \land 0.891 < v < 1.199 \land u + v < 2.25\]
    satisfy such requirements, therefore the invariant 
    \[\psi(u_0, v_0, u, v, g, t) \rightarrow [u' = f_1(u, v), v' = f_2(u, v), g' = 8g, t' = 1 \& t \leq 0.02 \land B(u, v)]\psi(u_0, v_0, u, v, g, t)\]
    with
    \begin{align*}
        \psi(u_0, v_0, u, v, g, t) & \equiv t \geq 0 \land g \geq 1 \land \Delta(u_0, v_0) \land \norm{(u - \Phi_1(u_0, v_0, t), v - \Phi_2(u_0, v_0, t))}^2 \leq \eps(g, t)^2\\
        \eps(g, t) &\equiv 4 \times 10^{-3} ((1 + t)g - 1)
    \end{align*}
    is valid and provable by \irref{dinv}. Crucially, while the approximation and $u, K, B(u, v)$ were all obtained \emph{numerically}, the validity of the invariant is \emph{deductively proven} with a \emph{proof} in $\dL$ that can be independently verified by proof checkers such as \KeYmaeraX \cite{Fulton_Mitsch_Quesel_Volp_Platzer_2015}. Later examples build off of this differential invariant and eventually prove that the approximant $\Phi(u_0, v_0, t)$ has an error of at most $5 \times 10^{-3}$ on $\eval{\Delta} \times [0, 0.02]$.  
\end{example}

\begin{remark}
\label{rem: bounded derivatives is enough for convergence}
While \rref{lem: continuous dependence} above applies to all LDAs, it would be interesting to know if the conditions can be relaxed to allow for approximants that do not converge in derivative. The above result still holds when there is only a \emph{subsequence} of approximants that converge in derivative to $\phi'$. Hence, the result remains true if we just assume that the approximants have bounded first and second derivatives, as this allows us to construct a convergent subsequence using Arzelà–Ascoli  \cite{Walter_1998}. Even though one cannot generally compute this convergent subsequence directly, since differential invariants can be effectively decided by \rref{thm:diff inv completeness}, it suffices to perform an unbounded search across all approximants, halting whenever one of the desired invariants is decided to be valid. 
\end{remark}

Building on \rref{lem: continuous dependence}, the following theorem reduces the problem of proving convergence of LDAs to arithmetic questions involving the exponential function. 
\begin{theorem}[Derivable LDA]
\label{thm: derivable lda}
Let $(C(x), f(x), [t_0, T])$ be a compact IVP with $\Phi$ a LDA, $B(x)$ a $\folr$ formula, $c, K \in \Q^+$ rational constants. Assume that the following holds:

\begin{enumerate}
    \item The flow $\phi(x, t)$ of the compact IVP is well-defined on $\eval{C} \times [t_0, T]$. 
    \item $\eval{B} \subset \R^n$ is a bounded set containing $\phi(\eval{C}, [t_0, T])$. 
    \item $K$ is greater than or equal to the Lipschitz constant of $f(x)$ on $\eval{B}$.
    \item $c > 1$. 
\end{enumerate}

Then for all $M, \eps \in \Q^+$, for all sufficiently large $k \in \N$, the following proof rule is syntactically derivable in $\dL$, where $x, g, t, x_0$ are symbolic variables.

\begin{calculus}
    \cinferenceRule[LDA_b|{$\text{LDA}$}]{LDA error bounds with bounded domain constraint}
        {\linferenceRule[sequent]
          {\lsequent{}{g = c \land t = t_0 \rightarrow \dbox{g' = Kg, t' = 1 \& t \leq T}{g \leq M}}}
          {\lsequent[g]{}{C(x) \land x = x_0 \land t = t_0 \rightarrow \dbox{x' = f(x), t' = 1 \& t \leq T \land B(x)}{\norm{x - \Phi_k(x_0, t)}^2 \leq M^2\eps^2}}}
        }
        {}
\end{calculus}

For each $\eps \in \Q^+$, a corresponding $k$ can be computed uniformly in the compact IVP, $\Phi$, $c$ and $\eps$. 
\end{theorem}

\rref{thm: derivable lda} gives an effective way of reducing rigorous proofs for error bounds of LDAs in $\dL$ under the presence of some bounded domain $B(x)$ to the problem of proving upper bounds of the exponential function over a bounded interval. \rref{sec: Taylor upper bounds} shows that proofs of such upper bounds are always possible even if decidability of the exponential field is a famous open problem \cite{real_exp_field_decidability}. In contrast to the rational constants $t_0, T, c, K, M$, the variables $x, g, t, x_0$ in the proof rule are \emph{symbolic}.

\begin{proof}
The proof directly follows from \rref{lem: continuous dependence}. Pick $n \in \N$ large enough such that $2^{-n}(1 + T - t_0) \leq \eps$ is satisfied. Since $\Phi$ is a LDA, taking $k$ to be large enough such that $\norm{\Phi_k - \phi}_{\eval{C} \times [t_0, T]} \leq 2^{-n}(c - 1)$ and \rref{lem: continuous dependence} holds with $h = 2^{-n}$ gives the following:

\begin{enumerate}
    \item $\forall x_0 \in \eval{C} \norm{x_0 - \Phi_k(x_0, t_0)} \leq 2^{-n}(c - 1)$.
    \item The following differential invariant is valid for $h = 2^{-n}$ (thus provable in $\dL$ by \rref{thm:diff inv completeness}):
    \[\psi_{k}(x_0, x, g, t) \rightarrow \dbox{x' = f(x), g' = Kg, t' = 1 \& t \leq T \land B(x)}{\psi_{k}(x_0, x, g, t)}\]
\end{enumerate}

The desired proof in $\dL$ can now be constructed via the steps below by cutting in the differential invariant. First abbreviate
\[\alpha \equiv x' = f(x), g' = Kg, t' = 1 \& t \leq T \land B(x)\]

\begin{sequentdeduction}
    \linfer[implyr+DG+existsr]
    {\linfer[cut+implyl]
        {
            \linfer[qear]
                {\lclose}
            {\lsequent{C(x), x = x_0, g = c, t = t_0}{\psi_k(x_0, x, g, t)}}
            &
            \linfer[dinv]
                {\lclose}
            {\lsequent{\psi_{k}(x_0, x, g, t)}{\dbox{\alpha}{\psi_{k}(x_0, x, g, t)}}}
            &
            \linfer[]
                {}
            {\text{\textcircled{1}}}
        }
        {\lsequent{C(x), x = x_0, t = t_0, g = c}{\dbox{x' = f(x), g' = Kg, t' = 1 \& t \leq T \land B(x)}{\norm{x - \Phi_k(x_0, t)}^2 \leq M^2\eps^2}}}
    }
    {\lsequent{}{C(x) \land x = x_0 \land t = t_0 \rightarrow \dbox{x' = f(x), t' = 1 \& t \leq T \land B(x)}{\norm{x - \Phi_k(x_0, t)}^2 \leq M^2\eps^2}}}
\end{sequentdeduction}

The left premise closes by \irref{qear} from item (1), the second premise closes by \rref{lem: continuous dependence}, and the final remaining premise is:
\[\text{\textcircled{1}} \equiv \lsequent{C(x), x = x_0, t = t_0, g = c, \dbox{\alpha}{\psi_{k}(x_0, x, g, t)}}{\dbox{\alpha}{\norm{x - \Phi_k(x_0, t)}^2 \leq \eps^2M^2}}\]
Which can be handled with \irref{dW} and cutting in the bound $\dbox{\alpha}{g \leq M}$ with \irref{dC}. Crucially the application of \irref{DGi} to remove $x' = f(x)$ is sound since $x \notin K, M$.

\begin{sequentdeduction}
    \linfer[dC+dW]
    {\linfer[qear]
        {\lclose}
    {\lsequent{g \leq M, t \leq T, \psi_k(x_0, x, g, t)}{\norm{x - \Phi_k(x_0, t)}^2 \leq \eps^2M^2}}
    &
    \linfer[DGi]
        {\text{\textcircled{2}}}
    {\lsequent{g = c, t = t_0}{\dbox{\alpha}{g \leq M}}}
    }
    {\lsequent{x = x_0, t = t_0, g = c, \dbox{\alpha}{\psi_{k}(x_0, x, g, t)}}{\dbox{\alpha}{\norm{x - \Phi_k(x_0, t)}^2 \leq \eps^2M^2}}}
\end{sequentdeduction}

where the remaining premise on the right is
\[\text{\textcircled{2}} \equiv \lsequent{g = c, t = t_0}{\dbox{g' = Kg, t' = 1 \& t \leq T}{g \leq M}}\]
For the left premise, notice that the following is a valid formula of $\folr$, and therefore provable:
\[t \leq T \land g \leq M \land \psi_k(x, x_0, g, t) \rightarrow \norm{x - \Phi_k(x_0, t)}^2 \leq (2^{-n}(1 + T - t_0)M)^2\]
thus, the left premise closes by our choice of $n \in \N$. The proof of the desired formula has now been reduced to an upper bound on the exponential function (premise \textcircled{2}), completing the derivation. Since $k$ was only required to satisfy conditions (1), (2) and a satisfying witness for \rref{lem: continuous dependence} can be computed, such a $k$ can be computed as well, completing the proof of the theorem. 
\end{proof}

\begin{example}[Constrained exponential bound for Moore-Greitzer]
    \label{ex: jet engine mod exponential}
    \rref{thm: derivable lda} applies to the Moore-Greitzer jet engine model introduced in \rref{ex: jet engine} with its invariant established in \rref{ex: jet engine invariant}. We apply proof rule \irref{LDA_b} using
    \begin{itemize}
        \item $B(u, v) \equiv 0.781 < u < 1.109 \land 0.891 < v < 1.199 \land u + v < 2.25$
        \item $K = 8$
        \item $c = 1.1$
        \item $\eps = 4 \times 10^{-3} \times (1 + 0.02)$
        \item $M = 1.2$
    \end{itemize}
    Using these values, \irref{LDA_b} proves the following
    \begin{sequentdeduction}
        \linfer[LDA_b]
            {\lsequent{g = 1.1, t = 0}{\dbox{g' = 8g, t' = 1 \& t \leq 0.02}{g \leq 1.2}}}
        {\lsequent{\Delta(u_0, v_0),u = u_0,v = v_0,t = 0}{\dbox{(u', v') = f(u, v) \& t \leq 0.02 \land B(u, v)}{\norm{(u, v) - \Phi(u_0, v_0, t)}^2 < (5 \times 10^{-3})^2}}}
    \end{sequentdeduction}
    where the constant $5$ was chosen as 
    \[\eps M = 4 \times 1.2 \times 1.02 \times 10^{-3} < 5 \times 10^{-3}\]
    As the derivation shows, the proof rule \irref{LDA_b} reduced the problem of proving an error bound of $5 \times 10^{-3}$ to the problem of upper bounding exponentials on $[0, 0.02]$. Importantly, while all of the values were chosen numerically, the proof rule \irref{LDA_b} is derived. Therefore the validity of the formula is backed up by a corresponding \emph{syntactic proof}.
\end{example}

\rref{thm: derivable lda} still holds even if $(\Phi_k)_k$ does not converge in derivative, as long as it has bounded first and second time derivatives (implicitly requiring it to be twice differentiable), since \rref{lem: continuous dependence} still holds in this case per \rref{rem: bounded derivatives is enough for convergence}. 

\subsection{Provable Taylor Bounds on Exponentials}
\label{sec: Taylor upper bounds}

\rref{thm: derivable lda} reduced the proof of error bounds for LDAs to proving upper bounds for the exponential function on compact intervals. In this section, we show that $\dL$ is capable of proving arbitrarily accurate upper bounds on the exponential function via Taylor polynomials on the compact interval $[0, T]$.

\begin{proposition}[provable Taylor approximants]
    \label{thm: taylor upperbounds for exponential}
    Let $K, T \in \Q^+$ be rational constants. For all sufficiently large $n \in \N$, there is a syntactic term $\theta_n \in \Q[t]$ such that the following is a valid differential invariant
    \[g \leq \theta_n \rightarrow [g' = Kg, t' = 1 \& t \leq T]g \leq \theta_n\]
    Furthermore, $\theta_n \to e^{Kt}$ on $[-T, T]$ as $n \to \infty$ where $\theta_n$ is treated as a function in $t$. Finally, for all $n \in \N$ we have $\theta_n(0) = 1$ and $\theta_n$ can be computed uniformly in $K, T, n$.
\end{proposition}

\begin{proof}
    For $n \in \N$, let us denote $q_n(t)$ as the $n$-th Taylor approximant of $e^{Kt}$ i.e.
    \[q_n(t) = \sum_{i = 0}^n \frac{K^it^i}{i!}\]
    Let
    \[\theta_n(t) = q_n(t) + \frac{Mt^n}{n!}\qquad M = \frac{K^{n + 1}T}{n - KT}\]
    which is well-defined for all $n > KT$. By the Darboux inequality \cite[Corollary~3.2]{DBLP:journals/jacm/PlatzerT20}, the validity of the invariant follows from the validity of $(\theta_n(t))' \geq K\theta_n(t)$. 
    Computing $(\theta_n(t))'$ gives
    \[(\theta_n(t))' = Kq_{n - 1}(t) + \frac{Mt^{n - 1}}{(n - 1)!}\]
    So we have
    \begin{align*}
        (\theta_n(t))' - K\theta_n(t) &= \frac{Mt^{n - 1}}{(n - 1)!} - \frac{K^{n + 1}t^n}{n!} - \frac{KMt^n}{n!}\\
        &\geq \frac{t^{n - 1}}{n!}\left(nM - KMT - K^{n + 1}T\right) = \frac{t^{n - 1}}{n!}\left((n - KT)M - K^{n + 1}T\right) = 0
    \end{align*}
    Therefore the invariant is indeed valid for all $n > KT$. To witness the desired convergence, note 
    \[\frac{Mt^n}{n!} \xrightarrow{n \to \infty} 0\]
    and $q_n \xrightarrow{n \to \infty} e^{Kt}$ on $[-T, T]$ by Taylor's theorem. The proof is therefore complete.
\end{proof}

It now follows that $\dL$ is capable of proving arbitrarily accurate upper bounds on the exponential function on bounded intervals. 

\begin{corollary}[bounded exponentials]
    \label{cor: exponential bound}
    Let $c, K \in \Q^+$ be constants and $[t_0, T] \in \IQ$ be a rational interval. For all $M \in \Q^+$ that satisfy $ce^{K(T - t_0)} < M$, the following formula is provable in $\dL$:
    \[g = c \land t = t_0 \rightarrow \dbox{g' = Kg, t' = 1 \& t \leq T}{g \leq M}\]
\end{corollary}
\begin{proof}
    We first begin with standard reductions using \irref{DG} and \irref{dinv}, reducing the proof down to upper bounds on the standard exponential IVP $x' = Kx$ with initial condition $x(t_0) = 1$. 
    \begin{sequentdeduction}
        \linfer[implyr+DG+existsr]
        {\linfer[K]
            {
            \linfer[band+andr]
            {\linfer[dinv]
                {\lclose}
            {\lsequent{g = cx}{\dbox{g' = Kg, x' = Kx, t' = 1 \& t \leq T}{g = cx}}}
            &
            \linfer[DGi]
                {\linfer[]
                    {\text{\textcircled{1}}}
                {\lsequent{x = 1, t = t_0}{\dbox{x' = Kx, t' = 1 \& t \leq T}{x \leq \frac{M}{c}}}}
                }
            {\lsequent{x = 1, t = t_0}{\dbox{g' = Kg, x' = Kx, t' = 1 \& t \leq T}{x \leq \frac{M}{c}}}}    
            }
        {\lsequent{g = c, x = 1, t = t_0}{\dbox{g' = Kg, x' = Kx, t' = 1 \& t \leq T}{\left(x \leq \frac{M}{c} \land g = cx\right)}}}
            }
            {\lsequent{g = c, x = 1, t = t_0}{\dbox{g' = Kg, x' = Kx, t' = 1 \& t \leq T}{g \leq M}}}
        }
        {\lsequent{}{g = c \land t = t_0 \rightarrow \dbox{g' = Kg, t' = 1 \& t \leq T}{g \leq M}}}
    \end{sequentdeduction}
    Where the left premise closes as it is a valid differential invariant. \rref{thm: taylor upperbounds for exponential} now gives some $\theta(s) \in \Q[s]$ such that $c\norm{\theta}_{[0, T - t_0]} \leq M$, $\theta(0) = 1$, and the following differential invariant is valid:
    \[x \leq \theta(s) \rightarrow \dbox{x' = Kx, s' = 1 \& s \leq T - t_0}{x \leq \theta(s)}\]
    Note that this is only possible by our assumption of $ce^{K(T - t_0)} < M$. Premise \textcircled{1} can now be handled by cutting in this invariant on $\theta(s)$.

    \begin{sequentdeduction}
        \linfer[DG+existsr]{
            \linfer[dC+dinv]
                {
                    \linfer[dC]
                        {
                            \linfer[DGi]
                                {
                                    \linfer[cut+dinv+K]
                                    {
                                        \linfer[qear]
                                            {\lclose}
                                        {\lsequent{x = 1, s = 0, t= t_0}{x \leq \theta(s)}}
                                        &&
                                        \linfer[dC+dinv]
                                            {
                                                \linfer[dW]
                                                    {\linfer[qear]
                                                        {\lclose}
                                                    {\lsequent{s \geq 0, s \leq T - t_0}{c\theta(s) \leq M}}
                                                    }
                                                    {\lsequent{s = 0}{\dbox{x' = Kx, s' = 1 \& s \leq T - t_0 \land s \geq 0}{\left(x \leq \theta(s) \rightarrow x \leq \frac{M}{c}\right)}}}
                                            }
                                        {\lsequent{s = 0}{\dbox{x' = Kx, s' = 1 \& s \leq T - t_0}{\left(x \leq \theta(s) \rightarrow x \leq \frac{M}{c}\right)}}}
                                    }
                                {\lsequent{x = 1, s = 0, t = t_0}{\dbox{x' = Kx, s' = 1 \& s \leq T - t_0}{x \leq \frac{M}{c}}}}
                                }
                            {\lsequent{x = 1, s = 0, t = t_0}{\dbox{x' = Kx, s' = 1, t' = 1 \& s \leq T - t_0}{x \leq \frac{M}{c}}}}
                        }
                    {\lsequent{x = 1, s = 0, t = t_0}{\dbox{x' = Kx, s' = 1, t' = 1 \& t \leq T \land s = t - t_0}{x \leq \frac{M}{c}}}}
                }
            {\lsequent{x = 1, s = 0, t = t_0}{\dbox{x' = Kx, s' = 1, t' = 1 \& t \leq T}{x \leq \frac{M}{c}}}}
        }
        {\lsequent{x = 1, t = t_0}{\dbox{x' = Kx, t' = 1 \& t \leq T}{x \leq \frac{M}{c}}}}
    \end{sequentdeduction}
Where the left premise closes as $\theta(0) = 1$, and the right premise closes since $\theta$ was constructed to satisfy $c\norm{\theta}_{[0, T - t_0]} \leq M$, which is therefore provable by \irref{qear}. This completes the proof.
\end{proof}

Chaining up the results of \rref{thm: derivable lda} and \rref{thm: taylor upperbounds for exponential} give complete proofs for accuracy bounds of LDAs for compact IVPs. This has many important consequences regarding the proof theory of $\dL$ which are listed below. The first of which says that for any LDA, for any desired accuracy, one can derive a proof certifying this accuracy within $\dL$ assuming the presence of some domain constraint. 

\begin{theorem}[completeness for LDAs with domain constraints]
    \label{thm: effective proofs without upper bounds}
    Let $(f(x), C(x), [t_0, T])$ be a compact IVP, $\Phi$ a LDA and $B(x)$ a $\folr$ formula. Assume that the following holds:

\begin{enumerate}
    \item The flow $\phi(x, t)$ of the compact IVP is well defined on $\eval{C} \times [t_0, T]$. 
    \item $\eval{B} \subset \R^n$ is a bounded set containing $\phi(\eval{C}, [t_0, T])$. 
\end{enumerate}

Then for all $\eps \in \Q^+$, for all sufficiently large $k \in \N$, the following formula is provable in $\dL$.
\[C(x) \land x = x_0 \land t = t_0 \rightarrow \dbox{x' = f(x), t' = 1 \& t \leq T \land B(x)}{\norm{x - \Phi_k(x_0, t)}^2 \leq \eps^2}\]

For each $\eps \in \Q^+$ a corresponding $k$ can be computed uniformly from the compact IVP, $\Phi$ and $\eps$.  
\end{theorem}

\begin{proof}
    Follows directly via \rref{thm: derivable lda} and \rref{cor: exponential bound}.
\end{proof}

\begin{example}[Constrained bound for Moore-Greitzer]
    \label{ex: jet engine constrained}
    Following \rref{thm: taylor upperbounds for exponential}, we prove
    \[\lsequent{g = 1.1, t = 0}{\dbox{g' = 8g, t' = 1 \& t \leq 0.02}{g \leq 1.2}}\]
    This derivation combined with \rref{ex: jet engine mod exponential} proves the validity of the following formula
    \[\linferenceRule[impll2]
        {\Delta(u_0, v_0) \land t = 0 \land u = u_0 \land v = v_0}
         {\dbox{(u', v') = f(u, v), t' = 1 \& t \leq 0.02 \land B(u, v)}{\norm{(u, v) - \Phi(u_0, v_0, t)}^2 < (0.005)^2}}
    \]
    which is a particular instance of \rref{thm: effective proofs without upper bounds}, syntactically proving an error bound of $0.005$ for the approximation $\Phi(u_0, v_0, t)$ under the assumption of the domain constraint $B(u, v)$. 
\end{example}

The following result syntactically proves the classical Stone–Weierstra{\ss} theorem in $\dL$ for flows of compact IVPs under the assumption of some bounded domain constraint. 

\begin{theorem}[Weierstra{\ss} approximation with domain constraints]
    \label{thm: stone weierstrass with domain constraint}
    Let $(f(x), C(x), [t_0, T])$ be a compact IVP and $B(x)$ a $\folr$ formula. Assume that the following holds:

\begin{enumerate}
    \item The flow $\phi(x, t)$ of the compact IVP is well defined on $\eval{C} \times [t_0, T]$. 
    \item $\eval{B} \subset \R^n$ is a bounded set containing $\phi(\eval{C}, [t_0, T])$.
\end{enumerate}

Then there is a computable sequence $(\theta_k)_k \in \Q^n[x_0, t]$ of approximants such that the following formulas are provable for all $k \in \N$:
\[C(x) \land x = x_0 \land t = t_0 \rightarrow \dbox{x' = f(x), t' = 1 \& t \leq T \land B(x)}{\norm{x - \theta_k(x_0, t)}^2 \leq 2^{-2k}}\]
\end{theorem}

\begin{proof}
    Follows directly from \rref{thm: effective proofs without upper bounds} and \rref{thm: abstract solvers always exist}.
\end{proof}

\rref{thm: effective proofs without upper bounds} and \rref{thm: stone weierstrass with domain constraint} proves that under the presence of some bounded domain constraint, flows of compact IVPs can be arbitrarily approximated by polynomials with provably accurate error bounds. As such, one can \emph{always} prove desired (open) properties of flows of compact IVPs by transferring to the case of polynomials, where the properties can then be proven by quantifier elimination with the proof rule \irref{qear}. The remaining sections handle the case where such domain constraints are not assumed to exist \emph{a priori}.

\section{Proving Domain Constraints and Bounded Completeness}
\label{sec: proving constraint}
A key assumption in the previous section is the existence of a $\folr$ formula $B(x)$ that bounds the evolution of the flow induced by the ODE, acting as a domain constraint. Such an assumption was a natural consequence of the fact that non-linear polynomial vector fields are only locally Lipschitz, and therefore some \emph{a priori} bound on the flow is required in order to computably utilize the continuity of the flow. In this section, we will first show how to eliminate such assumptions by proving them directly for compact IVPs and obtain a stronger version of \rref{thm: effective proofs without upper bounds}. Utilizing this, we prove that $\dL$'s axiomatization \cite{DBLP:journals/jar/Platzer17, DBLP:journals/jacm/PlatzerT20, DBLP:journals/fac/TanP21} enjoys completeness properties over compact time horizons without assuming bounded domain constraints. And finally, we discuss methods of handling domain constraints symbolically. Along the way, the syntactic provability of several axioms within $\dL$ that synthesize fundamental mathematical properties of ODEs is established, which are of independent interest.\\

\subsection{Error Bounds Without Domain Constraints}
Our main goal is the following strengthening of \rref{thm: effective proofs without upper bounds}, which \emph{does not} assume the existence of a bounded domain constraint. 

\begin{theorem}[Completeness for LDAs]
\label{thm: effective proofs without upper bounds and constraints}
Let $(f(x), C(x), [t_0, T])$ be a compact IVP with a well-defined flow $\phi : \eval{C} \times [t_0, T] \to \R^n$ and $\Phi$ a LDA. Then for all $\eps \in \Q^+$, for all sufficiently large $k \in \N$, the following formula is provable in $\dL$.
\[C(x) \land x = x_0 \land t = t_0 \rightarrow \dbox{x' = f(x), t' = 1 \& t \leq T}{\norm{x - \Phi_k(x_0, t)}^2 < \eps^2}\]

Furthermore, a satisfying $k$ can be computed uniformly from the compact IVP, $\Phi$ and $\eps$.
\end{theorem}

\begin{remark}
    \rref{thm: effective proofs without upper bounds and constraints} can be understood as a ``completeness for convergence of LDAs'' result. In the sense that if a sequence of definable functions $(g_k)_k : \eval{C} \times [t_0, T] \to \R^n$ converges to the true flow $\phi : \eval{C} \times [t_0, T] \to \R^n$ in the $C^1$ norm, then their convergence in the $C^0$ norm can be syntactically proven in $\dL$. While this result assumes the existence of the flow for a sufficient duration, \rref{thm: finite completeness - existence} shows that $\dL$ is complete for such existence properties as well. \rref{cor: completeness for convergence} further strengthens this theorem by weakening the assumption to $C^0$ convergence instead of $C^1$ convergence. 
\end{remark}

To prove \rref{thm: effective proofs without upper bounds and constraints}, the following lemma is needed. 

\begin{lemma}[completeness for bounded flows]
    \label{lem: finite completeness for open balls}
    Let $(f(x), C(x), [t_0, T])$ be a compact IVP, $\Phi$ a LDA and $R \in \Q^+$. Assume that the following holds: 

    \begin{enumerate}
        \item The flow $\phi(x, t)$ of the compact IVP is well-defined on $\eval{C} \times [t_0, T]$.
        \item $\phi(\eval{C}, [t_0, T]) \subseteq B(0, R)$, where $B(0, R)$ is the open ball of radius $R$ in $\R^n$.
    \end{enumerate}
    Then the following formula is provable in $\dL$.
    \[C(x) \land t = t_0 \rightarrow \dbox{x' = f(x), t' =1 \& t \leq T}{\norm{x}^2 < R^2}\]
\end{lemma}

\begin{proof}
    First note that rule \irref{enclosure} reduces the problem to:
    \begin{sequentdeduction}
        \linfer[enclosure]
        {\lsequent{C(x) \land t = t_0}{\dbox{x' = f(x), t' =1 \& t \leq T \land \norm{x}^2 \leq R^2}{\norm{x}^2 < R^2}}}
        {\lsequent{C(x) \land t = t_0}{\dbox{x' = f(x), t' =1 \& t \leq T}{\norm{x}^2 < R^2}}}
    \end{sequentdeduction}
    By \rref{thm: abstract solvers always exist}, we may compute some LDA $\Phi$ for this compact IVP. Now do an a priori unbounded search on pairs $(\eps, k) \in \Q^+ \times \N$ such that the following formulas are provable in $\dL$.
    \begin{align*}
        &C(x) \land x = x_0 \land t = t_0 \rightarrow \dbox{x' = f(x), t' = 1 \& t \leq T \land \norm{x}^2 \leq R^2}{\norm{x - \Phi_k(x_0, t)}^2 \leq \frac{\eps}}{2}\\
        &\forall x_0\forall t \left(t_0 \leq t \land t \leq T \land C(x_0) \rightarrow \norm{\Phi_k(x_0, t)}^2 < R^2 - \frac{\eps}{2}\right)
    \end{align*}
    In fact, such pairs necessarily exist and the search is bounded. To see this, note that $\phi(\eval{C}, [t_0, T])$ is a compact subset of the open set $B(0, R)$ by assumption, so there exists some $\eps \in \Q^+$ such that $B(\phi(\eval{C}, [t_0, T]), \eps) \subseteq B(0, R)$. By choosing this $\eps$ and $k \in \N$ sufficiently large, the first formula will be valid and therefore provable by \rref{thm: effective proofs without upper bounds}. The second formula is true and therefore provable by \irref{qear} for all sufficiently large $k \in \N$ since $\Phi$ is a LDA. Hence, we can computably find a pair $(\eps, k)$ with corresponding proofs to the formulas above. Now applying axiom \irref{V} (\rref{lem: vacuous axioms}) and \irref{dW} shows that $t_0 \leq t \leq T$ and $C(x_0)$ are always satisfied during the evolution of the ODE in the first formula. As such, applying these axioms on the formulas together with \irref{V} proves 
    \[\linferenceRule[impll2]
         {C(x) \land x = x_0 \land t = t_0}
         {\dbox{x' = f(x), t' = 1 \& t \leq T \land \norm{x}^2 \leq R^2}{\left(\norm{x - \Phi_k(x_0, t)}^2 \leq \frac{\eps}{2} \land \norm{\Phi_k(x_0, t)}^2 < R^2 - \frac{\eps}{2}\right)}}
    \]
    from which the remaining premise introduced by \irref{enclosure} follows.
\end{proof}

\rref{thm: effective proofs without upper bounds and constraints} can now be proven using \rref{lem: finite completeness for open balls} and \rref{thm: effective proofs without upper bounds}.

\begin{proof}[Proof of \rref{thm: effective proofs without upper bounds and constraints}]
First note that for any positive rational $R \in \Q^+$, cutting in the domain constraint $\norm{x}^2 < R^2$ with \irref{dC} (and applications of \irref{existsl} to introduce the variable $x_0$) reduces the proof obligation to proving the following premises:
\begin{align*}
    &\lsequent{}{C(x) \land x = x_0 \land t = t_0 \rightarrow \dbox{x' = f(x), t' = 1 \& t \leq T \land \norm{x}^2 < R^2}{\norm{x - \Phi_k(x_0, t)}^2 < \eps^2}}\\
    &\lsequent{}{C(x) \land t = t_0 \rightarrow \dbox{x' = f(x), t' = 1 \& t \leq T}{\norm{x}^2 < R^2}}
\end{align*}
Hence, we may do a bounded search on the pair $(R, k) \in \Q^+ \times \N$ such that the above are provable. This is a bounded search since $\phi(\eval{C}, [t_0, T])$ is a compact set, so for all sufficiently large $R$ we have $\phi(\eval{C}, [t_0, T]) \subseteq B(0, R)$, from which the provability of the two premises follows from \rref{thm: effective proofs without upper bounds} and \rref{lem: finite completeness for open balls} respectively. Furthermore, this is a computable search as \rref{thm: effective proofs without upper bounds} and \rref{lem: finite completeness for open balls} both hold computably. Once such a pair $(R, k)$ has been found with corresponding proofs, the premises are proven and therefore the proof is complete by applying axiom \irref{dC}.
\end{proof}

\rref{thm: effective proofs without upper bounds and constraints} proves that for all compact IVPs, for all corresponding LDAs, for all $\eps \in \Q^+$, one can find some corresponding proof in $\dL$ certifying the LDA to be at most $\eps$ away from the true solution. The following example applies this theorem to Moore-Greitzer's model of jet engines.

\begin{example}[Unconstrained bound for Moore-Greitzer]
    \label{ex: jet engine bound}
    \rref{ex: jet engine constrained} proved an error bound of $0.005$ under the assumption of a domain constraint $B(u, v)$ given by 
    \[B(u, v) \equiv 0.781 < u < 1.109 \land 0.891 < v < 1.199 \land u + v < 2.25\]
    Applying \rref{thm: effective proofs without upper bounds and constraints} and utilizing the constrained bound proven in \rref{ex: jet engine constrained} then proves an error bound of $0.005$ without assuming domain constraints.
    \[\Delta(u_0, v_0, t) \land t = 0\land u = u_0\land v = v_0 \rightarrow\dbox{(u', v') = f(u, v), t' = 1 \& t \leq 0.02}{\norm{(u, v) - \Phi(u_0, v_0, t)}^2 < 0.005^2}\]
    As such, we have syntactically proven the accuracy of a numerical approximation using deductive logic reasoning.
\end{example}

The following theorem proves a version of the Stone-Weierstra{\ss} theorem (\rref{thm: stone weierstrass with domain constraint}) without domain constraints.

\begin{theorem}[Stone-Weierstra{\ss}]
    \label{thm: stone weierstrass}
    Let $(f(x), C(x), [t_0, T])$ be a compact IVP with well-defined flow $\phi(x, t) : \eval{C} \times [t_0, T] \to \R^n$. Then there is a computable sequence $(\theta_k)_k \in \Q^n[x_0, t]$ of approximants such that the following formulas are provable for all $k \in \N$:
\[C(x) \land x = x_0 \land t = t_0 \rightarrow \dbox{x' = f(x), t' = 1 \& t \leq T}{\norm{x - \theta_k(x_0, t)}^2 \leq 2^{-2k}}\]
\end{theorem}

\begin{proof}
    Follows directly by \rref{thm: abstract solvers always exist} and \rref{thm: effective proofs without upper bounds and constraints}. 
\end{proof}

\rref{thm: effective proofs without upper bounds and constraints} can also be viewed as a ``completeness for convergence'' result that requires $C^1$ convergence and proves $C^0$ convergence. By utilizing \rref{thm: abstract solvers always exist} to provably compute a correct LDA, it is possible to strengthen \rref{thm: effective proofs without upper bounds and constraints} and only require $C^0$ convergence. 

\begin{corollary}[Completeness for convergence]
    \label{cor: completeness for convergence}
     Let $(f(x), C(x), [t_0, T])$ be a compact IVP with well-defined flow $\phi : \eval{C} \times [t_0, T] \to \R^n$. Further suppose that $(f_k)_k$ is a sequence of $\folr$ definable functions with $f_k : \eval{C} \times [t_0, T] \to \R^n$. Then $\dL$ is complete for convergence:
     \[\entails (f_k)_k \xrightarrow{n \to \infty} \phi \qquad \implies \qquad \infers(f_k)_k \xrightarrow{n \to \infty} \phi\]
     I.e. if $(f_k)_k$ converges to $\phi$ in $C^0(\eval{C} \times [t_0, T], \R^n)$, then for every $\eps \in \Q^+$, for all sufficiently large $k \in \N$, the following formula is provable
     \[C(x) \land x = x_0 \land t = t_0 \rightarrow \dbox{x' = f(x), t' = 1 \& t \leq T}{\norm{f_k(x_0, t) - x}^2 < \eps^2}\]
     Furthermore, a satisfying $k$ can be computed uniformly from the compact IVP, $\eps$ and $(f_k)_k$. 
\end{corollary}

\begin{proof}
    Let $\Phi$ be some LDA of the compact IVP computed by \rref{thm: abstract solvers always exist}, $\eps \in \Q^+$ be the desired accuracy. Let $k$ be large enough such that \rref{thm: effective proofs without upper bounds and constraints} holds with an accuracy of $\frac{\eps}{3}$ and $\norm{f_{k} - \phi}_{\eval{C} \times [t_0, T]} < \frac{\eps}{3}$ is satisfied. It suffices to show that the formula 
    \[C(x) \land x = x_0 \land t = t_0 \rightarrow \dbox{x' = f(x), t' = 1 \& t \leq T}{\norm{f_k(x_0, t) - x}^2 < \eps^2}\]
    is provable. Indeed, \rref{thm: effective proofs without upper bounds and constraints} and the choice of $k$ imply that the following formula is provable
    \[C(x) \land x = x_0 \land t = t_0 \rightarrow \dbox{x' = f(x), t' = 1 \& t \leq T}{\norm{\Phi_k(x_0, t) - x}^2 < \left(\frac{\eps}{3}\right)^2}\]
    By construction and \irref{qear}, it is also provable that $\norm{\Phi_k - f_k}_{\eval{C} \times [t_0, T]} < \frac{2\eps}{3}$. Hence an application of axiom \irref{K} implies that the following is provable, completing the proof.
    \[C(x) \land x = x_0 \land t = t_0 \rightarrow \dbox{x' = f(x), t' = 1 \& t \leq T}{\norm{f_k(x_0, t) - x}^2 < \left(\frac{\eps}{3} + \frac{2\eps}{3}\right)^2}\]
\end{proof}

\subsection{Symbolic Domain Constraints and Completeness on Compact Time Horizons}
\label{sec: symbolic axioms}
This section establishes completeness properties of $\dL$ over compact time horizons for compact IVPs. The main proof strategy is to utilize our results in previous sections which show that $\dL$ is complete for LDAs of compact IVPs, thereby reducing properties of such compact IVPs to decidable sentences in real arithmetic. This section also explores to what extent such results can be applied to IVPs with symbolic initial conditions that are not constrained to compact sets. The main technical results can be encapsulated in the following theorem, which asserts the provability of various axioms and proof rules in $\dL$.

\begin{theorem}
    \label{thm:axioms for step liveness}
    The following axioms and rules are syntactically derivable in $\dL$, thus sound. Where $M, R > 0$ are symbolic variables and $B(x)$ is a $\folr$ formula characterizing a bounded set.

    \begin{calculus}

        \cinferenceRule[dualRight|StepDual${}_{\rightarrow}$]
        {Duality with bounded time}
        {
        t \leq \tau \land \dbox{x' = f(x), t' = 1 \& t \leq \tau}{B(x)} \rightarrow \ddiamond{x' = f(x), t' = 1\& B(x)}{t = \tau}
        }{}

        \cinferenceRule[dualLeft|StepDual${}_{\leftarrow}$]
        {Duality with bounded time}
        {
        \ddiamond{x' = f(x), t' = 1\& Q}{t \geq \tau} \rightarrow \dbox{x' = f(x), t' = 1 \&t \leq \tau}{Q}
        }{}

        \cinferenceRule[stepEx|StepEx]
        {Step existence from Picard Lindelof}
        {
         \linferenceRule[impll2]
         {\forall y \left(y \in B[x_0, R] \rightarrow \norm{f(y)}^2 \leq M^2\right)}
         {\left(x = x_0 \land t = t_0 \rightarrow \ddiamond{x' = f(x), t' = 1 \& x \in B[x_0, R]}{t \geq t_0 + \frac{R}{M}}\right)}
        }{}
        
        \cinferenceRule[stepExt|StepExt]{extending solutions}
       {
       \linferenceRule[tableaux]
       {&\lsequent{t = t_0, P_1} {\Gamma_2}\\
       &\lsequent{\Gamma_1} {[x' = f(x), t' = 1 \& t \leq t_0]P_1\newline} \\
       &\lsequent{\Gamma_2} {[x' = f(x), t' = 1 \& t \leq t_0 + t_1]P_2\newline} 
       }
       {\lsequent{t \leq t_0, \Gamma_1} {[x' = f(x), t' = 1 \& t \leq t_0 + t_1]((t \leq t_0 \rightarrow P_1) \land (t > t_0 \rightarrow P_2))}}
       }{}
    \end{calculus}
\end{theorem}

\begin{remark}
    These axioms are capable of symbolically simulating a basic algorithm for certifying existence of ODEs, which essentially mimics the classical proof \cite{Hartman_2002}, such an algorithm has also been presented explicitly in more recent work \cite{Graca_Zhong_Buescu_2009}. \rref{ex: symbolic existence} shows how this can be done. 
\end{remark}
The following provides some intuitive explanation for the axioms/proof rules in \rref{thm:axioms for step liveness}. 
    
\begin{itemize}
    \item Axioms \irref{dualRight}, \irref{dualLeft} provide a duality between box and diamond modalities on compact time horizons for ODEs. These axioms are useful in proving that the flow is bounded within some bounded set over a fixed time interval. It is also worth noting that while axiom \irref{dualRight} requires a bounded set, axiom \irref{dualLeft} places no requirements on the domain constraint $Q$ as it follows from the uniqueness of flows for ODEs. 
    \item Axiom \irref{stepEx} is a quantitative version of the classical Picard-Lindel\"of theorem presented in the language of $\dL$, allowing one to symbolically prove that the solution exists for a duration of $\frac{R}{M}$, which is a lower-bound on how long it takes for the solution to escape the ball $B[x_0, R]$. 
    
    \item Proof rule \irref{stepExt} provides a way of concatenating information proven for different time steps together over the entire time step. Similar to the proof of computability of solutions to IVPs \cite{Graca_Zhong_Buescu_2009} which iteratively chains up Picard iterations at various time steps. 
\end{itemize} 
All of the above axioms/proof rules are syntactically derivable using just $\dL$'s axiomatization \cite{DBLP:journals/jacm/PlatzerT20, DBLP:journals/fac/TanP21}. It is important to note that the axioms in \rref{thm:axioms for step liveness} hold \emph{symbolically} and are \emph{not} limited to compact IVPs (e.g. $x_0,\tau, t, T, M, R$ are symbolic variables). 

In order to prove \rref{thm:axioms for step liveness}, the following lemma is needed, which establishes the provability of many fundamental properties of ODEs, and is therefore of independent interest. 
\begin{lemma}
    \label{lem:lemma axioms}
    The following axioms are derivable in $\dL$.
    
    \begin{calculus}
    \cinferenceRule[thereAndBack|Rev]{Derived version of there and back}
   {
   \linferenceRule[impl]
   {{P}}
   {{\dbox{x' = f(x) \& Q}{\ddiamond{x' = -f(x) \& Q}{P}}}}
   }{}
    \cinferenceRule[timeFixpoint|Stuck]{t' = 1 does not have fixed points}
   {
   \linferenceRule[impl]
   {{t = t_0}}
   {\left({\dbox{x' = f(x), t' = 1 \& t \leq t_0}{P} \leftrightarrow P}\right)}
   }{}
    \cinferenceRule[idempotence|Idem]{<x' = f(x)>P -> <x' = f(x) \& <x' = f(x)>P>P}
   {
   \linferenceRule[impl]
   {{\ddiamond{x' = f(x) \& Q}{P}}}
   {{\ddiamond{x' = f(x) \& Q \land \ddiamond{x' = f(x) \& Q}{P}}{P}}}
   }{}
    \cinferenceRule[uniqp|Uniq']{Fine-grained uniqueness}
   {
   \linferenceRule[impll2]
   {\ddiamond{x' = f(x) \& Q_1}{P_1} \land \ddiamond{x' = f(x) \& Q_2}{P_2}}
   {\ddiamond{x' = f(x) \& Q_1 \land Q_2}{\left(P_1 \land \ddiamond{x' = f(x) \& Q_2}{P_2}\right)} \lor \ddiamond{x' = f(x) \& Q_1 \land Q_2}{\left(P_2 \land \ddiamond{x' = f(x) \& Q_1}{P_1}\right)}}
   }{}

   \cinferenceRule[ivt|IVT]{intermediate value theorem}
   {
   \linferenceRule[impl]
   {e \leq 0 \land \ddiamond{x' = f(x), t' = 1 \& Q}{\left(t = \tau \land e > 0\right)}}
   {\ddiamond{x' = f(x), t' = 1 \& Q \land t < \tau \land e \leq 0}{e = 0}}
   }
   {}
    \end{calculus}
\end{lemma}

While \rref{lem:lemma axioms}'s purpose in this article is solely to prove \rref{thm:axioms for step liveness}, they also convey helpful properties of ODEs that are useful for other purposes. The following provides some intuition for these axioms.

\begin{itemize}
    \item Axiom \irref{thereAndBack} says that if a property $P$ is true, then after flowing along some ODE one can always flow back to a state where $P$ is true. A sort of ``there and back'' quantification that says the current state can always be reached by reversing the ODE flow. This axiom (and its proof) has already been implemented in \KeYmaeraX's tactics library, but we reproduce a proof here for completeness.

    \item Axiom \irref{timeFixpoint} expresses that the ODE $t' = 1$ is strictly monotone, and therefore does not have any fixed points. Thus, if the current state has $t = t_0$ and the domain constraint includes $t \leq t_0$, then the overall dynamical system is stuck and necessarily cannot evolve, resulting in the RHS of the axiom.

    \item Axiom \irref{idempotence} expresses an ``idempotence'' property of diamond modalities. If the current state can flow along some ODE to a target region, then every state along this flow can also flow to the target region. One can also view this as a statement on the uniqueness of flows \cite{DBLP:journals/jacm/PlatzerT20}.

    \item Axiom \irref{uniqp} is a more fine-grained version of $\dL$'s uniqueness axiom \cite{DBLP:journals/jacm/PlatzerT20} that deals with two potentially distinct target regions. While the implication may look complicated, it is just saying that if the flow along the same ODE can reach two regions $P_1, P_2$ under the domain constraints $Q_1, Q_2$ respectively, then by uniqueness of flows one flow will be the prefix of the other.

    \item Axiom \irref{ivt} internalizes the classical intermediate value theorem within $\dL$, saying if the term $e$ is initially non-positive and becomes positive along some flow, then it necessarily reaches $e = 0$ along the way and will do so while remaining in $e \leq 0$. 
\end{itemize}
Proofs of \rref{lem:lemma axioms} and \rref{thm:axioms for step liveness} are provided in \rref{app: derivations}.

The main use of \rref{thm:axioms for step liveness} in this article is to establish completeness results for compact IVPs. However, as the axioms/proof rules in \rref{thm:axioms for step liveness} are fully symbolic, they also enable deductive reasoning for general symbolic IVPs which is of independent interest, one such example is given below.

\begin{example}[Symbolic maximal interval of existence]
    \label{ex: symbolic existence}
    Consider the simple uni-variate ODE $x' = x^2 + 1$ with symbolic initial condition $x(0) = x_0$. Its exact solution is
    \[x(t) \equiv \tan(\arctan(x_0) + t)\]
    Thus, the (right) maximal interval of existence of the corresponding solution is $[0, \frac{\pi}{2} - \arctan(x_0))$, note that $x_0$ is a \emph{symbolic variable} rather than a fixed constant. This example shows how $\dL$ can essentially prove this symbolic interval of existence. Of course, since $\arctan(x_0)$ is not expressible in vanilla $\dL$, we resort to the following asymptotic approximation of $\arctan(x)$ (and $\frac{\pi}{2} - \arctan(x)$) obtained via its series expansion at infinity
    \begin{align*}
        \arctan(x) &\sim \frac{\pi}{2} - \frac{1}{x} + \frac{1}{3x^3} + o\left(\frac{1}{x^4}\right)\\
        \frac{\pi}{2} - \arctan(x) &\sim \frac{1}{x} - \frac{1}{3x^3} + o\left(\frac{1}{x^4}\right)
    \end{align*}
    With \rref{thm:axioms for step liveness}, it can be shown that for all $\eps \in \Q^+$ however small, the following formula (parametrized by $\eps$) is derivable in $\dL$\footnote{Since $x_0 \neq 0$ is enforced, the value $\frac{1}{x_0}$ is defined uniquely as some $c$ such that $cx_0 = 1$.} (the same technique in this example also works for higher-order bounds)
    \[x = x_0 \land t = 0 \land x_0 > 0 \rightarrow \ddiamond{x' = x^2 + 1, t' = 1}{t \geq (1 - \eps)\left(\frac{1}{x_0} - \frac{1}{3x_0^3}\right)}\]
    In other words, for every $\eps > 0$, one can \emph{symbolically prove} that the (right) maximal interval of existence is at least $(1 - \eps)\left(\frac{1}{x} - \frac{1}{3x^3}\right)$. Importantly, such a bound is interesting because $x_0$ is symbolic and can be unbounded, hence the provability of this formula does not directly follow from the completeness results for compact IVPs. Indeed, for $x_0$ sufficiently small the bound tends to $-\infty$, which is trivially satisfied. The assumption of $x_0 > 0$ is added for clarity in derivations only, and an identical formula can also be derived for $x_0 < 0$. 
\end{example}
\begin{proof}
    A complete proof is provided in \rref{app: derivations}. The main idea is to derive a numerical approximation purely symbolically using \irref{stepEx}. For a symbolic initial value $x_0$, bounding the maximum derivative in $B(x_0, x_0)$ gives some positive duration of existence. Running this procedure iteratively for $n$ steps gives rise to $n$ such values, adding these up with axiom \irref{stepExt} gives a lower-bound on the duration of existence while remaining in the region $B(x_0, nx_0)$. By picking $n \in \N$ large enough (independent of $x_0$), this procedure proves the desired lower-bound. 
\end{proof}

With \rref{thm:axioms for step liveness}, various completeness properties of $\dL$ for compact IVPs can now be proven. 

\begin{theorem}[Completeness for bounded safety]
    \label{thm: finite completeness - safety}
    Let $(f(x), C(x), [t_0, T])$ be a compact IVP and $O(x)$ a $\folr$ formula characterizing a bounded open set. Then $\dL$ is complete for formulas of the form 
    \[C(x) \land t = t_0 \rightarrow \dbox{x' = f(x), t' = 1 \& t \leq T}{O(x)}\]
    I.e. the following equivalence holds
    \begin{align*}
        &\models C(x) \land t = t_0 \rightarrow \dbox{x' = f(x), t' = 1 \& t \leq T}{O(x)} \iff\\
        &\lsequent{}{C(x) \land t = t_0 \rightarrow \dbox{x' = f(x), t' = 1 \& t \leq T}{O(x)}}
    \end{align*}
\end{theorem}

\begin{proof}
    The $\impliedby$ implication is soundness and follows by soundness of $\dL$'s axiomatization \cite{DBLP:journals/jar/Platzer08, DBLP:journals/jacm/PlatzerT20, DBLP:journals/fac/TanP21}, so it remains to prove the $\implies$ implication. To this end, let us assume the validity of such a formula. Since $O(x)$ is a bounded set, this implies that the flow $\phi : \eval{C} \times [t_0, T] \to \R^n$ of the compact IVP is well-defined as it does not exhibit finite time blow-up. By validity of the formula, we have $\phi(\eval{C}, [t_0, T]) \subseteq O(x)$. Since $O(x)$ is open and $\phi(\eval{C}, [t_0, T])$ is compact, there necessarily exists some $\eps \in \Q^+$ such that $B(\phi(\eval{C}, [t_0, T]), \eps) \subseteq \eval{O}$. For each $n \in \N$, denote by $\theta_n$ the (vectorial) polynomial of error at most $2^{-n}$ as computed by \rref{thm: stone weierstrass}. Now note that for all sufficiently large $n \in \N$, the following formulas will be valid
    \begin{align*}
        &C(x) \land x = x_0 \land t = t_0 \rightarrow \dbox{x' = f(x), t' = 1 \& t \leq T}{\norm{x - \theta_n(x_0, t)}^2 \leq 2^{-2n}}\\
        &\forall x_0\forall t (C(x_0) \land t_0 \leq t \land t \leq T \rightarrow B[\theta_n(x_0, t), 2^{-n}] \subseteq \eval{O})
    \end{align*}
    Furthermore, they are both provable via \rref{thm: stone weierstrass} and \irref{qear} respectively. Thus, doing a bounded search on $n \in \N$ will find one where the two formulas above are provable. From this applications of \irref{V+dW} on the first formula proves
    \[C(x) \land x = x_0 \land t = t_0 \rightarrow \dbox{x' = f(x), t' = 1 \& t \leq T}{\left(t_0 \leq t \land t \leq T \land C(x_0) \land \norm{x - \theta_n(x_0, t)}^2 \leq 2^{-2n}\right)}\]
    Another application of \irref{V} (\rref{lem: vacuous axioms}) brings the second formula in, proving
    \[C(x) \land x = x_0 \land t = t_0 \rightarrow \dbox{x' = f(x), t' = 1 \& t \leq T}{\left(B[\theta_n(x_0, t), 2^{-n}] \subseteq \eval{O} \land \norm{x - \theta_n(x_0, t)}^2 \leq 2^{-2n}\right)}\]
    The desired formula of $C(x) \land t = t_0 \rightarrow \dbox{x' = f(x), t' = 1 \& t \leq T}{O(x)}$ then follows by applying \irref{K} and \irref{qear}, completing the proof. 
\end{proof}

\begin{theorem}[Completeness for bounded existence]
    \label{thm: finite completeness - existence}
    Let $(f(x), C(x), [t_0, T])$ be a compact IVP. $\dL$ is complete for formulas of the form 
    \[C(x) \land t = t_0 \rightarrow \ddiamond{x' = f(x), t' = 1}{t \geq T}\]
    Where $T \in \Q^+$ is a rational constant. I.e. the following equivalence holds
    \begin{align*}
        &\models C(x) \land t = t_0 \rightarrow \ddiamond{x' = f(x), t' = 1}{t \geq T} \iff\\
        &\lsequent{}{C(x) \land t = t_0 \rightarrow \ddiamond{x' = f(x), t' = 1}{t \geq T}}
    \end{align*}

\end{theorem}

\begin{proof}
    Again $\impliedby$ follows from $\dL$'s soundness \cite{DBLP:journals/jar/Platzer17, DBLP:journals/jacm/PlatzerT20, DBLP:journals/fac/TanP21}, so it suffices to prove $\implies$. Assuming that such a formula is valid, the flow $\phi : \eval{C} \times [t_0, T] \to \R^n$ of the compact IVP is necessarily well-defined and therefore does not exhibit finite time blow-up on the time interval $[t_0, T]$. Thus, for all sufficiently large $R \in \Q^+$, the following formula will be valid
    \[C(x) \land t = t_0 \rightarrow \dbox{x' = f(x), t' = 1 \& t \leq T}{\norm{x}^2 < R^2}\]
    By \rref{thm: finite completeness - safety}, this will furthermore be provable in $\dL$ because $\norm{x}^2 < R^2$ is open. Thus, we may do a search for $R \in \Q^+$ until we find a value for which the formula above is provable. Once such a value is found, the desired formula can be proven via the following derivation
    \begin{sequentdeduction}
        \linfer[dualRight+drw]
            {\linfer
                {\lclose}
            {\lsequent{}{C(x) \land t = t_0 \rightarrow \dbox{x' = f(x), t' = 1 \& t \leq T}{\norm{x}^2 < R^2}}}
            }
        {\lsequent{}{C(x) \land t = t_0 \rightarrow \ddiamond{x' = f(x), t' = 1}{t \geq T}}}
    \end{sequentdeduction}
    where the premise is proven by application of \rref{lem: finite completeness for open balls}. This completes the proof. 
\end{proof}

\rref{thm: finite completeness - safety} and \rref{thm: finite completeness - existence} do \emph{not} require the flow of the compact IVP to be well-defined a priori, as $\dL$ is capable of proving this from the validity of the formulas in question. 

There is a natural dual part to \rref{thm: finite completeness - safety}, involving liveness formulas of the form 
\[\lsequent{}{C(x) \land t = t_0 \rightarrow \ddiamond{x' = f(x), t' = 1 \& t \leq T}{O(x)}}\]
$\dL$ is indeed also complete for formulas of this form, and the requirements on $O(x)$ can even be slightly relaxed in comparison with the earlier theorems to just characterizing an open set that is not necessarily bounded.

\begin{theorem}[Completeness for liveness]
    \label{thm: finite completeness - liveness}
    Let $(f(x), C(x), [t_0, T])$ be a compact IVP with well-defined flow $\phi : \eval{C} \times [t_0, T] \to \R^n$ and $O(x)$ a $\folr$ formula characterizing an open set. Then $\dL$ is complete for formulas of the form 
    \[C(x) \land t = t_0 \rightarrow \ddiamond{x' = f(x), t' = 1 \& t \leq T}{O(x)}\]
    I.e. the following equivalence holds
    \begin{align*}
        &\models C(x) \land t = t_0 \rightarrow \ddiamond{x' = f(x), t' = 1 \& t \leq T}{O(x)} \iff\\
        &\lsequent{}{C(x) \land t = t_0 \rightarrow \ddiamond{x' = f(x), t' = 1 \& t \leq T}{O(x)}}
    \end{align*}
\end{theorem}

\begin{proof}
    As $\impliedby$ is soundness, we only handle $\implies$, so suppose that the formula is valid. Similar to the proof of \rref{thm: finite completeness - safety}, denote by $\theta_n$ the (vectorial) polynomial of error at most $2^{-n}$ as computed by \rref{thm: stone weierstrass} for each $n \in \N$. Now (computably) search for some $n \in \N$ such that the following formulas are valid (note that the first formula is always valid by construction of $\theta_n$)
    \begin{align*}
        & C(x) \land x = x_0 \land t = t_0 \rightarrow \dbox{x' = f(x), t' = 1 \& t \leq T}{\norm{\theta_n(x_0, t) - x}^2 \leq 2^{-2n}}\\
        &\forall x_0 \in \eval{C}\exists t \left(t_0 \leq t \land t \leq T \land B[\theta_n(x_0, t), 2^{-n}] \subseteq \eval{O} \right)
    \end{align*}
    For this to be a well-defined procedure, we prove that such an $n \in \N$ necessarily exists. Suppose for the sake of contradiction that this is false, then for all $n \in \N$, the following hold:
    \begin{enumerate}
        \item $\norm{\theta_n - \phi}_{C^0(\eval{C} \times [t_0, T])} \leq 2^{-n}$
        \item There exists some $z_n \in \eval{C}$ such that for all $t \in [t_0, T]$, $B[\theta_n(z_n, t), 2^{-n}] \nsubseteq \eval{O}$.
    \end{enumerate}
    Since $\eval{C}$ is compact, we may assume without loss of generality (by re-indexing if necessary), that the sequence $z_n \to z \in \eval{C}$ converges to some $z$. To achieve a contradiction, it suffices to show that $\phi(z, t) \notin \eval{O}$ for all $t \in [t_0, T]$. Let $t \in [t_0, T]$ be arbitrary and denote $d : \R^n \to \R$ as the distance function associated to the closed set $\eval{O}^C$. For all $n \in \N$, we have
    \begin{align*}
        d(\phi(z, t)) &\leq d(\theta_n(z_n, t)) + \norm{\theta_n(z_n, t) - \phi(z, t)}\\
        &\leq 2^{-n} + \norm{\theta_n(z_n, t) - \phi(z, t)}  \text{\qquad\qquad(by choice of $z_n$ in (2))}\\ 
        &\leq 2^{-n} + \norm{\theta_n(z, t) - \phi(z, t)} + \norm{\theta_n(z_n, t) - \theta_n(z, t)}\\
        &\leq 2^{-n} + \norm{\theta_n - \phi}_{C^0(\eval{C} \times [t_0, T])} + \norm{\theta_n(z_n, t) - \phi(z_n, t)} + \norm{\phi(z_n, t) - \theta_n(z, t)}\\
        &\leq 2^{-n + 1} + \norm{\theta_n - \phi}_{C^0(\eval{C} \times [t_0, T])} + \norm{\phi(z_n, t) - \phi(z, t)} + \norm{\phi(z, t) - \theta_n(z, t)}\\
        &\leq 2^{-n + 2} + \norm{\phi(z_n, t) - \phi(z, t)} \xrightarrow{n \to \infty} 0
    \end{align*}
    where the final convergence uses the fact that $\phi$ is continuous. Since the argument above holds for all $t \in [t_0, T]$, this shows $\phi(z, [t_0, T]) \cap \eval{O} = \emptyset$, a contradiction. 
    Thus, there necessarily exists some $n$ such that both formulas are valid and therefore provable via \rref{thm: stone weierstrass}
    and \irref{qear}. To continue, first note that the following is provable
    \[C(x) \land x = x_0 \land t = t_0 \rightarrow \ddiamond{t' = 1 \& t \leq T}{B[\theta_n(x_0, t), 2^{-n}] \subseteq \eval{O}}\]
    with derivation
    \begin{sequentdeduction}
        \linfer[implyr]
            {\linfer[thereAndBack_d+existsr]
                {\linfer[Kd]
                    {\linfer[evolved+qear]
                        {\linfer[qear]
                            {\lclose}
                        {\lsequent{C(x), x = x_0}{\exists t\left( t_0 \leq t \land t \leq T \land B[\theta_n(x_0, t), 2^{-n}] \subseteq \eval{O}\right)}}
                        }
                    {\lsequent{C(x), x = x_0, t = t_0}{\ddiamond{t' = 1}{\left(B[\theta_n(x_0, t), 2^{-n}] \subseteq \eval{O} \land t \leq T\right)}}}
                    }
                {\lsequent{C(x), x = x_0, t = t_0}{\ddiamond{t' = 1}{\left(B[\theta_n(x_0, t), 2^{-n}] \subseteq \eval{O} \land \dbox{t' = -1}{(t \geq t_0 \rightarrow t \leq T)}\right)}}}
                }
            {\lsequent{C(x), x = x_0, t = t_0}{\ddiamond{t' = 1 \& t \leq T}{B[\theta_n(x_0, t), 2^{-n}] \subseteq \eval{O}}}}
            }
        {\lsequent{}{C(x) \land x = x_0 \land t = t_0 \rightarrow \ddiamond{t' = 1 \& t \leq T}{B[\theta_n(x_0, t), 2^{-n}] \subseteq \eval{O}}}}
    \end{sequentdeduction}
    where the final application of \irref{qear} is sound by the construction of $n$, and axiom \irref{Kd} was applied assuming $t \leq T \rightarrow \dbox{t' = -1}{t \leq T}$, which is a valid invariant and can be proven by \irref{dinv}. Next, the following formula can be derived with a direct application of axiom \irref{bdgd}
    \[C(x) \land x = x_0 \land t = t_0 \rightarrow \ddiamond{x' = f(x), t' = 1 \& t \leq T}{B[\theta_n(x_0, t), 2^{-n}] \subseteq \eval{O}}\]
    which uses the (provable) formulas 
    \begin{align*}\
        &C(x) \land x = x_0 \land t = t_0 \rightarrow \dbox{x' = f(x), t' = 1 \& t \leq T}{\norm{\theta_n(x_0, t) - x}^2 \leq 2^{-2n}}\\
        &C(x) \land x = x_0 \land t = t_0 \rightarrow \ddiamond{t' = 1 \& t \leq T}{B[\theta_n(x_0, t), 2^{-n}] \subseteq \eval{O}}
    \end{align*}
    Finally, applying axioms \irref{drd+dWd} with the (provable) formulas
    \begin{align*}
        &C(x) \land x = x_0 \land t = t_0 \rightarrow \dbox{x' = f(x), t' = 1 \& t \leq T}{\norm{\theta_n(x_0, t) - x}^2 \leq 2^{-2n}}\\
        &C(x) \land x = x_0 \land t = t_0 \rightarrow \ddiamond{x' = f(x), t' = 1 \& t \leq T}{B[\theta_n(x_0, t), 2^{-n}] \subseteq \eval{O}}
    \end{align*}
    proves
    \[C(x) \land x = x_0 \land t = t_0 \rightarrow \ddiamond{x' = f(x), t' = 1 \& t \leq T}{(B[\theta_n(x_0, t), 2^{-n}] \subseteq \eval{O} \land \norm{\theta_n(x_0, t) - x}^2 \leq 2^{-2n})}\]
    another application of \irref{Kd} gives
    \[C(x) \land t = t_0 \rightarrow \ddiamond{x' = f(x), t' = 1 \& t \leq T} O(x)\]
    completing the proof of completeness for open properties. 
\end{proof}

\section{Conclusion}

By unifying both deductive and numerical techniques, this article establishes several completeness properties of compact IVPs. On a theoretical level, this proves complete reasoning principles for compact IVPs from purely qualitative properties. On a practical level, these results show that it is possible both to enjoy the capabilities of numerical methods, whilst retaining the rigorous level of trust provided by deductive, symbolic proofs. Alternatively, one could view such completeness results as a strengthening in the uniformity of such numerical algorithms. Standard numerical algorithms take in a single input and compute a corresponding output. As such, a different certifying proof of the output is needed for each individual input. This article improves on the level of uniformity for compact IVPs and establishes that there exists a single, symbolic proof in $\dL$ which proves the desired properties of the given compact IVP for \emph{all initial conditions} from the compact domain.

To achieve these completeness results, the article crucially establishes that rigorous error bounds can be proved in $\dL$ by reducing them down to differential invariance questions, providing a modular, rigorous way of verifying error bounds for numerical approximations. Utilizing this result, we then prove that $\dL$ is complete for (open and bounded) safety and liveness properties, as well as convergence for compact IVPs. This proof-theoretic result shows that not only is $\dL$ expressive enough, its axiomatization is also powerful enough to prove all such true properties of compact IVPs. The article also presented derivations of several classical theorems in $\dL$ along the way to establishing completeness, which are of independent interest. Notably including the Weierstrass approximation theorem, intermediate value theorem and the correspondence of global existence of solutions to IVPs and absence of finite time blow-up. 

For future work, it would be interesting to establish specific classes of LDAs that are general enough to preserve the completeness results while having a more tractable complexity in proving their error bounds in the sense of \rref{thm: effective proofs without upper bounds and constraints}. 

\textit{Acknowledgment.}
Funding has been provided by an Alexander von Humboldt Professorship and the National Science Foundation under Grant No. CCF 2220311.

\newpage
\section*{Appendix}
\appendix

\section{$\dL$ axiomatization}
\label{app: dL axiomatization}
This section provides a complete record of $\dL$'s axiomatization that is needed for the article.

\begin{theorem}[\cite{DBLP:conf/lics/Platzer12b, DBLP:journals/jacm/PlatzerT20, DBLP:journals/fac/TanP21}]
    \label{thm: base axiomatization of dL}
    The following are sound axioms of $\dL$. In axioms \irref{cont}, \irref{dadj}, \irref{bdg}, the variables $y$ is fresh. In axiom \irref{bdg}, $Q(x)$ is required to be a formula of real arithmetic. 

    \begin{calculus}
        \cinferenceRule[qear|\usebox{\Rval}]{quantifier elimination real arithmetic}
        {\linferenceRule[sequent]
          {}
          {\lsequent[g]{\Gamma}{\Delta}}
        }{$\text{if}~\landfold_{\ausfml\in\Gamma} \ausfml \limply \lorfold_{\busfml\in\Delta} \busfml ~\text{is valid in \LOS[\reals]}$}%
    
        \cinferenceRule[diamond|$\didia{\cdot}$]{diamond axiom}
        {\linferenceRule[equiv]
          {\lnot\dbox{\ausprg}{\lnot \ausfml}}
          {\ddiamond{\ausprg}{\ausfml}}
        }
        {}

        \cinferenceRule[evolved|$\didia{'}$]{evolve}
        {\linferenceRule[equiv]
          {\lexists{t{\geq}0}{\ddiamond{\pupdate{\pumod{x}{y(t)}}}{p(x)}}\hspace{1cm}}
          {\ddiamond{\pevolve{\D{x}=\genDE{x}}}{p(x)}}
        }{$\m{\D{y}(t)=\genDE{y}}$}%

        \cinferenceRule[B|B$'$]{}
        {\linferenceRule[equiv]
          {\lexists{y}{\ddiamond{\pevolvein{x'=f(x)}{Q(x)}}{\rfvar(x,y)}}}
          {\ddiamond{\pevolvein{x'=f(x)}{Q(x)}}{\exists{y}\rfvar(x,y)}}
        }{\text{$y \not\in x$}}
        
        \cinferenceRule[K|K]{K axiom / modal modus ponens} %
        {\linferenceRule[impl]
          {\dbox{\alpha}{(\fvarA \limply \fvarB)}}
          {(\dbox{\alpha}{\fvarA}\limply\dbox{\alpha}{\fvarB})}
        }{}
        \cinferenceRule[V|V]{vacuous $\dbox{}{}$}
         {\linferenceRule[impl]
           {\fvarA}
           {\dbox{\alpha}{\fvarA}}
         }{\text{no free variable of $\fvarA$ is bound by $\alpha$}}
        \cinferenceRule[G|G]{$\dbox{}{}$ generalization} %
        {\linferenceRule[formula]
          {\lsequent{}{\fvarA}}
          {\lsequent{\Gamma}{\dbox{\alpha}{\fvarA}}}
        }{}
        
        \cinferenceRule[dW|dW]{}
        {\linferenceRule
          {\lsequent{\ivr}{P}}
          {\lsequent{\Gamma}{\dbox{\pevolvein{\D{x}=\genDE{x}}{\ivr}}{P}}}
        }{}
        
        \cinferenceRule[dC|dC]{differential cut}%
        {\linferenceRule[sequent]
          {\lsequent[L]{}{\dbox{\pevolvein{\D{x}=\genDE{x}}{\ivr}}{\cusfml}}
          &\lsequent[L]{}{\dbox{\pevolvein{\D{x}=\genDE{x}}{(\ivr\land \cusfml)}}{\ousfml[x]}}}
          {\lsequent[L]{}{\dbox{\pevolvein{\D{x}=\genDE{x}}{\ivr}}{\ousfml[x]}}}
        }{}
        
        \cinferenceRule[DG|DG]{differential ghost variables}
        {\linferenceRule[viuqe]
          {\dbox{\pevolvein{x'=\genDE{x}}{\ivr(x)}}{\ousfml[x](x)}}
          {\lexists{y}{\dbox{\pevolvein{x'=\genDE{x}\syssep y'=a(x)\cdot y+b(x)}{\ivr(x)}}{\ousfml[x](x)}}}
        }
        {}

        \cinferenceRule[DGi|DGi]{differential ghost variables}
        {\linferenceRule[impl]
          {\dbox{\pevolvein{x'=\genDE{x}}{\ivr(x)}}{\ousfml[x](x)}}
          {\forall{y}{\dbox{\pevolvein{x'=\genDE{x}\syssep y'= g(x, y)}{\ivr(x)}}{\ousfml[x](x)}}}
        }
        {}
        
        \cinferenceRule[thereAndBack_orig|{[$\&$]}]{there and back quantification}
       {
       \linferenceRule[equiv]
       {\forall t_0 {=} c_0 \dbox{x' = \theta}{\left(\dbox{x' = -\theta}{\left(c_0 \geq t_0 \rightarrow \chi\right) \rightarrow \phi}\right)}}
        {\dbox{x' = \theta \& \chi}{\phi}}
       }
       {}

        \cinferenceRule[dx|DX]{differential skip}
        {
            \linferenceRule[equiv]
            {\left(Q \rightarrow P \land \dbox{x' = f(x) \& Q}{P}\right)}
            {\dbox{x'= f(x) \& Q}{P}}
        }
        {$x' \notin P, Q$}
        
        \cinferenceRule[uniq|Uniq]{vanilla uniqueness}
        {\linferenceRule[equiv]
        {\left(\ddiamond{x'= f(x) \& Q_1}{P}\right) \land \left(\ddiamond{x'= f(x) \& Q_2}{P}\right)}
        {\ddiamond{x'= f(x) \& Q_1 \land Q_2}{P}}
        }{}

        \cinferenceRule[cont|Cont]{continuity of ODE}
        {\linferenceRule[impl]
        {x = y}
        {\left(\ddiamond{x'= f(x) \& e > 0}{x \neq y} \leftrightarrow e > 0\right)}
        }{$f(x) \neq 0$}

        \cinferenceRule[dadj|Dadj]{reverse flow of ODEs}
        {\linferenceRule[equiv]
        {\ddiamond{y' = -f(y) \& Q(y)}{y = x}}
        {\ddiamond{x' = f(x) \& Q(x)}{x = y}}
        }{}

        \cinferenceRule[ri|RI]{real induction axiom}
        {\linferenceRule[equiv]
        {\forall y\dbox{x' = f(x) \& P \lor x = y}{\left(x = y \rightarrow P \land \ddiamond{x' = f(x) \& P \lor x = y}{x \neq y}\right)}}
        {\dbox{x' = f(x)}{P}}
        }{}
        
        \cinferenceRule[bdg|BDG]{bounded differential ghost}
       {
       \linferenceRule[impll]
       {\dbox{x' = f(x), y' = g(x, y) \& Q(x)}{\norm{y}^2 \leq p(x)}}
        {\left(\dbox{x' = f(x) \& Q(x)}{P(x)} \leftrightarrow \dbox{x' = f(x), y' = g(x, y) \& Q(x)}{P(x)}\right)}
       }
       {}
       
    \end{calculus}
\end{theorem}

\begin{remark}
    In axioms \irref{thereAndBack_orig} and \irref{cont}, it is assumed that the ODE $x' = f(x)$ includes a clock variable $c_0' = 1$. This assumption can be made without loss of generality since a clock variable can always be added using \irref{DG}. The variable $t_0$ is also assumed to be fresh in \irref{thereAndBack_orig}.
\end{remark}
The following derivable axioms will also be used.

\begin{theorem}[\cite{DBLP:journals/jacm/PlatzerT20, DBLP:journals/fac/TanP21, DBLP:conf/lics/Platzer12b}]
    \label{thm: bounded differential ghost + differential refinement}
    The following axioms are derivable in $\dL$, where $e$ is a term. In axiom \irref{bdgd}, $Q(x)$ is required to be a formula of real arithmetic.

    \begin{calculus}
        \cinferenceRule[drd|DR$\didia{\cdot}$]{differential refinement}
        {
        \linferenceRule[impl]
        {\dbox{x' = f(x) \& R}{Q}}
        {\left(\ddiamond{x' = f(x) \& R}{P} \rightarrow \ddiamond{x' = f(x) \& Q}{P}\right)}
        }{}

        \cinferenceRule[drw|dRW$\didia{\cdot}$]{differential refinement}
        {
        \linferenceRule[sequent]
        {\lsequent{R}{Q} & \lsequent{\Gamma}{\ddiamond{x' = f(x) \& R}{P}}}
        {\lsequent{\Gamma}{\ddiamond{x' = f(x) \& Q}{P}}}
        }{}
        
        \cinferenceRule[bdgd|{BDG$\langle\cdot\rangle$}]{bounded differential ghost for diamond}
        {
        \linferenceRule[impll]
        {\dbox{x' = f(x), y' = g(x, y) \& Q(x)}{\norm{y}^2 \leq p(x)}}
        {\left(\ddiamond{x' = f(x) \& Q(x)}{P(x)} \rightarrow \ddiamond{x' = f(x), y' = g(x, y) \& Q(x)}{P(x)}\right)}
        }{}
        
        \cinferenceRule[Kd|$K\didia{\cdot}$]{Diamond Kripke axiom}
        {\linferenceRule[impl]
        {\dbox{\alpha}{\left(\phi \rightarrow \psi\right)}}
        {\left(\ddiamond{\alpha}{\phi} \rightarrow \ddiamond{\alpha}{\psi}\right)}
        }
        {}

        \cinferenceRule[dor|$\didia{}\lor$]{diamond or axiom}
        {\linferenceRule[equiv]
        {\ddiamond{\alpha}{\phi} \lor \ddiamond{\alpha}{\psi}}
        {\ddiamond{\alpha}{\left(\phi \lor \psi\right)}}
        }
        {}

        \cinferenceRule[band|${[]\land}$]{$\dbox{\cdot}{\land}$}
        {\linferenceRule[equiv]
          {\dbox{\alpha}{\phi} \land \dbox{\alpha}{\psi}}
          {\dbox{\alpha}{(\phi \land \psi)}}
        }{}%
        
        \cinferenceRule[enclosure|Enc]{topological enclosure}
        {
        \linferenceRule[sequent]
        {\lsequent{\Gamma}{e \geq 0} & \lsequent{\Gamma}{\dbox{x' = f(x) \& Q \land e \geq 0}{e > 0}}}
        {\lsequent{\Gamma}{\dbox{x' = f(x) \& Q}{e > 0}}}
        }
        {}
    \end{calculus}
\end{theorem}

\begin{remark}
    This article adopts ``rich-test'' $\dL$ which allows domain constraints $Q$ to be general $\dL$ formulas with modalities rather than just first order formulas of arithmetic unless explicitly restricted otherwise. Thus one should be cautious when employing previous axiomatization \cite{DBLP:journals/jacm/PlatzerT20, DBLP:journals/fac/TanP21} and ensure that they are still sound. Indeed, while earlier works stated the soundness of such axioms under the assumption of ``poor-test'' $\dL$, the proofs are more general and extend to ``rich-test'' $\dL$. 
\end{remark}

This concludes the brief overview of $\dL$'s proof calculus that will be needed for this paper. The usual FOL proof rules are listed below for completeness \cite{DBLP:journals/jar/Platzer08}.\\
\quad{}
\begin{calculuscollection}
\begin{calculus}
    \cinferenceRule[notl|$\lnot$\leftrule]{$\lnot$ left}
    {\linferenceRule[sequent]
      {\lsequent[L]{}{\asfml}}
      {\lsequent[L]{\lnot \asfml}{}}
    }{}%
    \cinferenceRule[andl|$\land$\leftrule]{$\land$ left}
    {\linferenceRule[sequent]
      {\lsequent[L]{\asfml , \bsfml}{}}
      {\lsequent[L]{\asfml \land \bsfml}{}}
    }{}%
    \cinferenceRule[orl|$\lor$\leftrule]{$\lor$ left}
    {\linferenceRule[sequent]
      {\lsequent[L]{\asfml}{}
        & \lsequent[L]{\bsfml}{}}
      {\lsequent[L]{\asfml \lor \bsfml}{}}
    }{}%
    \cinferenceRule[notr|$\lnot$\rightrule]{$\lnot$ right}
    {\linferenceRule[sequent]
      {\lsequent[L]{\asfml}{}}
      {\lsequent[L]{}{\lnot \asfml}}
    }{}%
    \cinferenceRule[andr|$\land$\rightrule]{$\land$ right}
    {\linferenceRule[sequent]
      {\lsequent[L]{}{\asfml}
        & \lsequent[L]{}{\bsfml}}
      {\lsequent[L]{}{\asfml \land \bsfml}}
    }{}%
    \cinferenceRule[cut|cut]{cut}
    {\linferenceRule[sequent]
      {\lsequent[L]{}{\cusfml}
      &\lsequent[L]{\cusfml}{}}
      {\lsequent[L]{}{}}
    }{}%
    \cinferenceRule[orr|$\lor$\rightrule]{$\lor$ right}
    {\linferenceRule[sequent]
      {\lsequent[L]{}{\asfml, \bsfml}}
      {\lsequent[L]{}{\asfml \lor \bsfml}}
    }{}%
\end{calculus}
\qquad
\begin{calculus}
    \cinferenceRule[implyl|$\limply$\leftrule]{$\limply$ left}
    {\linferenceRule[sequent]
      {\lsequent[L]{}{\asfml}
        & \lsequent[L]{\bsfml}{}}
      {\lsequent[L]{\asfml \limply \bsfml}{}}
    }{}%
    \cinferenceRule[alll|$\forall$\leftrule]{$\lforall{}{}$ left instantiation}
    {\linferenceRule[sequent]
      {\lsequent[L]{p(\astrm)}{}}
      {\lsequent[L]{\lforall{x}{p(x)}}{}}
        \qquad{}\qquad{}
    }{arbitrary term $\astrm$}%
    \cinferenceRule[existsl|$\exists$\leftrule]{$\lexists{}{}$ left}
    {\linferenceRule[sequent]
      {\lsequent[L]{p(y)} {}}
      {\lsequent[L]{\lexists{x}{p(x)}} {}}
    }{\m{y\not\in\Gamma{,}\Delta{,}\lexists{x}{p(x)}}}%
    \cinferenceRule[implyr|$\limply$\rightrule]{$\limply$ right}
    {\linferenceRule[sequent]
      {\lsequent[L]{\asfml}{\bsfml}}
      {\lsequent[L]{}{\asfml \limply \bsfml}}
    }{}%
    \cinferenceRule[allr|$\forall$\rightrule]{$\lforall{}{}$ right}
    {\linferenceRule[sequent]
      {\lsequent[L]{}{p(y)}}
      {\lsequent[L]{}{\lforall{x}{p(x)}}}
    }{\m{y\not\in\Gamma{,}\Delta{,}\lforall{x}{p(x)}}}%
    \cinferenceRule[existsr|$\exists$\rightrule]{$\lexists{}{}$ right}
    {\linferenceRule[sequent]
      {\lsequent[L]{}{p(\astrm)}}
      {\lsequent[L]{}{\lexists{x}{p(x)}}}
    }{arbitrary term $\astrm$}%

    \cinferenceRule[id|id]{identity}
    {\linferenceRule[sequent]
      {\lclose}
      {\lsequent[L]{\asfml}{\asfml}}
    }{}%
\end{calculus}
\end{calculuscollection}

\section{Derived axioms and proof rules}
\label{app: derivations}
This section proves \rref{lem:lemma axioms} and subsequently \rref{thm:axioms for step liveness}. The following lemma proves useful properties of constant assumptions and a diamond analog of axiom \irref{thereAndBack_orig}. 

\begin{lemma}[{\cite[Appendix~A.2]{DBLP:journals/jacm/PlatzerT20}}]
    \label{lem: vacuous axioms}
    The following axioms/proof rules are derivable and thus sound, where $R(y)$ is a $\dL$ formula only depending on its free variables $y$ which has no differential equation in $x' = f(x)$.
    
    \begin{calculus}
    \cinferenceRule[dWd|{dW$\didia{\cdot}$}]{Diamond differential weakening}
        {
        \linferenceRule[impl]
        {\ddiamond{x' = f(x) \& Q}{P}}
        {\ddiamond{x' = f(x) \& Q}{\left(P \land Q\right)}}
        }{}
    \cinferenceRule[V_1|V]{vacuous}
    {\linferenceRule[sequent]
        {\lsequent{\Gamma}{\dbox{x' = f(x) \& Q \land R(y)}{P}}}
        {\lsequent{\Gamma, R(y)}{\dbox{x' = f(x) \& Q}{P}}}
    }{}
    \cinferenceRule[V_2|V]{vacuous}
    {\linferenceRule[sequent]
        {\lclose}
        {\lsequent{\Gamma, \ddiamond{x' = f(x) \& Q}{(P \land R(y))}}{R(y)}}
    }{}
    \cinferenceRule[V_3|V]{vacuous}
    {\linferenceRule[sequent]
        {\lclose}
        {\lsequent{R(y), \ddiamond{x' = f(x) \& Q}{P}}{\ddiamond{x' = f(x) \& Q}{(P \land R(y))}}}
    }{}
    \cinferenceRule[thereAndBack_d|{$\langle\&\rangle$}]{there and back diamond}
       {
       \linferenceRule[equiv]
       {\exists t_0 {=} c_0 \ddiamond{x' = \theta}{\left(\phi \land \dbox{x' = -\theta}{\left(c_0 \geq t_0 \rightarrow \chi\right)}\right)}}
        {\ddiamond{x' = \theta \& \chi}{\phi}}
        \qquad
       }
       {}
\end{calculus}
\end{lemma}

\begin{proof}
    Axiom \irref{dWd} can be derived as follows
    \begin{sequentdeduction}
    \linfer[Kd]
    {\linfer[K]
        {\linfer[dW]
        {\lclose}
        {\lsequent{}{\dbox{x' = f(x) \& Q}{Q}}}
        }
        {\lsequent{}{\dbox{x' = f(x) \& Q}{\left(P \rightarrow P \land Q\right)}}}
    }
    {\lsequent{}{\ddiamond{x' = f(x) \& Q}{P} \rightarrow \ddiamond{x' = f(x) \& Q}{\left(P \land Q\right)}}}
    \end{sequentdeduction}
    The last proof rule labeled as \irref{V} can be derived using \irref{dWd}
    \begin{sequentdeduction}
    \linfer[drd+V]
        {
            \linfer[dWd]
                {   
                    \linfer[drw]
                        {\lclose}
                    {\lsequent{\ddiamond{x' = f(x) \& Q \land R(y)}{(P \land R(y))}}{\ddiamond{x' = f(x) \& Q}{(P \land R(y))}}}
                }
            {\lsequent{\ddiamond{x' = f(x) \& Q \land R(y)}{P}}{\ddiamond{x' = f(x) \& Q}{(P \land R(y))}}}
        }
    {\lsequent{R(y), \ddiamond{x' = f(x) \& Q}{P}}{\ddiamond{x' = f(x) \& Q}{(P \land R(y))}}}
\end{sequentdeduction}
    Axiom \irref{thereAndBack_d} can be derived directly from axioms \irref{thereAndBack_orig+diamond}, and the remaining proof-rules have been derived in earlier works \cite[Appendix~A.2]{DBLP:journals/jacm/PlatzerT20}.
\end{proof}

Axiom \irref{dWd} asserts that domain constraints are always satisfied along the flow, the next three proof rules assert that the truth of constant properties remain unchanged along the ODE flows, all special cases of axiom \irref{V} \cite[Appendix~A.2]{DBLP:journals/jacm/PlatzerT20} and thus have the same name. Similar to earlier works \cite{DBLP:journals/jacm/PlatzerT20}, manipulations of constant properties in derivations will be abbreviated with \irref{V}. Axiom \irref{thereAndBack_d} is the diamond analog of \irref{thereAndBack_orig}, similar to \irref{thereAndBack_orig}, $c_0' = 1$ is a clock variable in $x' = f(x)$ and $t_0$ is fresh. 
\begin{proof}[Proof of \rref{lem:lemma axioms}]
    \irref{thereAndBack}: Suppose for the sake of contradiction that the claim was false, there would be some state along the flow of $x' = f(x)$ such that reversing the flow does not return to the original state where $P$ is true. But this directly contradicts axiom \irref{dadj}, which says that it is always possible to reach the initial state by following the reverse flow.

\begin{sequentdeduction}
    \linfer[implyr+diamond+notr]
    {\linfer[Kd+G]
      {\linfer[B]{
        \linfer[existsl]{
            \linfer[cut+Kd]{
          \linfer[V]
            {  
             \linfer[dadj]{
                \linfer[V]
                 {\linfer[diamond+notl+id]
                 {\lclose}
                 {\lsequent{P(x), \ddiamond{y' = -f(y) \& Q}{P(y)},  \dbox{y' = -f(y) \& Q}{\neg P(y)}}{\bot}}
                 }
                 {\lsequent{P(x), \ddiamond{y' = -f(y) \& Q}{x = y},  \dbox{y' = -f(y) \& Q}{\neg P(y)}}{\bot}}
             }
             {\lsequent{P(x), \ddiamond{x' = f(x) \& Q}{x = y},  \dbox{y' = -f(y) \& Q}{\neg P(y)}}{\bot}}
            }
            {\lsequent{P(x), \ddiamond{x' = f(x) \& Q}{(x = y \land \dbox{y' = -f(y) \& Q}{\neg P(y)}})} {\bot}}
            &
            {\text{\textcircled{1}}}
       }
       {\lsequent{P(x), \ddiamond{x' = f(x) \& Q}{(x = y \land \dbox{x' = -f(x) \& Q}{\neg P(x)}})} {\bot}}
        }
        {\lsequent{P(x), \exists y\ddiamond{x' = f(x) \& Q}{(x = y \land \dbox{x' = -f(x) \& Q}{\neg P(x)}})} {\bot}}
      }
      {\lsequent{P(x), \ddiamond{x' = f(x) \& Q}{(\exists y (x = y) \land \dbox{x' = -f(x) \& Q}{\neg P(x)}})} {\bot}}   
      }
    {\lsequent{P(x), \ddiamond{x' = f(x) \& Q}{\dbox{x' = -f(x) \& Q}{\neg P(x)}}} {\bot}}
    }
    {\lsequent{} {P(x) \limply \dbox{x' = f(x) \& Q}{\ddiamond{x' = -f(x) \& Q}{P(x)}}}}
\end{sequentdeduction}
The open premise resulting from a cut with $x = y \land \dbox{x' = f(x) \& Q}{\neg P(x)} \rightarrow\dbox{y' = f(y) \& Q}{\neg P(y)}$ is
\[\text{\textcircled{1}} \equiv \lsequent{x = y, \dbox{x' = f(x) \& Q}{\neg P(x)}} {\dbox{y' = f(y) \& Q}{\neg P(y)}}\]

To complete the proof of \irref{thereAndBack}, premise \textcircled{1} needs to be resolved.
\begin{sequentdeduction}
\linfer[diamond+notr]
    {\linfer[G+Kd]
        {\linfer[B]
            {\linfer[existsl]
                {\linfer[Kd]
                {\linfer[V]{
                    \linfer[dadj]{
                        \linfer[V+Kd]{
                            \linfer[dadj]{
                                \linfer[V]{
                                    \linfer[Kd]{
                                        \linfer[diamond+notl+id]{
                                            \lclose
                                        }
                                        {\lsequent{x = y, \dbox{x' = f(x) \& Q}{\neg P(x)}, \ddiamond{x' = f(x) \& Q}{P(x)}} {\bot}}
                                    }
                                    {\lsequent{x = y, \dbox{x' = f(x) \& Q}{\neg P(x)}, \ddiamond{x' = f(x) \& Q}{(x = z \land P(z))}} {\bot}}
                                }
                            {\lsequent{x = y, \dbox{x' = f(x) \& Q}{\neg P(x)}, \ddiamond{x' = f(x) \& Q}{x = z}, P(z)} {\bot}}
                            }
                            {\lsequent{x = y, \dbox{x' = f(x) \& Q}{\neg P(x)}, \ddiamond{z' = -f(z) \& Q}{z = x}, P(z)} {\bot}}
                        }
                        {\lsequent{x = y, \dbox{x' = f(x) \& Q}{\neg P(x)}, \ddiamond{z' = -f(z) \& Q}{z = y}, P(z)} {\bot}}
                    }
                    {\lsequent{x = y, \dbox{x' = f(x) \& Q}{\neg P(x)}, \ddiamond{y' = f(y) \& Q}{y = z}, P(z)} {\bot}}
                    }
                    {\lsequent{x = y, \dbox{x' = f(x) \& Q}{\neg P(x)}, \ddiamond{y' = f(y) \& Q}{(y = z \land P(z))}} {\bot}}
                }    
                {\lsequent{x = y, \dbox{x' = f(x) \& Q}{\neg P(x)}, \ddiamond{y' = f(y) \& Q}{(y = z \land P(y))}} {\bot}}
                }
                {\lsequent{x = y, \dbox{x' = f(x) \& Q}{\neg P(x)}, \exists z\ddiamond{y' = f(y) \& Q}{(y = z \land P(y))}} {\bot}}
            }
        {\lsequent{x = y, \dbox{x' = f(x) \& Q}{\neg P(x)}, \ddiamond{y' = f(y) \& Q}{(\exists z (y = z) \land P(y))}} {\bot}}
        }
    {\lsequent{x = y, \dbox{x' = f(x) \& Q}{\neg P(x)}, \ddiamond{y' = f(y) \& Q}{P(y)}} {\bot}}
    }
{\lsequent {x = y, \dbox{x' = f(x) \& Q}{\neg P(x)}} {\dbox{y' = f(y) \& Q}{\neg P(y)}}}
\end{sequentdeduction}

This completes the proof of axiom \irref{thereAndBack}.\\
\\
\irref{timeFixpoint}: While there might be easier ways to prove this, the completeness axiom \irref{dinv} for differential invariants gives the difficult direction immediately.

\begin{sequentdeduction}
\linfer[implyr]
    {\linfer[implyr]
        {\linfer
            {\text{\textcircled{1}}}
        {\lsequent{t = t_0, P} {\dbox{x' = f(x), t' = 1 \& t \leq t_0}{P}}}
        &\linfer[dx]
        {\linfer[qear]
            {\lclose}
            {\lsequent{t = t_0, t \leq t_0 \rightarrow P}{P}}
        }
        {\lsequent{t = t_0, \dbox{x' = f(x), t' = 1 \& t \leq t_0}{P}} {P}}
        }
    {\lsequent{t = t_0} {{\dbox{x' = f(x), t' = 1 \& t \leq t_0}{P} \leftrightarrow P}}}
    }
{\lsequent {} {t = t_0 \rightarrow \left({\dbox{x' = f(x), t' = 1 \& t \leq t_0}{P} \leftrightarrow P}\right)}}
\end{sequentdeduction}

Premise $\text{\textcircled{1}}$ is easily proven by noting that $I \equiv t = t_0 \land P$ is a valid differential invariant of the ODE $x' = f(x), t' = 1 \& t \leq t_0$, and can therefore be proven via \irref{dinv}.

\begin{sequentdeduction}
\linfer[cut]
    {\linfer[implyl+K]
        {\lclose}
    {\lsequent{I, I \rightarrow [x' = f(x), t' = 1 \& t \leq t_0]I} {\dbox{x' = f(x), t' = 1 \& t \leq t_0}{P}}}
    &
    \linfer[dinv]
    {\lclose}
    {\lsequent{} {I \rightarrow [x' = f(x), t' = 1 \& t \leq t_0]I}}
    }
{\lsequent{t = t_0, P} {\dbox{x' = f(x), t' = 1 \& t \leq t_0}{P}}}
\end{sequentdeduction}

This completes the proof of \irref{timeFixpoint}. The following useful corollary of \irref{timeFixpoint} will be used in future derivations. 
\[t = t_0 \land \neg P \land \ddiamond{x' = f(x), t' = 1 \& Q}{P} \rightarrow \ddiamond{x' = f(x), t' = 1 \& Q}{\left(P \land t > t_0\right)}\]
Semantically this is not surprising, if $P$ is not satisfied at the initial state, then there must be some evolution along the ODE to reach a state where $P$ is true, and since $t' = 1$ is strictly increasing, such a state must also satisfy $t > t_0$. However, axiom \irref{timeFixpoint} provides a syntactic derivation of the axiom. The derivation begins with axiom \irref{B} to quantify the final time value reached, cutting in an appropriate domain constraint that captures the monotonicity of $t' = 1$ then completes the derivation with an application of axiom \irref{timeFixpoint}.
\begin{sequentdeduction}
    \linfer[implyr+Kd]{
        \linfer[B+existsl]{
           \linfer[Kd]
        {\linfer[K]
            {\linfer[V]
                {\linfer[cut]
                    {\linfer[drw+Kd]{
                        \linfer[notr]{
                            \linfer[drd+V]{
                                \linfer[diamond+notl]{
                                    \linfer[timeFixpoint]
                                        {
                                            \lclose
                                        }
                                        {\lsequent{t = t_0, \neg P}{\dbox{x' = f(x), t' = 1 \& t \leq t_0}{\neg P }}}
                                }
                                {\lsequent{t = t_0, \neg P, \ddiamond{x' = f(x), t' = 1 \& t \leq t_0}{P}}{\bot}}
                            }
                            {\lsequent{t = t_0, s \leq t_0, \neg P, \ddiamond{x' = f(x), t' = 1 \& t \leq s}{P}}{\bot}}
                        }
                        {\lsequent{t = t_0, \neg P, \ddiamond{x' = f(x), t' = 1 \& t \leq s}{P}}{s > t_0}}
                    }
                    {\lsequent{t = t_0, \neg P, \ddiamond{x' = f(x), t' = 1 \& Q \land t \leq s}{\left(P \land t = s\right)}}{s > t_0}}
                    &
                    {\text{\textcircled{2}}}
                    }
                {\lsequent{t = t_0, \neg P, \ddiamond{x' = f(x), t' = 1 \& Q}{\left(P \land t = s\right)}}{s > t_0}}
                }
            {\lsequent{t = t_0, \neg P, \ddiamond{x' = f(x), t' = 1 \& Q}{\left(P \land t = s\right)}}{\dbox{x' = f(x), t' = 1 \& Q}{s > t_0}}}
            }
        {\lsequent{t = t_0, \neg P, \ddiamond{x' = f(x), t' = 1 \& Q}{\left(P \land t = s\right)}}{\dbox{x' = f(x), t' = 1 \& Q}{(P \land t = s \rightarrow P \land t > t_0)}}}
        }
    {\lsequent{t = t_0, \neg P, \ddiamond{x' = f(x), t' = 1 \& Q}{\left(P \land t = s\right)}}{\ddiamond{x' = f(x), t' = 1 \& Q}{\left(P \land t > t_0\right)}}}
        }
    {\lsequent{t = t_0, \neg P, \ddiamond{x' = f(x), t' = 1 \& Q}{\left(P \land \exists s (t = s)\right)}}{\ddiamond{x' = f(x), t' = 1 \& Q}{\left(P \land t > t_0\right)}}}
    }
    {\lsequent{}{t = t_0 \land \neg P \land \ddiamond{x' = f(x), t' = 1 \& Q}{P} \rightarrow \ddiamond{x' = f(x), t' = 1 \& Q}{\left(P \land t > t_0\right)}}}
\end{sequentdeduction}

The open premise \textcircled{2} arising from cutting in $\ddiamond{x' = f(x), t' = 1 \& Q \land t \leq s}{(P \land t = s)}$ is:
\[\text{\textcircled{2}} \equiv \lsequent{\ddiamond{x' = f(x), t' = 1 \& Q}{\left(P \land t = s\right)}}{\ddiamond{x' = f(x), t' = 1 \& Q \land t \leq s}{\left(P \land t = s\right)}}\]
It therefore remains to prove \textcircled{2}, which follows directly from axiom \irref{thereAndBack_d} with clock variable $t' = 1$ and noting that $t \leq s \rightarrow \dbox{x' = -f(x), t' = -1}{t \leq s}$ is a valid differential invariant, we also make the following abbreviation for clarity.
\[A \equiv \ddiamond{x' = f(x), t' = 1 \& Q \land t \leq s}{\left(P \land t = s\right)}\]
\begin{sequentdeduction}
    \linfer[thereAndBack_d+existsl]
    {\linfer[dinv+Kd]
        {\linfer[thereAndBack_d+id]
            {\lclose}
        {\lsequent{t_0 = t, \ddiamond{x' = f(x), t' = 1}{\left(P \land t = s \land \dbox{x' = -f(x), t' = -1}{\left(t \geq t_0 \rightarrow Q \land t \leq s\right)}\right)}}{A}}
        }
    {\lsequent{t_0 = t, \ddiamond{x' = f(x), t' = 1}{\left(P \land t = s \land \dbox{x' = -f(x), t' = -1}{\left(t \geq t_0 \rightarrow Q\right)}\right)}}{A}}
    }
    {\lsequent{\ddiamond{x' = f(x), t' = 1 \& Q}{\left(P \land t = s\right)}}{\ddiamond{x' = f(x), t' = 1 \& Q \land t \leq s}{\left(P \land t = s\right)}}}
\end{sequentdeduction}
This completes the proof. The corollaries of \irref{timeFixpoint} are recorded below as axioms, the positive time versions (i.e. \irref{timeFixpoint_f+mont_f}) have been derived above, and the negative time versions can be derived in exactly the same fashion with $t' = -1$ instead of $t' = 1$. 

\begin{calculus}
    \cinferenceRule[timeFixpoint_f|Stuck$^+$]{time fix forward}
   {
   \linferenceRule[impl]
   {{t = t_0 \land \neg P \land \ddiamond{x' = f(x), t' = 1 \& Q}{P}}}
   {\ddiamond{x' = f(x), t' = 1 \& Q}{\left(P \land t > t_0\right)}}
   }{}

   \cinferenceRule[timeFixpoint_b|Stuck$^-$]{time fix backward}
   {
   \linferenceRule[impl]
   {{t = t_0 \land \neg P \land \ddiamond{x' = f(x), t' = -1 \& Q}{P}}}
   {\ddiamond{x' = f(x), t' = -1 \& Q}{\left(P \land t < t_0\right)}}
   }{}
\end{calculus}

Premise \textcircled{2} and its negative time version will also be useful.

\begin{calculus}
    \cinferenceRule[mont_f|Mont$^+$]{<x' = f(x), t' = 1 \& Q>(P & t = s) -> <x' = f(x), t'= 1 \& Q & t <= s>(P & t = s)}
   {
   \linferenceRule[impl]
   {{\ddiamond{x' = f(x), t' = 1 \& Q}{\left(P \land t = s\right)}}}
   {{\ddiamond{x' = f(x), t' = 1 \& Q \land t \leq s}{\left(P \land t = s \right)}}}
   }{}

   \cinferenceRule[mont_b|Mont$^-$]{<x' = f(x), t' = 1 \& Q>(P & t = s) -> <x' = f(x), t'= 1 \& Q & t <= s>(P & t = s)}
   {
   \linferenceRule[impl]
   {{\ddiamond{x' = f(x), t' = -1 \& Q}{\left(P \land t = s\right)}}}
   {{\ddiamond{x' = f(x), t' = -1 \& Q \land t \geq s}{\left(P \land t = s \right)}}}
   }{}
\end{calculus}

It is useful to note that by utilizing axiom \irref{B}, the condition $t = s$ in axioms \irref{mont_f+mont_b} can also be substituted by $t \leq s$ and $t \geq s$ respectively. 

\irref{idempotence}: The derivation of axiom \irref{idempotence} heavily relies upon axiom \irref{thereAndBack_d} to repeatedly remove domain constrains within modalities. The following abbreviation will be useful in its derivation.
\[A \equiv \dbox{x' = -f(x)}{\left(c_0 \geq t_0 \rightarrow Q\right)}\]
The proof first applies \irref{thereAndBack_d} with the clock variable $c_0$ to the antecedent followed by Skolemizing the initial time value with \irref{existsl} to the witness $t_0$. Similarly, the second application of \irref{thereAndBack_d} is applied to the succedent followed by Skolemizing the clock variable to the same witness $t_0$. Axiom \irref{Kd} then reduces the open premise to proving an implication between the inner box-modalities that arised from \irref{thereAndBack_d}.
\begin{sequentdeduction}
    \linfer[thereAndBack_d]
    {
        \linfer[existsl]{
            \linfer[thereAndBack_d]{
                \linfer[existsr]{
                    \linfer[band+Kd]{
                        \linfer[Kd]{
                            \linfer[andr+id]{
                                \linfer[]{
                                    {\text{\textcircled{1}}}
                                }
                                {\lsequent {P, A} {\dbox{x' = -f(x)}{\left(c_0 \geq t_0 \rightarrow \ddiamond{x' = f(x) \& Q}{P}\right)}}}
                            }
                            {\lsequent {P, A} {P \land A \land \dbox{x' = -f(x)}{\left(c_0 \geq t_0 \rightarrow \ddiamond{x' = f(x) \& Q}{P}\right)}}}
                        }
                        {\lsequent {t_0 = c_0, \ddiamond{x' = f(x)}{\left(P \land A\right)}} {\ddiamond{x' = f(x)}{\left(P \land A \land \dbox{x' = -f(x)}{\left(c_0 \geq t_0 \rightarrow \ddiamond{x' = f(x) \& Q}{P}\right)}\right)}}}
                    }
                    {\lsequent {t_0 = c_0, \ddiamond{x' = f(x)}{\left(P \land A\right)}} {\ddiamond{x' = f(x)}{\left(P \land \dbox{x' = -f(x)}{\left(c_0 \geq t_0 \rightarrow Q \land \ddiamond{x' = f(x) \& Q}{P}\right)}\right)}}}
                }
                {\lsequent {t_0 = c_0, \ddiamond{x' = f(x)}{\left(P \land A\right)}} {\exists s_0 = c_0\ddiamond{x' = f(x)}{\left(P \land \dbox{x' = -f(x)}{\left(c_0 \geq s_0 \rightarrow Q \land \ddiamond{x' = f(x) \& Q}{P}\right)}\right)}}}
            }
            {\lsequent {t_0 = c_0, \ddiamond{x' = f(x)}{\left(P \land A\right)}} {\ddiamond{x' = f(x) \& Q \land \ddiamond{x' = f(x) \& Q}{P}}{P}}}
        }
        {\lsequent {\exists t_0 = c_0\ddiamond{x' = f(x)}{\left(P \land \dbox{x' = -f(x)}{\left(c_0 \geq t_0 \rightarrow Q\right)}\right)}} {\ddiamond{x' = f(x) \& Q \land \ddiamond{x' = f(x) \& Q}{P}}{P}}}
    }
    {\lsequent{\ddiamond{x' = f(x) \& Q}{P}} {\ddiamond{x' = f(x) \& Q \land \ddiamond{x' = f(x) \& Q}{P}}{P}}}
\end{sequentdeduction}

The open premise \textcircled{1} can now be proven by first negating the succedent and then applying axiom \irref{mont_b}. Recall that $x' = -f(x)$ is assumed to contain the clock variable $c_0' = -1$ and therefore \irref{mont_b} is applicable with the time variable $t$ being $c_0$. 

\begin{sequentdeduction}
    \linfer[diamond+notr]
        {\linfer[mont_b]
            {\linfer[drd]
                {\linfer[thereAndBack]
                    {\linfer[diamond+notl+id]
                        {\lclose}
                    {\lsequent{\dbox{x' = -f(x) \& Q}{\ddiamond{x' = f(x) \& Q}{P}}, \ddiamond{x' = -f(x) \& Q}{\dbox{x' = f(x) \& Q}{\neg P}}}{\bot}}
                    }
                {\lsequent{P, \ddiamond{x' = -f(x) \& Q}{\dbox{x' = f(x) \& Q}{\neg P}}}{\bot}}
                }
            {\lsequent{P, \dbox{x' = -f(x)}{\left(c_0 \geq t_0 \rightarrow Q\right)}, \ddiamond{x' = -f(x) \& c_0 \geq t_0}{\dbox{x' = f(x) \& Q}{\neg P}}}{\bot}}
            }
        {\lsequent{P, \dbox{x' = -f(x)}{\left(c_0 \geq t_0 \rightarrow Q\right)}, \ddiamond{x' = -f(x)}{\left(c_0 \geq t_0 \land \dbox{x' = f(x) \& Q}{\neg P}\right)}}{\bot}}
        }
    {\lsequent{P, \dbox{x' = -f(x)}{\left(c_0 \geq t_0 \rightarrow Q\right)}} {\dbox{x' = -f(x)}{\left(c_0 \geq t_0 \rightarrow \ddiamond{x' = f(x) \& Q}{P}\right)}}}
\end{sequentdeduction}

\irref{uniqp}: Before deriving this axiom, we make the following abbreviations for brevity:
\begin{align*}
    A &\equiv \ddiamond{x' = f(x) \& Q_1}{P_1}\\
    \A &\equiv \ddiamond{x' = f(x) \& Q_1 \land A}{P_1}\\
    B &\equiv \ddiamond{x' = f(x) \& Q_2}{P_2}\\
    \B &\equiv \ddiamond{x' = f(x) \& Q_2 \land B}{P_2}\\
    C &\equiv \ddiamond{x' = f(x) \& Q_1 \land Q_2}{\left(P_1 \land B\right)} \lor \ddiamond{x' = f(x) \& Q_1 \land Q_2}{\left(P_2 \land A\right)}
\end{align*}

\begin{sequentdeduction}
    \linfer[implyr+idempotence]
    {\linfer[Kd+uniq]
        {
            \linfer[dor]
            {
                \linfer[dWd+drw]
                {
                \linfer[id]
                {\lclose}
                {\lsequent{\ddiamond{x' = f(x) \& Q_1 \land Q_2}{\left(P_1 \land B\right)} \lor \ddiamond{x' = f(x) \& Q_1 \land Q_2}{\left(P_2 \land A\right)}}{C}}
                }
            {\lsequent{\ddiamond{x' = f(x) \& Q_1 \land Q_2 \land A \land B}{P_1} \lor \ddiamond{x' = f(x) \& Q_1 \land Q_2 \land A \land B}{P_2}}{C}}
            }
        {\lsequent{\ddiamond{x' = f(x) \& Q_1 \land Q_2 \land A \land B}{\left(P_1 \lor P_2\right)}} {C}}
        }
    {\lsequent{\A, \B}{\ddiamond{x' = f(x) \& Q_1 \land Q_2}{\left(P_1 \land B\right)} \lor \ddiamond{x' = f(x) \& Q_1 \land Q_2}{\left(P_2 \land A\right)}}}
    }
    {\lsequent{}{A \land B \rightarrow \ddiamond{x' = f(x) \& Q_1 \land Q_2}{\left(P_1 \land B\right)} \lor \ddiamond{x' = f(x) \& Q_1 \land Q_2}{\left(P_2 \land A\right)}}}
\end{sequentdeduction}

\irref{ivt}: Classically, the intermediate value theorem is usually proven directly from the completeness of $\R$ (and indeed they are equivalent), so it might be expected that axiom \irref{ri} is utilized. Indeed, the derivation relies upon the derived proof rule \irref{enclosure}, which itself relies on \irref{ri}. The derivation begins by contradiction, negating the succedent and applying \irref{enclosure} proves that $e < 0$ always holds along the flow under the domain constraint $t < \tau$. Note that $e \neq 0$ under the domain constraint of $e \leq 0$ reduces down to $e < 0$ by \irref{K}. 

\begin{sequentdeduction}
    \linfer[implyr+diamond]
    {\linfer[notr+K]{
        \linfer[enclosure]
        {\linfer
            {\text{\textcircled{1}}}
        {\lsequent{e \leq 0, \ddiamond{x' = f(x), t' = 1 \& Q}{\left(t = \tau \land e > 0\right)}, \dbox{x' = f(x), t' = 1 \& Q \land t < \tau}{e < 0}}{\bot}}
        }
        {\lsequent{e \leq 0, \ddiamond{x' = f(x), t' = 1 \& Q}{\left(t = \tau \land e > 0\right)}, \dbox{x' = f(x), t' = 1 \& Q \land t < \tau \land e \leq 0}{e < 0}}{\bot}}
        }
        {\lsequent{e \leq 0, \ddiamond{x' = f(x), t' = 1 \& Q}{\left(t = \tau \land e > 0\right)}}{\neg\dbox{x' = f(x), t' = 1 \& Q \land t < \tau \land e \leq 0}{e \neq 0}}}
    }
    {\lsequent{}{e \leq 0 \land \ddiamond{x' = f(x), t' = 1 \& Q}{\left(t = \tau \land e > 0\right)} \rightarrow \ddiamond{x' = f(x), t' = 1 \& Q \land t < \tau \land e \leq 0}{e = 0}}}
\end{sequentdeduction}

Continuing from \textcircled{1}, the derivation crucially relies on \irref{dadj} which flows along the reverse ODE $x' = -f(x), t' = -1$ to reach a state where $t_0 < t < \tau \land e > 0$, with $t_0$ being the initial time value. Semantically, we have found a flow along $x' = f(x), t' = 1$ to a state where both $t < \tau$ and $e > 0$ are true, contradicting the fact that $e < 0$ holds along the flow while $t < \tau$. Synthesizing the argument above within $\dL$ first requires extensive use of \irref{B} to instantiate extra variables which allows us to apply axiom \irref{cont}, flowing to a state where both $t < \tau$ and $e > 0$ hold. A final application of \irref{uniqp} then gives the desired contradictions. Again for brevity, we first make the following abbreviation.
\begin{align*}
    \alpha(x, t) &\equiv x' = f(x), t' = 1\\
    -\alpha(x, t) &\equiv x' = -f(x), t' = -1\\
    A(x, t, \tau) &\equiv \dbox{\alpha(x, t) \& Q(x, t) \land t < \tau}{e(x, t) < 0}\
\end{align*}
\begin{sequentdeduction}
    \linfer[Kd]
    {\linfer[B+existsl]
        {\linfer[V+dadj]
            {\linfer[V]
                {\linfer[cont+timeFixpoint_b]
                {\linfer[uniqp+orl]
                    {\text{\textcircled{1}}
                    &
                    \text{\textcircled{2}}
                    }
                {\lsequent{e(x, t) \leq 0, A(x, t, s), e(y, \tau) > 0, \ddiamond{-\alpha(y, \tau) \& Q(y, \tau)}{\left(x = y \land \tau = t\right)}, \ddiamond{-\alpha(y, \tau) \& e(y, \tau) > 0}{\left(\tau < s\right)}}{\bot}}
                }
                {\lsequent{e(x, t) \leq 0, A(x, t, s), e(y, \tau) > 0, y = z, \tau = s, \ddiamond{-\alpha(y, \tau) \& Q(y, \tau)}{\left(x = y \land \tau = t\right)}}{\bot}}
                }
            {\lsequent{e(x, t) \leq 0, A(x, t, \tau), e(y, \tau) > 0, y = z, \tau = s, \ddiamond{-\alpha(y, \tau) \& Q(y, \tau)}{\left(x = y \land \tau = t\right)}}{\bot}}
            }
        {\lsequent{e(x, t) \leq 0, \ddiamond{\alpha(x, t) \& Q(x, t)}{\left(x = y \land y = z \land t = \tau \land \tau = s \land e(x, t) > 0\right)}, A(x, t, \tau)}{\bot}}
        }
        {\lsequent{e(x, t) \leq 0, \ddiamond{\alpha(x, t) \& Q(x, t)}{\left(\exists y = x \land \exists z = x \land \exists s = \tau \land t = \tau \land e(x, t) > 0\right)}, A(x, t, \tau)}{\bot}}
    }
    {\lsequent{e(x, t) \leq 0, \ddiamond{x' = f(x), t' = 1 \& Q(x, t)}{\left(t = \tau \land e(x, t) > 0\right)}, A(x, t, \tau)}{\bot}}
\end{sequentdeduction}
Where the open premises arising from \irref{uniqp+orl} are
\begin{align*}
    \text{\textcircled{1}} &\equiv \lsequent{e(x, t) \leq 0, \ddiamond{-\alpha(y, \tau) \& Q(y, \tau)}{\left(x = y \land t = \tau \land \ddiamond{-\alpha(y, \tau) \& e(y, \tau) > 0}{\left(\tau < s\right)}\right)}}{\bot}\\
    \text{\textcircled{2}} &\equiv \lsequent{A(x, t, s), \ddiamond{-\alpha(y, \tau) \& e(y, \tau) > 0}{\left(\tau < s \land \ddiamond{-\alpha(y, \tau) \& Q(y, \tau)}{\left(x = y \land t = \tau\right)}\right)}}{\bot}
\end{align*}
Intuitively, \textcircled{1} yields a contradiction since the first diamond modality flows to a state where $e(y, \tau) = e(x, t) \leq 0$ is true, but the second diamond modality \emph{requires} $e(y, \tau) > 0$ as a domain constraint. Since the domain constraint is not satisfied, the overall formula is indeed false. For \textcircled{2}, another application of \irref{dadj} to the inner modality gives a flow that contradicts $A(x, t, s)$. We deal with \textcircled{1} first:

\begin{sequentdeduction}
    \linfer[V]
        {\linfer[Kd+V]
            {\linfer[diamond+notl]
                {\linfer[dx]
                {\linfer[implyr]
                    {\linfer[qear]
                        {\linfer
                            {\lclose}
                            {\lsequent{\bot}{\tau \geq s \land \dbox{-\alpha(y, \tau) \& e(y, \tau) > 0}{\left(\tau \geq s\right)}}}
                        }
                        {\lsequent{e(y, \tau) \leq 0, e(y, \tau) > 0}{\tau \geq s \land \dbox{-\alpha(y, \tau) \& e(y, \tau) > 0}{\left(\tau \geq s\right)}}}
                    }
                    {\lsequent{e(y, \tau) \leq 0}{e(y, \tau) > 0 \rightarrow \tau \geq s \land \dbox{-\alpha(y, \tau) \& e(y, \tau) > 0}{\left(\tau \geq s\right)}}}
                }
                {\lsequent{e(y, \tau) \leq 0}{\dbox{-\alpha(y, \tau) \& e(y, \tau) > 0}{\left(\tau \geq s\right)}}}
            }
            {\lsequent{e(y, \tau) \leq 0, \ddiamond{-\alpha(y, \tau) \& e(y, \tau) > 0}{\left(\tau < s\right)}}{\bot}}
        }
        {\lsequent{e(x, t) \leq 0, \ddiamond{-\alpha(y, \tau) \& Q(y, \tau)}{\left(e(y, \tau) \leq 0 \land \ddiamond{-\alpha(y, \tau) \& e(y, \tau) > 0}{\left(\tau < s\right)}\right)}}{\bot}}
        }
    {\lsequent{e(x, t) \leq 0, \ddiamond{-\alpha(y, \tau) \& Q(y, \tau)}{\left(x = y \land t = \tau \land \ddiamond{-\alpha(y, \tau) \& e(y, \tau) > 0}{\left(\tau < s\right)}\right)}}{\bot}}
\end{sequentdeduction}

Continuing with the proof of \textcircled{2}, we have:

\begin{sequentdeduction}
    \linfer[V+dWd]
    {\linfer[Kd+V]
        {\linfer[dadj]
            {\linfer[mont_f]
                {\linfer[drd+V]
                    {\linfer[diamond]
                        {\linfer[notl+id]
                            {\lclose}
                            {\lsequent{A(x, t, s), \neg A(x, t, s)}{\bot}}
                        }
                        {\lsequent{A(x, t, s), \ddiamond{\alpha(x, t) \& Q(x, t) \land t < s}{\left(e(x, t) > 0\right)}}{\bot}}
                    }
                    {\lsequent{A(x, t, s), e(y, \tau) > 0, \tau < s, \ddiamond{\alpha(x, t) \& Q(x, t) \land t \leq \tau}{\left(x = y \land t = \tau\right)}}{\bot}}
                }
                {\lsequent{A(x, t, s), e(y, \tau) > 0, \tau < s, \ddiamond{\alpha(x, t) \& Q(x, t)}{\left(x = y \land t = \tau\right)}}{\bot}}
            }
            {\lsequent{A(x, t, s), e(y, \tau) > 0, \tau < s, \ddiamond{-\alpha(y, \tau) \& Q(y, \tau)}{\left(x = y \land t = \tau\right)}}{\bot}}
        }
        {\lsequent{\ddiamond{-\alpha(y, \tau) \& e(y, \tau) > 0}{\left(A(x, t, s) \land e(y, \tau) > 0 \land \tau < s \land \ddiamond{-\alpha(y, \tau) \& Q(y, \tau)}{\left(x = y \land t = \tau\right)}\right)}}{\bot}}
    }
    {\lsequent{A(x, t, s), \ddiamond{-\alpha(y, \tau) \& e(y, \tau) > 0}{\left(\tau < s \land \ddiamond{-\alpha(y, \tau) \& Q(y, \tau)}{\left(x = y \land t = \tau\right)}\right)}}{\bot}}
\end{sequentdeduction}
This concludes the proof of \rref{lem:lemma axioms}. Note that for \irref{ivt}, axiom \irref{B} allows us to relax the condition of $t = \tau$ in the antecedent to $t \leq \tau$ without loss of provability. 
\end{proof}

\begin{proof}[Proof of \rref{thm:axioms for step liveness}]
    \irref{dualRight}: To derive this axiom, we first utilize \irref{bdg} to cut in the formula 
\[\ddiamond{x'= f(x), t' = 1 \& t \leq \tau}{t = \tau}\] 
after which an application of \irref{drd} will give the desired outcome.

\begin{sequentdeduction}
    \linfer[implyr+cut]
    {
    \linfer[]
        {\text{\textcircled{1}}}
    {\lsequent{t \leq \tau, \dbox{x' = f(x), t' = 1 \& t \leq \tau}{B(x)}} {\ddiamond{x' = f(x), t' = 1 \& t \leq \tau}{t = \tau}}}
    &
    {\text{\textcircled{2}}}
    }
    {\lsequent{}{t \leq \tau \land \dbox{x' = f(x), t' = 1 \& t \leq \tau}{B(x)} \rightarrow \ddiamond{x' = f(x), t' = 1\& B(x)}{t = \tau}}}
\end{sequentdeduction}
Where the premises arising from \irref{cut} are
\begin{align*}
    \text{\textcircled{1}} &\equiv \lsequent{t \leq \tau, \dbox{x' = f(x), t' = 1 \& t \leq \tau}{B(x)}} {\ddiamond{x' = f(x), t' = 1 \& t \leq \tau}{t = \tau}}\\
    \text{\textcircled{2}} &\equiv \lsequent{\dbox{x' = f(x), t' = 1 \& t \leq \tau}{B(x)}, \ddiamond{x' = f(x), t' = 1 \& t \leq \tau} {t = \tau}} {\ddiamond{x' = f(x), t' = 1\& B(x)}{t = \tau}}
\end{align*}
To prove \textcircled{1}, axiom \irref{bdg} reduces the problem to proving $\dbox{x' = f(x), t' = 1 \& t \leq \tau}{\norm{x}^2 \leq p(t)}$, where $p(t)$ is some polynomial in terms of $t$. Since $B(x)$ is a bounded set, the $\folr$ formula $\exists D \forall x \left(B(x) \rightarrow \norm{x}^2 \leq D\right)$ is valid, and therefore $p(t)$ can be simply be chosen to be $p(t) \equiv D$ for some $D \in \Q^+$ a witness of the $\folr$ formula. Expressing this argument in sequent form gives:

\begin{sequentdeduction}
    \linfer[cut+existsl+qear]
    {
        \linfer[bdgd]
            {
                {\text{\textcircled{3}}}
                &
                \linfer[evolved]
                    {\lclose}
                {\lsequent{t \leq \tau}{\ddiamond{t' = 1 \& t \leq \tau}{t = \tau}}}
            }
        {\lsequent{t \leq \tau, \forall x \left(B(x) \rightarrow \norm{x}^2 \leq D\right), \dbox{x' = f(x), t' = 1 \& t \leq \tau}{B(x)}} {\ddiamond{x' = f(x), t' = 1 \& t \leq \tau}{t = \tau}}}
    }
    {\lsequent{t \leq \tau, \dbox{x' = f(x), t' = 1 \& t \leq \tau}{B(x)}} {\ddiamond{x' = f(x), t' = 1 \& t \leq \tau}{t = \tau}}}
\end{sequentdeduction}

Where
\[\text{\textcircled{3}} \equiv \lsequent{\forall x \left(B(x) \rightarrow \norm{x}^2 \leq D\right), \dbox{x' = f(x), t' = 1 \& t \leq \tau}{B(x)}} {\dbox{x' = f(x), t' = 1 \& t \leq \tau}{\norm{x}^2 \leq D}}\]

\textcircled{3} is proven by first applying \irref{dC} to cut in the domain constraint $B(x)$ to the succedent, after which an application of \irref{dW} completes the proof since the formula $\forall x \left(B(x) \rightarrow \norm{x}^2 \leq D\right)$ is independent of the ODE $x' = f(x), t' = 1$.

\begin{sequentdeduction}
    \linfer[dC]
    {\linfer[dW]
    {\linfer[qear]
        {\lclose}
        {\lsequent{\forall x \left(B(x) \rightarrow \norm{x}^2 \leq D\right), t \leq \tau \land B(x)}{\norm{x}^2 \leq D}}
    }
    {\lsequent{\forall x \left(B(x) \rightarrow \norm{x}^2 \leq D\right)}{\dbox{x' = f(x), t' = 1 \& t \leq \tau \land B(x)}{\norm{x}^2 \leq D}}}
    }  
    {\lsequent{\forall x \left(B(x) \rightarrow \norm{x}^2 \leq D\right), \dbox{x' = f(x), t' = 1 \& t \leq \tau}{B(x)}} {\dbox{x' = f(x), t' = 1 \& t \leq \tau}{\norm{x}^2 \leq D}}}
\end{sequentdeduction}
The open premise \textcircled{3} has been proven and therefore so has \textcircled{1}. \textcircled{2} is now proved utilizing \irref{drd} to add in the domain constraint of $B(x)$.

\begin{sequentdeduction}
    \linfer[drd]
    {
    \linfer[drw]
    {\lclose}
    {\lsequent{\ddiamond{x' = f(x), t' = 1 \& t \leq \tau \land B(x)} {t = \tau}} {{\ddiamond{x' = f(x), t' = 1\& B(x)}{t = \tau}}}}
    }
    {\lsequent{\dbox{x' = f(x), t' = 1 \& t \leq \tau}{B(x)}, \ddiamond{x' = f(x), t' = 1 \& t \leq \tau} {t = \tau}} {\ddiamond{x' = f(x), t' = 1\& B(x)}{t = \tau}}}
\end{sequentdeduction}
This completes the proof of \irref{dualRight}.

\irref{dualLeft}: For brevity, first make the following abbreviations:
\begin{align*}
    A \equiv &\ddiamond{x' = f(x), t' = 1 \& Q}{t \geq \tau}\\
    B \equiv &\ddiamond{x' = f(x), t' = 1 \& t \leq \tau}{\neg Q}
\end{align*}
Axiom \irref{dualLeft} says that if there is some flow of the ODE $x' = f(x), t' = 1 \& Q$ where time surpasses $\tau$, then \emph{every} flow of this ODE before time $t = \tau$ will remain within the domain constraint $Q$. Alternatively, this axiom is precisely the uniqueness property of ODE flows. Consequently, our derivation will follow the classical soundness argument. If the implication is not valid, then there are two disjoint flows of the ODE, contradicting \irref{uniqp}. For brevity, we also make the following abbreviations:

\begin{sequentdeduction}
    \linfer[implyr+diamond+notr]
    {\linfer[uniqp]
        {\linfer[orl]
            {
            \linfer
            {\text{\textcircled{1}}}
            {\lsequent{\ddiamond{x' = f(x), t' = 1 \& Q \land t \leq \tau}{\left(t \geq \tau \land B\right)}} {\bot}}
            &
            \linfer[dWd]
            {
                \linfer[V]
                {\lclose}
                {\lsequent{\ddiamond{x' = f(x), t' = 1 \& Q \land t \leq \tau}{\left(\neg Q \land Q \land A\right)}} {\bot}}
            }
            {\lsequent{\ddiamond{x' = f(x), t' = 1 \& Q \land t \leq \tau}{\left(\neg Q \land A\right)}} {\bot}}
            }
        {\lsequent{\ddiamond{x' = f(x), t' = 1 \& Q \land t \leq \tau}{\left(t \geq \tau \land B\right)} \lor \ddiamond{x' = f(x), t' = 1 \& Q \land t \leq \tau}{\left(\neg Q \land A\right)}} {\bot}}
        }
    {\lsequent{\ddiamond{x' = f(x), t' = 1\& Q}{t \geq \tau}, \ddiamond{x' = f(x), t' = 1 \&t \leq \tau}{\neg Q}}{\bot}}    
    }
    {\lsequent{}{\ddiamond{x' = f(x), t' = 1\& Q}{t \geq \tau} \rightarrow \dbox{x' = f(x), t' = 1 \&t \leq \tau}{Q}}}
\end{sequentdeduction}
The right branch arising from \irref{orl} closes easily by noting $\neg Q \land Q \equiv \bot$ and applying axiom \irref{dWd}. For \textcircled{1}, $B$ says there is some flow along $x' = f(x), t' = 1 \& t \leq \tau$ reaching $\neg Q$. But since the first diamond modality already reaches a state where $t \geq \tau$, there cannot possibly be any non-trivial evolution, and therefore $Q$ must remain true, a contradiction.

\begin{sequentdeduction}
    \linfer[dWd]
    {\linfer[Kd+timeFixpoint]
        {
            \linfer[V]
            {\lclose}
            {\lsequent{\ddiamond{x' = f(x), t' = 1 \& Q \land t \leq \tau}{\left(t = \tau \land Q \land \neg Q\right)}} {\bot}}
        }
    {\lsequent{\ddiamond{x' = f(x), t' = 1 \& Q \land t \leq \tau}{\left(t = \tau \land Q \land \ddiamond{x' = f(x), t' = 1 \& t \leq \tau}{\neg Q}\right)}} {\bot}}
    }
    {\lsequent{\ddiamond{x' = f(x), t' = 1 \& Q \land t \leq \tau}{\left(t \geq \tau \land B\right)}} {\bot}}
\end{sequentdeduction}
where we used \irref{timeFixpoint} by negating the succedent, resulting in
\[t = \tau \rightarrow\left(\ddiamond{x' = f(x), t' = 1 \& t \leq \tau}{\neg Q} \leftrightarrow \neg Q\right)\]
this completes the derivation of \irref{dualLeft}.

\irref{stepEx}: Derivation of \irref{stepEx} mostly follows from axioms \irref{ivt} and \irref{dinv}. Mathematically, the boundedness requirement on $f(x)$ implies $\norm{x(t) - x_0} \leq M(t - t_0)$, essentially the multivariate mean value theorem where $x(t)$ is the flow of $x' = f(x), x(t_0) = x_0$ at time $t$. By \irref{dualRight} and \irref{ivt}, if the axiom does not hold, then there exists some point where the bound $\norm{x(t) - x_0} \leq M(t - t_0)$ is violated, resulting in a contradiction. The derivation is as follows ($\max_{y \in B[x_0, R]} \norm{f(y)}^2 \leq M^2$ abbreviates $\forall y \left(y \in B[x_0, R] \rightarrow \norm{f(y)}^2 \leq M^2\right)$).
\begin{sequentdeduction}
    \linfer[implyr+implyr]
    {\linfer[dualRight]
        {\linfer[diamond+notr]
            {
                \linfer[ivt]
                    {\linfer[cut]
                        {\text{\textcircled{1}}
                            &
                        \text{\textcircled{2}}
                        }
                    {\lsequent{\max_{y \in B[x_0, R]} \norm{f(y)}^2 \leq M^2, x = x_0, t = t_0, \ddiamond{x' = f(x), t' = 1 \& t < t_0 + \frac{R}{M} \land x \in B[x_0, R]}{\norm{x - x_0}^2 = R^2}} {\bot}}
                    }
                {\lsequent{\max_{y \in B[x_0, R]} \norm{f(y)}^2 \leq M^2, x = x_0, t = t_0, \ddiamond{x' = f(x), t' = 1 \& t \leq t_0 + \frac{R}{M}}{\norm{x - x_0}^2 > R^2}} {\bot}}
            }
        {\lsequent{\max_{y \in B[x_0, R]} \norm{f(y)}^2 \leq M^2, x = x_0, t = t_0} {\dbox{x' = f(x), t' = 1 \& t \leq t_0 + \frac{R}{M}}{x \in B[x_0, R]}}}
        }
    {\lsequent{\max_{y \in B[x_0, R]} \norm{f(y)}^2 \leq M^2, x = x_0, t = t_0} {\ddiamond{x' = f(x), t' = 1 \& x \in B[x_0, R]}{t \geq t_0 + \frac{R}{M}}}}
    }
    {\lsequent{}{\max_{y \in B[x_0, R]} \norm{f(y)}^2 \leq M^2 \rightarrow \left(x = x_0 \land t = t_0 \rightarrow \ddiamond{x' = f(x), t' = 1 \& x \in B[x_0, R]}{t \geq t_0 + \frac{R}{M}}\right)}}
\end{sequentdeduction}

Where we are cutting in the differential invariant representing the multivariate mean value theorem at the last step, giving:
\begin{align*}
    \alpha(x, t) &\equiv x' = f(x), t' = 1 \& t < t_0 + \frac{R}{M} \land x \in B[x_0, R]\\
    I(x, t) &\equiv D(x, t) \rightarrow \dbox{x' = f(x), t' = 1 \& x \in B[x_0, R] \land \max_{y \in B[x_0, R]} \norm{f(y)}^2 \leq M^2}{D(x, t)}\\
    D(x, t) &\equiv \norm{x - x_0}^2 \leq M^2(t - t_0)^2 \land t \geq t_0\\
    \text{\textcircled{1}} &\equiv \lsequent{x = x_0, t = t_0, \max_{y \in B[x_0, R]} \norm{f(y)}^2 \leq M^2, I(x, t), \ddiamond{\alpha(x, t)}{\norm{x - x_0}^2 = R^2}} {\bot}\\
    \text{\textcircled{2}} &\equiv \lsequent{} {I(x, t)}
\end{align*}

\textcircled{2} is derived first. By the completeness of differential invariants, it is suffices to establish the validity of $I(x, t)$ semantically. To show that \textcircled{2} holds semantically, let $\omega \in \S$ be some arbitrary state where $\omega \models D$. Let $\phi : [0, \tau] \to \S$ be any solution satisfying the ODE $x' = f(x), t' = 1$ with $\phi(0) = \omega$ and $\phi(t) \models x' = f(x), t' = 1 \land x \in B[x_0, R] \land \max_{y \in B[x_0, R]} \norm{f(y)}^2 \leq M^2$ for all $t \in [0, \tau]$. We want to show that $\phi(\tau) \models D$ as well. To this end, let us denote $x(t) = \phi(t)(x)$ as the trajectory of $x$ under the given ODE. The triangle inequality together with the multivariate mean value theorem gives
\[\norm{x(\tau) - x_0} \leq \norm{x(0) - x_0} + \norm{x(\tau) - x(0)} \leq M(\phi(0)(t) - t_0) + \max_{\zeta \in [t_0, \tau]}\norm{x'(\zeta)}\tau\]
Note that taking square roots implicitly used the assumption $\phi(0) \models D$ and therefore $\phi(0)(t) - t_0 \geq 0$. We will write $x_0$ (and similarly $t_0$) instead of $\phi(s)(x_0)$ for all $s \in [0, \tau]$, as $x_0, t_0$ are just constants along the given ODE, and will therefore not vary along $\phi$. Since $\phi \models x' = f(x) \land x \in B[x_0, R] \land \max_{y \in B[x_0, R]} \norm{f(y)}^2 \leq M^2$, this gives the bound $\max_{\zeta \in [t_0, \tau]}\norm{x'(\zeta)} \leq M$ and consequently the following bound on $\norm{x(\tau) - x_0}$
\begin{align*}
    \norm{x(\tau) - x_0} \leq M(\phi(0)(t) - t_0) + \max_{\zeta \in [t_0, \tau]}\norm{x'(\zeta)}\tau &\leq M(\phi(0)(t) - t_0) + M\tau\\
    &= M(\tau + \phi(0)(t) - t_0)
\end{align*}
Finally, $\phi \models t' = 1$ implies $\phi(s)(t) = s + \phi(0)(t)$ since the solution to $t' = 1$ is just $t(s) = s + t(0)$. In particular, this yields $\tau + \phi(0)(t) = \phi(\tau)(t)$, so we have
\[\norm{x(\tau) - x_0} \leq M(\tau + \phi(0)(t) - t_0) = M(\phi(\tau)(t) - t_0)\]
Squaring both sides then gives the desired claim of $\phi(\tau) \models D$, proving $I(x, t)$ to be valid. Consequently \textcircled{2} closes by a single application of axiom \irref{dinv}.

For premise \textcircled{1}:

\begin{sequentdeduction}
    \linfer[implyl]
    {\linfer[dWd]
        {\linfer[drw]
            {\linfer[Kd+qear]
            {\linfer[diamond+notl+id]
                {\lclose}
            {\lsequent{\dbox{x' = f(x), t' = 1 \& x \in B[x_0, R]}{D(x, t)}, \ddiamond{x' = f(x), t' = 1 \& x \in B[x_0, R]}{\neg D(x, t)}}{\bot}}
            }
        {\lsequent{\dbox{x' = f(x), t' = 1 \& x \in B[x_0, R]}{D(x, t)}, \ddiamond{x' = f(x), t' = 1 \& x \in B[x_0, R]}{\left(\norm{x - x_0}^2 = R^2 \land t - t_0 < \frac{R}{M}\right)}}{\bot}}}
            {\lsequent{\dbox{x' = f(x), t' = 1 \& x \in B[x_0, R]}{D(x, t)}, \ddiamond{\alpha(x, t)}{\left(\norm{x - x_0}^2 = R^2 \land t - t_0 < \frac{R}{M}\right)}}{\bot}}
        }
    {\lsequent{\dbox{x' = f(x), t' = 1 \& x \in B[x_0, R]}{D(x, t)}, \ddiamond{\alpha(x, t)}{\norm{x - x_0}^2 = R^2}}{\bot}}
    }
    {\lsequent{x = x_0, t = t_0,\max_{y \in B[x_0, R]} \norm{f(y)}^2 \leq M^2, I(x, t), \ddiamond{\alpha(x, t)}{\norm{x - x_0}^2 = R^2}} {\bot}}
\end{sequentdeduction}

This completes the derivation of \irref{stepEx}.

\irref{stepExt}: The main idea in deriving this axiom is to note that \irref{ivt} and \irref{uniqp} allows one to decompose diamond modalities into different time steps, from which the premises allow us to complete the proof. We denote the premises as \textcircled{a}, \textcircled{b}, \textcircled{c} and will indicate when they are used during the derivation, where:
\begin{align*}
    \text{\textcircled{a}} &\equiv \lsequent{\Gamma_1}{\dbox{x' = f(x), t' = 1 \& t \leq t_0}{P_1}}\\
    \text{\textcircled{b}} &\equiv \lsequent{\Gamma_2}{\dbox{x' = f(x), t' = 1 \& t \leq t_0+t_1}{P_2}}\\
    \text{\textcircled{c}} &\equiv \lsequent{t = t_0, P_1}{\Gamma_2}
\end{align*}
\begin{sequentdeduction}
    \linfer[diamond+notr]
    {\linfer[dor]
        {
        {\text{\textcircled{1}}}
        &
        {\text{\textcircled{2}}}
        }
    {\lsequent{t \leq t_0, \Gamma_1, \ddiamond{x' = f(x), t' = 1 \& t \leq t_0 + t_1}{\left((t \leq t_0 \land \neg P_1) \lor (t > t_0 \land \neg P_2)\right)}}{\bot}}
    }
    {\lsequent{t \leq t_0, \Gamma_1}{\dbox{x' = f(x), t' = 1 \& t \leq t_0 + t_1}{\left((t \leq t_0 \rightarrow P_1) \land (t > t_0 \rightarrow P_2)\right)}}}
\end{sequentdeduction}
With the open premises being
\[\text{\textcircled{1}} \equiv \lsequent{t \leq t_0, \Gamma_1, \ddiamond{x' = f(x), t' = 1 \& t \leq t_0 + t_1}{\left(t \leq t_0 \land \neg P_1\right)}}{\bot}\]
\[\text{\textcircled{2}} \equiv \lsequent{t \leq t_0, \Gamma_1, \ddiamond{x' = f(x), t' = 1 \& t \leq t_0 + t_1}{\left(t > t_0 \land \neg P_2\right)}}{\bot}\]
Premise \textcircled{1} is proven first, noting that the diamond modality directly contradicts \textcircled{a} after applying \irref{mont_f} to add in the domain constraint of $t \leq t_0$. 

\begin{sequentdeduction}
    \linfer[mont_f]
    {\linfer[drw]
        {\linfer[Kd]
            {\linfer[cut]
            {\linfer[diamond+notl+id]
                {\lclose}
            {\lsequent{t \leq t_0, \Gamma_1, \ddiamond{x' = f(x), t' = 1 \& t \leq t_0}{\neg P_1}, \dbox{x' = f(x), t' = 1 \& t \leq t_0}{P_1}}{\bot}}
            &
            {\text{\textcircled{3}}}
            }
        {\lsequent{t \leq t_0, \Gamma_1, \ddiamond{x' = f(x), t' = 1 \& t \leq t_0}{\neg P_1}}{\bot}}
            }
            {\lsequent{t \leq t_0, \Gamma_1, \ddiamond{x' = f(x), t' = 1 \& t \leq t_0}{(t \leq t_0 \land \neg P_1)}}{\bot}}
        }
    {\lsequent{t \leq t_0, \Gamma_1, \ddiamond{x' = f(x), t' = 1 \& t \leq t_0 + t_1 \land t \leq t_0}{\left(t \leq t_0 \land \neg P_1\right)}}{\bot}}
    }
    {\lsequent{t \leq t_0, \Gamma_1, \ddiamond{x' = f(x), t' = 1 \& t \leq t_0 + t_1}{\left(t \leq t_0 \land \neg P_1\right)}}{\bot}}
\end{sequentdeduction}
And \textcircled{3} is
\[\text{\textcircled{3}} \equiv \lsequent{\Gamma_1}{\dbox{x' = f(x), t' = 1 \& t \leq t_0}{P_1} \equiv \text{\textcircled{a}}}\]
In other words, \textcircled{1} derives with premise \textcircled{a}. We prove premise \textcircled{2} next, by first applying axiom \irref{ivt} on the term $e \equiv t - t_0$.

\begin{sequentdeduction}
    \linfer[ivt]
    {\linfer[uniqp+orl]
        {
            {\text{\textcircled{4}}}
            &
            {\text{\textcircled{5}}}
        }
    {\lsequent{t \leq t_0, \Gamma_1, \ddiamond{x' = f(x), t' = 1 \& t \leq t_0 + t_1}{\left(t > t_0 \land \neg P_2\right)}, \ddiamond{x' = f(x), t' = 1 \& t \leq t_0}{t = t_0}}{\bot}}
    }
    {\lsequent{t \leq t_0, \Gamma_1, \ddiamond{x' = f(x), t' = 1 \& t \leq t_0 + t_1}{\left(t > t_0 \land \neg P_2\right)}}{\bot}}
\end{sequentdeduction}
Where the open premises are the ones arising from the disjunction in \irref{uniqp}, we have:
\begin{align*}
    \text{\textcircled{4}} &\equiv \lsequent{\Gamma_1, \ddiamond{x' = f(x), t' = 1 \& t \leq t_0}{\left(t = t_0 \land \ddiamond{x' = f(x), t' = 1 \& t \leq t_0 + t_1}{\left(t > t_0 \land \neg P_2\right)}\right)}}{\bot}\\
    \text{\textcircled{5}} &\equiv \lsequent{\ddiamond{x' = f(x), t' = 1 \& t \leq t_0 + t_1}{\left(t > t_0 \land \neg P_2 \land \ddiamond{x' = f(x), t' = 1 \& t \leq t_0}{t = t_0}\right)}}{\bot}
\end{align*}
Premise \textcircled{5} is resolved first, noticing that the inequality $t > t_0$ contradicts with the domain constraint of the second diamond modality, so axiom \irref{dx} yields the desired derivation.
\begin{sequentdeduction}
    \linfer[V]
        {\linfer[diamond+notl]
            {\linfer[dx]
                {\linfer[implyr]
                    {\linfer[qear]
                        {\lclose}
                        {\lsequent{t > t_0, t \leq t_0,\neg P_2}{t = t_0 \land \dbox{x' = f(x), t' = 1 \& t \leq t_0}{t = t_0}}}
                    }
                {\lsequent{t > t_0, \neg P_2}{t \leq t_0 \rightarrow t = t_0 \land \dbox{x' = f(x), t' = 1 \& t \leq t_0}{t = t_0}}}
                }
                {\lsequent{t > t_0, \neg P_2}{\dbox{x' = f(x), t' = 1 \& t \leq t_0}{t = t_0}}}
            }
        {\lsequent{t > t_0, \neg P_2, \ddiamond{x' = f(x), t' = 1 \& t \leq t_0}{t = t_0}}{\bot}}
        }
    { \lsequent{\ddiamond{x' = f(x), t' = 1 \& t \leq t_0 + t_1}{\left(t > t_0 \land \neg P_2 \land \ddiamond{x' = f(x), t' = 1 \& t \leq t_0}{t = t_0}\right)}}{\bot}}
\end{sequentdeduction}

We now prove the remaining premise \textcircled{4}. Intuitively speaking, the first diamond modality reaches some state where $t = t_0, P_1$ are both true (by premise \textcircled{a}), from which premise \textcircled{c} implies the truth of $\Gamma_2$, and therefore premise \textcircled{b} gives a contradiction with the second modality. This gives the following derivation, where each of \textcircled{a}, \textcircled{b}, \textcircled{c} indicates the corresponding assumption being cut in:

\begin{sequentdeduction}
    \linfer[cut+drd]
    {\linfer[dWd]
        {\linfer[V]
            {\linfer[cut]
                {\linfer[cut+implyl]
                    {\linfer[diamond+notl]
                        {
                            \linfer[K]
                                {\lclose}
                            {\lsequent{\dbox{x' = f(x), t' = 1 \& t \leq t_0 + t_1}{P_2}}{\dbox{x' = f(x), t' = 1 \& t \leq t_0 + t_1}{\left(t \leq t_0 \lor P_2\right)}}}
                        }
                        {\lsequent{\dbox{x' = f(x), t' = 1 \& t \leq t_0 + t_1}{P_2}, \ddiamond{x' = f(x), t' = 1 \& t \leq t_0 + t_1}{\left(t > t_0 \land \neg P_2\right)}}{\bot}}
                        &
                        {\text{\textcircled{b}}}
                    }
                    {\lsequent{\Gamma_2, \ddiamond{x' = f(x), t' = 1 \& t \leq t_0 + t_1}{\left(t > t_0 \land \neg P_2\right)}}{\bot}}
                    &
                    {\text{\textcircled{c}}}
                }
                {\lsequent{t = t_0, P_1, \ddiamond{x' = f(x), t' = 1 \& t \leq t_0 + t_1}{\left(t > t_0 \land \neg P_2\right)}}{\bot}}
            }
        {\lsequent{\ddiamond{x' = f(x), t' = 1 \& t \leq t_0 \land P_1}{\left(t = t_0 \land P_1 \land \ddiamond{x' = f(x), t' = 1 \& t \leq t_0 + t_1}{\left(t > t_0 \land \neg P_2\right)}\right)}}{\bot}}
        }
    {\lsequent{\Gamma_1, \ddiamond{x' = f(x), t' = 1 \& t \leq t_0 \land P_1}{\left(t = t_0 \land \ddiamond{x' = f(x), t' = 1 \& t \leq t_0 + t_1}{\left(t > t_0 \land \neg P_2\right)}\right)}}{\bot}}
        &
    {\text{\textcircled{a}}}
    }
    {\lsequent{\Gamma_1, \ddiamond{x' = f(x), t' = 1 \& t \leq t_0}{\left(t = t_0 \land \ddiamond{x' = f(x), t' = 1 \& t \leq t_0 + t_1}{\left(t > t_0 \land \neg P_2\right)}\right)}}{\bot}}
\end{sequentdeduction}
This completes the derivation of \irref{stepExt} and thus also completing the proof of \rref{thm:axioms for step liveness}.
\end{proof}

Finally, we finish the proof of \rref{ex: symbolic existence}.

\begin{proof}[Proof of \rref{ex: symbolic existence}]
    Let $\eps \in \Q^+$ be arbitrary, and assume without loss of generality that $\eps < 1$. The main idea of the derivation is to iteratively apply axiom \irref{stepEx} to obtain (shrinking) iterates of existence intervals using $R$ = $\alpha \abs{x_0}$ for some suitably chosen $\alpha \in \Q^+$, these existence intervals can be chained together using \irref{stepExt}, giving the desired proof. First pick $\alpha \in \Q^+$ sufficiently small and $N \in \N$ sufficiently large such that the following hold
    \begin{itemize}
        \item $\frac{1}{\alpha + 1} - \frac{1}{\alpha N + 1} > 1 - \eps$
        \item $\frac{1}{\alpha + 1} - \frac{1}{\alpha N + 1} \geq \frac{1}{(\alpha + 1)^3}$. 
    \end{itemize}
    such choices are possible since $N \in \N$ is allowed to be arbitrarily large and dependent on $\alpha$. The derivation will use $N$ steps of \irref{stepEx} to show that $x(t)$ is bounded in $B[x_0, \alpha N x_0]$ for $t \in \left[0, (1 - \eps)\left(\frac{1}{x_0} - \frac{1}{3x_0^3}\right)\right)$ from which the desired claim concludes by axiom \irref{dualRight}. Note that the bound $\max_{y \in B(x_0, nx_0)} \abs{x^2 + 1} \leq (n + 1)^2x_0^2 + 1$ holds for all $n \in \N$. The derivation first begins by handling the trivial case where $\frac{1}{x_0} - \frac{1}{3x_0^3} < 0$ holds. 
    
    \begin{sequentdeduction}
        \linfer[implyr+cut+qear]
            {\linfer[orl]
                {\linfer[diamond+notr]
                    {\linfer[dx]
                        {\linfer[qear]
                            {\lclose}
                            {\lsequent{x = x_0, t = 0, x_0 > 0, \frac{1}{x_0} - \frac{1}{3x_0^3} < 0, t < (1 - \eps)\left(\frac{1}{x_0} - \frac{1}{3x_0^3}\right)}{\bot}}
                        }
                        {\lsequent{x = x_0, t = 0, x_0 > 0, \frac{1}{x_0} - \frac{1}{3x_0^3} < 0, \dbox{x' = x^2 + 1, t' = 1}{t < (1 - \eps)\left(\frac{1}{x_0} - \frac{1}{3x_0^3}\right)}}{\bot}}
                    }
                    {\lsequent{x = x_0, t = 0, x_0 > 0, \frac{1}{x_0} - \frac{1}{3x_0^3} < 0}{\ddiamond{x' = x^2 + 1, t' = 1}{t \geq (1 - \eps)\left(\frac{1}{x_0} - \frac{1}{3x_0^3}\right)}}}
                    &
                    {\text{\textcircled{1}}}
                }
            {\lsequent{x = x_0, t = 0, x_0 > 0, \frac{1}{x_0} - \frac{1}{3x_0^3} < 0 \lor \frac{1}{x_0} - \frac{1}{3x_0^3} \geq 0}{\ddiamond{x' = x^2 + 1, t' = 1}{t \geq (1 - \eps)\left(\frac{1}{x_0} - \frac{1}{3x_0^3}\right)}}}
            }
        {\lsequent{}{x = x_0 \land t = 0 \land x_0 > 0 \rightarrow \ddiamond{x' = x^2 + 1, t' = 1}{t \geq (1 - \eps)\left(\frac{1}{x_0} - \frac{1}{3x_0^3}\right)}}}
    \end{sequentdeduction}
    The remaining premise \textcircled{1} represents the case where $\frac{1}{x_0} - \frac{1}{3x_0^3} \geq 0$
    \[\text{\textcircled{1}} \equiv \lsequent{x = x_0, t = 0, x_0 > 0, \frac{1}{x_0} - \frac{1}{3x_0^3} \geq 0}{\ddiamond{x' = x^2 + 1, t' = 1}{t \geq (1 - \eps)\left(\frac{1}{x_0} - \frac{1}{3x_0^3}\right)}}\]
    The derivation of \textcircled{1} begins with axiom \irref{dualRight} and the bounded set $B[x_0, \alpha N x_0]$ (closed ball centered around $x_0$ with radius $\alpha N x_0$), followed by repeated applications of \irref{stepEx} and \irref{stepExt}. It uses the following constructs:
    \begin{itemize}
        \item Define the sequence $\{t_{n}\}_{0 \leq n \leq N}$ recursively with $t_0 = 0$, $t_{n} = t_{n - 1} + \frac{\alpha x_0}{(\alpha n + 1)^2x_0^2 + 1}$.
        \item For each $0 \leq n \leq N$, define the ODEs
        \begin{align*}
            \gamma_n &\equiv x' = x^2 + 1, t' = 1 \& t \leq t_n\\
            \beta_n &\equiv x' = x^2 + 1, t' = 1 \& x \in B[x_0, \alpha nx_0]
        \end{align*}
        \item For each $1 \leq n \leq N$, define the formula 
        \[\Gamma_n \equiv x_0 > 0 \land x \in B[x_0, \alpha nx_0] \land t = t_n\]
        where $\alpha nx_0$ denotes standard multiplication. Note that crucially the upper bound $N$ and the parameter $\alpha$ are constant, fixed values. 
    \end{itemize}
    Note that formulas of the form 
    \[\phi_n \equiv \Gamma_n \rightarrow \dbox{\gamma_{n + 1}}{(x_0 > 0 \land x \in B[x_0, \alpha(n + 1)x_0])}\]
    are valid and derivable from axiom \irref{stepEx} for every $0 \leq n \leq N - 1$, the proof is as follows. 

    \begin{sequentdeduction}
        \linfer[implyr+V]
            {\linfer[cut+stepEx+drw]
                {
                    \linfer[dualLeft]
                    {\linfer[id]
                        {\lclose}
                    {\lsequent{\dbox{\gamma_{n + 1}}{x \in B[x_0, \alpha(n + 1)x_0]}}{\dbox{\gamma_{n + 1}}{x \in B[x_0, \alpha(n + 1)x_0]}}}
                    }
                    {\lsequent{\ddiamond{\beta_{n + 1}}{t \geq t_{n + 1}}}{\dbox{\gamma_{n + 1}}{x \in B[x_0, \alpha(n + 1)x_0]}}}
                    &
                    {\text{\textcircled{2}}}
                }
                {\lsequent{\Gamma_n}{\dbox{\gamma_{n + 1}}{x \in B[x_0, \alpha(n + 1)x_0]}}}
            }
        {\lsequent{}{\Gamma_n \rightarrow \dbox{\gamma_{n + 1}}{(x_0 > 0 \land x \in B[x_0, \alpha(n + 1)x_0])}}}
    \end{sequentdeduction}
    Where the open premise \textcircled{2} arising from \irref{cut} is
    \[\lsequent{}{ \forall x \in B[x_0, \alpha n x_0] \forall y \in B[x, \alpha x_0]\norm{y^2 + 1} \leq (\alpha (n + 1) + 1)^2x_0^2 + 1}\]
    which is valid as $y^2 + 1$ is maximized when $\abs{y}$ is maximized, therefore the maximum is attained when $y = x_0 + \alpha n x_0 + \alpha x_0 = (\alpha (n + 1) + 1)x_0$ and the premise is proven by axiom \irref{qear}. We can now derive the example.
    
    \begin{sequentdeduction}
        \linfer[dualRight]
            {\linfer[dC+dW]
                {\linfer
                    {\text{\textcircled{3}}}
                    {\lsequent{x = x_0, t = 0, x_0 > 0}{\dbox{\gamma_{N}}{x \in B[x_0, \alpha N x_0]}}}
                    &
                    \linfer[qear]
                    {\lclose}
                    {\lsequent{t \leq (1 - \eps)\left(\frac{1}{x_0} - \frac{1}{3x_0^3}\right)}{t \leq t_{N}}}
                }
                {\lsequent{x = x_0, t = 0, x_0 > 0}{\dbox{x' = x^2 + 1, t' = 1 \& t \leq (1 - \eps)\left(\frac{1}{x_0} - \frac{1}{3x_0^3}\right)}{x \in B[x_0, \alpha N  x_0]}}}
            }
        {\lsequent{x = x_0, t = 0, x_0 > 0, \frac{1}{x_0} - \frac{1}{3x_0^3} \geq 0}{\ddiamond{x' = x^2 + 1, t' = 1}{t \geq (1 - \eps)\left(\frac{1}{x_0} - \frac{1}{3x_0^3}\right)}}}
    \end{sequentdeduction}
    The resolution of the right premise with axiom \irref{qear} requires justification, it is not trivial that the inequality $t_{N} \geq (1 - \eps)\left(\frac{1}{x_0} - \frac{1}{3x_0^3}\right)$ holds. Lower-bounding $t_N$ with the corresponding integral yields:
    \begin{align*}
        t_N = \sum_{n = 1}^N \frac{\alpha x_0}{(\alpha n + 1)^2x_0^2 + 1} \geq \int_{1}^{N} \frac{\alpha x_0}{(\alpha t + 1)^2x_0^2 + 1} dt = \arctan((\alpha N + 1)x_0) - \arctan((\alpha + 1)x_0)
    \end{align*}
    It is well-known that the bound 
    \[\frac{\pi}{2} - \frac{1}{x} \leq \arctan(x) \leq \frac{\pi}{2} - \frac{1}{x} + \frac{1}{3x^3}\]
    holds for all $x > 0$ (can be derived from standard Taylor bounds of $\arctan(x)$ and the identity $\arctan(x) + \arctan\left(\frac{1}{x}\right) = \frac{\pi}{2}$). Utilizing this, we have
    \begin{align*}
        t_N &\geq \arctan((\alpha N + 1)x_0) - \arctan((\alpha + 1)x_0) \\
        &\geq \arctan((\alpha N + 1)x_0) - \frac{\pi}{2} + \frac{1}{(\alpha + 1)x_0} - \frac{1}{3(\alpha + 1)^3x_0^3}\\
        &\geq \frac{\pi}{2} - \frac{1}{(\alpha N + 1)x_0} - \frac{\pi}{2} + \frac{1}{(\alpha + 1)x_0} - \frac{1}{3(\alpha + 1)^3x_0^3}\\
        &= \frac{1}{x_0}\left(\frac{1}{\alpha + 1} - \frac{1}{\alpha N + 1}\right) - \frac{1}{3x_0^3} \left(\frac{1}{(\alpha + 1)^3}\right)\\
        &\geq \left(\frac{1}{\alpha + 1} - \frac{1}{\alpha N + 1}\right)\left(\frac{1}{x_0} - \frac{1}{3x_0^3}\right)
    \end{align*}
    where the final inequality follows from the assumption that $\frac{1}{\alpha + 1} - \frac{1}{\alpha N + 1} \geq \frac{1}{(\alpha + 1)^3}$. Finally, since $\frac{1}{\alpha + 1} - \frac{1}{\alpha N + 1} \geq 1 - \eps$ by construction, the desired bound holds and the application of axiom \irref{qear} is justified. At last, the derivation of \textcircled{3} can be completed by iteratively applying axioms \irref{stepEx+stepExt}, note that by construction $\lsequent{t = t_n, x_0 > 0, x \in B[x_0, \alpha n x_0]}{\Gamma_{n}}$ is always valid. 
    
    \begin{sequentdeduction}
        \linfer[cut+stepEx]
            {\linfer[implyl]
                {\linfer[stepExt+implyl]
                    {\linfer[stepExt+implyl]
                        {\linfer[stepExt+implyl]
                            {\linfer[id]
                                {\lclose}
                                {\lsequent{x_0 > 0, t = 0, \dbox{\gamma_N}{x \in B[x_0, \alpha N x_0]}}{\dbox{\gamma_{N}}{x \in B[x_0, \alpha N x_0]}}}
                            }
                            {\cdots}
                        }
                        {\lsequent{x_0 > 0, t = 0, \dbox{\gamma_2}{x \in B[x_0, 2\alpha x_0]}}{\dbox{\gamma_{N}}{x \in B[x_0, \alpha N x_0]}}}
                        &
                        \linfer[qear]
                        {\lclose}
                        {\lsequent{t = t_1, x_0 > 0, x \in B[x_0, \alpha x_0]}{\Gamma_1}}
                        &
                        \linfer[stepEx]
                        {\lclose}
                        {\phi_1}
                    }
                    {\lsequent{x_0 > 0, t = 0, \dbox{\gamma_1}{(x_0 > 0 \land x \in B[x_0, \alpha N x_0])}}{\dbox{\gamma_{N}}{x \in B[x_0, \alpha N x_0]}}}
                }
                {\lsequent{x = x_0, t = 0, x_0 > 0, \phi_0}{\dbox{\gamma_{N}}{x \in B[x_0, \alpha N x_0]}}}
            }
        {\lsequent{x = x_0, t = 0, x_0 > 0}{\dbox{\gamma_{N}}{x \in B[x_0, \alpha N x_0]}}}
    \end{sequentdeduction}
Where the abbreviated derivation consists of $N$ levels, at the $n$-th level $\dbox{\gamma_n}{x \in B[x_0, \alpha n x_0]}$ is proven via applications of \irref{stepEx+stepExt}, this completes the derivation of the example. Note that when choosing the parameters $\alpha, N$, all sufficiently small $\alpha$ and all sufficiently large $N$ will suffice. As an example, suppose that the desired error threshold is $1 - 0.1 = 0.9$ with $\eps = 0.1$, then $\alpha = 0.01$ and $N = 10^4$ works. 
\end{proof}

\bibliographystyle{ACM-Reference-Format}
\bibliography{references}

\end{document}